\documentclass[12pt,paper]{article}
\usepackage{amsmath,amssymb,amsthm,amsfonts,fullpage}
\usepackage{enumitem}  
\usepackage[
backend=bibtex,
style=numeric,
sorting=nyt
]{biblatex}
\usepackage{authblk}

\usepackage[hidelinks,pdfauthor={Abdelghaffar Chibloun}]{hyperref}
\hypersetup{
	colorlinks   = true,
	citecolor    = blue
}
\usepackage{relsize}
\usepackage{cleveref}
\usepackage{comment}
\usepackage{color,xcolor,soul}
\newcommand{\hlc}[2][yellow]{ {\sethlcolor{#1} \hl{#2}} }

\soulregister\cite7
\soulregister\ref7
\providecommand{\keywords}[1]
{
  {\small	
  \textbf{{Keywords--}} #1}
}

\usepackage{comment}
\theoremstyle{plain}
\newtheorem{theorem}{Theorem}[section] 
    \newtheorem{lemma}[theorem]{Lemma} 
    \newtheorem{proposition}[theorem]{Proposition}
    \newtheorem{corollary}[theorem]{Corollary}
    \newtheorem{definition}[theorem]{Definition}
    \newtheorem{example}[theorem]{Example}
    \newtheorem{remark}[theorem]{Remark}

\def\dsum #1#2{\displaystyle{\sum_{#1}^{#2}}}

\newcommand{\Fq}{\mathbb{F}_{q}}

\DeclareMathOperator{\Z}{\mathcal{Z}}
\DeclareMathOperator{\rank}{rank}

\DeclareMathOperator{\ord}{ord}

\DeclareMathOperator{\aut}{Aut}

\DeclareMathOperator{\GL}{GL_{n}}
\DeclareMathOperator{\lcm}{lcm}
\DeclareMathOperator{\lclm}{lclm}        
\DeclareMathOperator{\gcrd}{gcrd}

\DeclareMathOperator{\id}{Id}

\title{\textbf{Bounds and Equivalence of Skew Polycyclic Codes over Finite Fields}}
 \author[1]{Hassan  Ou-azzou\thanks{E-mail: \rm{hassan.ouazzou@student.unisg.ch}. Supported by a Swiss Government Excellence Scholarship (ESKAS), no:  2024.0504.}}
  \author[2]{Anna-Lena Horlemann\thanks{E-mail: \rm{anna-lena.horlemann@unisg.ch}}}
   \author[3]{Nuh Aydin\thanks{E-mail: \rm{aydinn@kenyon.edu}}}

\affil[1,2]{School of Computer Science, University of St.Gallen, 9000 St.\ Gallen, Switzerland}
\affil[3]{Department of Mathematics and Statistics, Kenyon College, Gambier, OH 43022, USA}

\date{ \today}
\begin{document}

	\maketitle
    
    \begin{abstract}

    We study skew polycyclic codes over a finite field $\mathbb{F}_q$, associated with a skew polynomial $f(x) \in \mathbb{F}_q[x;\sigma]$, where $\sigma$ is an automorphism of $\mathbb{F}_q$.  We establish Roos-like bounds on the minimum Hamming and rank distances for these codes, generalizing several known distance bounds for various (consta)cyclic and skew (consta)cyclic code families. To construct codes with a prescribed minimum distance, we develop an approach based on $\mu$-closed sets, which serve as a non-commutative analog of cyclotomic cosets in the commutative case. We further study the equivalence of ambient spaces for skew polycyclic codes under both the Hamming and rank metrics. General conditions are derived under which two such ambient spaces yield equivalent code families, allowing a complete classification of equivalence classes. In the special case of skew trinomial codes, we determine when a code is equivalent to one defined by a trinomial of the form $x^n-x^{\ell}-1,$ and we describe a general procedure to determine the equivalence class of any skew polycyclic code. Several examples are provided to illustrate the application of the theoretical results on code equivalence.

\end{abstract}

\section{Introduction}
Coding theory plays a fundamental role in various applications, such as error detection and correction, data transmission, data storage and reliable communication. It involves the study of efficient encoding and decoding methods for transmitting data reliably over noisy channels. {Cyclic codes} are one of the most important families of linear codes for both theoretical and practical reasons. They establish a key link between coding theory and algebra and their structure often makes them convenient for implementations. Furthermore, many of the most important classes of codes are cyclic or related to cyclic codes. Cyclic codes were introduced in the 1950's by Prange in \cite{Prange1957} as linear codes with the property that the cyclic shift of any codeword is another codeword. They were later generalized to {constacyclic codes} in \cite{Berlekamp1968} and to {polycyclic codes} (also known as {pseudo-cyclic codes}) in \cite{10}.  Like cyclic (and constacyclic) codes, each polycyclic code over a finite field $\Fq$ can be described by an ideal of the polynomial ring $\Fq[x]/\langle f(x)\rangle,$ where $f(x)$ is a nonzero polynomial in $\Fq[x]$.  Polycyclic codes contain constacyclic codes as a special case when $f(x)=x^n-\lambda,$ for some non-zero $\lambda$ in $\Fq.$

One of the generalizations of cyclic codes are skew cyclic codes (\cite{Boucher2009,Boucher2014}). They can be viewed as left $\Fq[x;\sigma]$-submodules of $\Fq[x;\sigma] / \langle x^n-1\rangle,$ where $\sigma$ is an automorphism on $\mathbb{F}_q$, and $\langle x^n-1\rangle$ is the left ideal of $\Fq[x;\sigma]$ generated by $ x^n-1.$ Similarly to classical cyclic codes, skew cyclic codes are generalized in various ways, such as skew constacyclic codes, skew polycyclic  codes, and others. We refer to the following references for more details \cite{Almendras2018, Bag2025, Boulanouar2021, Ouazzou2023}.

In \cite{Chen2014}, an equivalence relation was introduced to study the equivalence of classes of constacyclic codes of length $n$ over $\mathbb{F}_q$, depending on the ideal defining the quotient ring. This was generalized to skew constacyclic codes over $\Fq$ \cite{Ouazzou2025}, and to the case of constacyclic codes over finite chain rings \cite{Chibloun2024}. In \cite{Boulanouar2021,Lobillo2025} some characterizations on equivalence of skew constacyclic codes were given. Regarding polycyclic codes,  several properties of trinomial codes and a number of conjectures related to their equivalence and duality are presented in \cite{Aydin2022, Shi2023}. The study of polycyclic codes was continued by extending the notion of $n$-equivalence to the case of polycyclic codes in \cite{Equiv2025} 
where the authors studied the properties of this equivalence relation, computed the number of $n$-equivalence classes, and provided conditions under which two ambient spaces of polycyclic (or $\ell$-trinomial) codes are equivalent. 

In this paper,  we study skew polycyclic codes over a finite field $\mathbb{F}_q$, associated with a skew polynomial $f(x) \in \mathbb{F}_q[x;\sigma]$, where $\sigma$ is an automorphism of $\mathbb{F}_q$. We start  by  proving  the Roos-like bound for both the Hamming and the rank metric for this class of codes. Next, we  focus on the Hamming and rank equivalence between two classes of polycyclic codes by introducing an equivalence relation and describing its equivalence classes. 
    Finally, we present examples that illustrate applications of the theory developed in this paper.

The remainder of this paper is organized as follows. Section 2 provides a review of the basic background on skew polynomial rings and skew polycyclic codes, plus the necessary background on the Hamming and the rank metric.  In Section 3, we prove the Roos-like bound for both metrics for this class of codes and study constructions of such codes with a designed distance. In Section 4 we study the equivalence of ambient spaces for skew polycyclic codes and provide conditions under which two skew polycyclic code families are equivalent -- again for both the Hamming and the rank metric. Finally, in Section 5, we provide examples for code equivalences studied in section 4.

\section{Preliminaries}
In this section we recall the necessary background about error correcting codes in the Hamming and rank metric, skew polynomials, and skew polycyclic codes.

Throughout we will denote a finite field with $q$ elements (for $q$ being a prime power) by $\Fq$. Similarly an extension field of $\Fq$ of degree $m$ is denoted by $\mathbb F_{q^m}$.  The general linear group $\GL(\mathbb{F}_{q}) $ is the group of invertible $n\times n$ matrices with entries from $\mathbb{F}_{q}$.

\subsection{Hamming and rank metric codes}

We start with definitions and results on error correcting codes in the Hamming and the rank distance. For more details  we refer the reader to \cite{augot2018generalized,gabidulin2025rank,Gabidulin1985,Huffman2003} 

\begin{definition}
      A $q$-ary \emph{code} of length $n$ is simply a subset of $\Fq^n$. A \emph{linear code} is a vector subspace of $\Fq^n$.
  \end{definition}

  \begin{definition}[Hamming distance]
  Let  $\Fq$  be a finite field.
  \begin{enumerate}
       \item The \emph{Hamming weight} of a vector $x=(x_0,x_1,\ldots, x_{n-1})\in \Fq^n$,  denoted by $  w_{H}(x),$ is the number of  nonzero coordinates of $x.$
       \item The \emph{Hamming distance} $d_H(x,y)$ between two vectors $x,y \in \Fq^n$ is defined as 
       $$ d_H(x,y):=w_H(x-y) .$$
\item  The   \emph{minimum Hamming distance} of a linear code $C\subseteq \Fq^n$, denoted $ d_H(C)$, is defined as the minimum Hamming weight $w_H(c)$ over all nonzero codewords $c \in C$, i.e.,
$$
d_{H}(C):=\min \left\{w_H(x): x \in C, x \neq 0\right\} .
$$
  \end{enumerate}
  \end{definition}

The Hamming distance satisfies the following bound, called the \emph{Singleton bound} (see, e.g., \cite{Huffman2003}): 
\begin{theorem}
Let $C\subseteq \Fq^n$ be a linear code of dimension $k$. Then
    $$ d_H(C)\leq n-k+1.$$
When equality is obtained, we say that $C$ is a  \emph{Maximum Distance Separable (MDS)} code.
\end{theorem}

  While codes in the Hamming metric are defined as sets of vectors, rank metric codes can be defined as subsets of the ring of $m\times n$ matrices over $\mathbb F_q$, $ \mathbb M_{m\times n}(\Fq).$ Such a code is called \emph{linear}\footnote{The notion for linearity differs in the literature, often depending on the largest field over which a code is linear.}, if it is a linear  subspace of $
\mathbb M_{m\times n}(\Fq).$ 
Via a vector space isomorphism $\mathbb M_{m\times n}(\Fq) \cong \mathbb{F}_{q^m}^n $ rank metric codes can also be represented as subsets of $\mathbb{F}_{q^m}^n$. A linear code in the above sense is then an $\Fq$-linear subspace of $\mathbb{F}_{q^m}^n$.



\begin{definition}[Rank distance]\label{rank_weight}
Consider an extension field $\mathbb{F}_{q^m}$.
\begin{enumerate}
\item  The \emph{rank weight} of $ x=(x_0, \ldots, x_{n-1}) \in \mathbb{F}_{q^m}^n$ over $ \mathbb{F}_{q},$ denoted $w_{r}(x)$, is the dimension of the $ \mathbb{F}_{q}$-vector space spanned by $x_0, \ldots, x_{n-1},$ i.e., $$ w_{r}(x):=\dim_{\mathbb{F}_{q}}(\langle x_0,\ldots,x_{n-1}\rangle).$$ 
    \item The relation  $d_R(x, y):=w_r(x-y)$ for $x, y \in \mathbb F_{q^m}^n$ defines a distance on $\mathbb F_{q^m}^n$ called the \emph{rank distance.}
    \item  The \emph{minimum rank distance} $d_R(C)$ of a linear code $C\subseteq \mathbb F_{q^m}^n$ is defined as 
$$ d_R(C):=\min \left\{ w_r(x) \ : \ x \in C, x \neq 0\right\} .$$
\end{enumerate}
\end{definition}

The rank distance $d_R$ of $C$ satisfies the following bound, called the \emph{Singleton-like (or rank Singleton) bound} (see, e.g., \cite{gabidulin2025rank}):
\begin{theorem}
Let $C\subseteq \mathbb{F}_{q^m}^n$ be a linear code. Then
    $$
\dim_{\mathbb{F}_{q}}(C) \leq \max \{n, m\}(\min \{n, m \}-d_R(C)+1) .
$$
 When equality in the Singleton-like bound occurs, we say  that $C$ is a  \emph{Maximum Rank Distance (MRD)} code.
\end{theorem}

In the following proposition from \cite[Proposition 7]{Alfarano2021}, we recall a connection between the minimum rank distance and the minimum Hamming distance of a code.
\begin{proposition}\label{Prop_Hamming_rank}
\cite[Proposition 7]{Alfarano2021}
 Let $C$ be a linear code of length $n$ over $ \mathbb{F}_{q^m}.$ Then
$$ d_R(C)=\min \left\{d_H(C \cdot M): M \in \GL(\mathbb{F}_{q})\right\} .$$
\end{proposition}


\vspace{0.2cm}

The notion of equivalence is usually defined via linear or semi-linear isometries of the ambient metric space. For more information on the Hamming and rank isometries we refer the interested reader to \cite{Berger2003Isometries,Huffman2003,Neri2020}. In our work, we will restrict ourselves to linear isometries and hence get the following two types of equivalence:

\begin{definition} Let $C_1$ and $C_2$ be two linear codes of length $n$ over $\mathbb{F}_{q^m}$. Then 
\begin{enumerate}
    \item   $C_1$ and $C_2$ are called \emph{Hamming equivalent} if there exists a $n \times n$ monomial matrix $P$ with entries in $ \mathbb{F}_{q^m}$ such that $C_1=C_2 P$.
    \item  $C_1, C_2 $ are \emph{rank equivalent}  if there exists an invertible  matrix $P\in \GL(\mathbb{F}_{q}),$ and a scalar $\lambda \in \mathbb{F}_{q^m}^* $ such that $C_1= \lambda C_2 P.$
\end{enumerate}
\end{definition} 

\subsection{Skew polynomial rings}
In this subsection we give necessary background we need in skew polynomials, for more details see \cite{Ore33,Lam1988,Lam2004-I,Leroy2004,Leroy2008}. Let $\Fq$ be the finite field with $q$ elements, where $q=p^{s}$ for some prime number $p$ and $s \in \mathbb{Z}_{\geq 1}$.  Let $\sigma : \Fq \rightarrow \Fq$ be an automorphism of $\Fq$. Recall that the  \emph{Frobenius  automorphism} of  $\Fq$   is the automorphism $ \tau : \Fq \rightarrow \Fq$ defined by $$\tau(a):=a^{p}.$$ 
 Note that $\tau $ generates the cyclic group of automorphisms of $  \Fq,$ denoted by $\aut(\Fq)$, hence each automorphism $\sigma$ is of the form $ \sigma=\tau^k, \ 0\leq k\leq  s-1.$   Since $\tau $ is a generator of the cyclic group $ \aut(\Fq)$ its order is $s:=[\Fq:\mathbb{F}_p],$ and other generators of $\aut(\Fq)$ are of the form $ \sigma = \tau^{k},$ where $ \gcd(s,k)=1.$  
 
 Recall also that an element $a\in \Fq$ is \emph{fixed} by $\sigma$ if $\sigma(a)=a,$ and the set of fixed elements by $\sigma,$ denoted by $\Fq^{\sigma},$ is a subfield of $\Fq.$   Note that if  $\sigma=\tau^r,$  then its \emph{fixed subfield} is  $\Fq^{\sigma}= \mathbb{F}_{p^{\gcd(r,s)}}.$ Hence the fixed subfield of $\Fq$ by any generator of  $\aut(\Fq)$ is $\mathbb{F}_p.$ 

 Throughout  this paper we fix the following notations: we consider an automorphism  $\sigma \in \aut(\Fq),$ such that $$\sigma(a)=a^{p^r},$$ of order $ \mu:=\frac{s}{\gcd(r,s)},$ and its fixed subfield $\Fq^{\sigma}=\mathbb{F}_{q_0}$, for $\ q_0:=p^{\gcd(s,r)}.$

\begin{definition}
The \emph{$\sigma$-skew polynomial ring} $\Fq [x;\sigma]$ is the set 
$$ \Fq [x;\sigma]:=\left\{a_{0}+a_{1} x+\cdots+a_{n-1} x^{n-1}: a_{i} \in \Fq, n \in \mathbb{N}\right\}$$
endowed with  the usual polynomial addition and   multiplication defined by 
$$ xa=  \sigma(a) x ,$$
extended to all polynomials by distributivity. 
\end{definition}

Recall that an element $f\in \Fq[x;\sigma]$ is called   \emph{central} if $f(x) g(x) = g(x) f(x)$ for all $f(x) \in \Fq[x;\sigma]$, i.e., $f$ commutes with every element of $\Fq[x;\sigma]$. The set of all elements of $\Fq[x;\sigma]$ that commute with every element of $\Fq[x;\sigma]$ is called the  \emph{center} of $\Fq[x;\sigma]$ which is defined and  denoted by
\begin{center}
 $ \Z(\Fq[x;\sigma]):=\left\lbrace f(x) \in \Fq[x;\sigma] \ : \  f(x) g(x) = g(x) f(x), \ \text{ for all} \ \ g(x) \in \Fq[x;\sigma] \right\rbrace .$
\end{center}

In the following we collect some properties of  the  skew polynomial ring $\Fq[x;\sigma]$, see also \cite[p. 483-486]{Ore33}.

\begin{proposition} \label{PPP.1}
\begin{enumerate}
\item The ring $\Fq[x;\sigma]$  is non-commutative unless $\sigma $ is the identity automorphism of $\Fq.$ 
\item   An element $f \in \Fq[x;\sigma]$ is central if and only if $f$ belongs to $ \Fq^{\sigma}[x^{\mu}]$ i.e., $  \Z(\Fq[x;\sigma])= \Fq^{\sigma}[x^{\mu}],$ where  $ \Fq^{\sigma}$ is the fixed subfield of $\Fq$ by $\sigma$ and $\mu$ is the order  of $\sigma.$
\item The ring $\Fq[x;\sigma]$ is a right (respectively left)   Euclidean domain. By  "$x \mid_r y$"   we denote that $x$ right divides $y$.
\item Left and right ideals of $\Fq[x;\sigma]$ are principal.
\end{enumerate}
\end{proposition}


Since $\Fq[x,\sigma]$ is a right (left) Euclidean domain, we have the notion of the greatest common right divisor of polynomials ($\gcrd$) in $\Fq[x,\sigma]$ and  least common left multiple ($\lclm$).\footnote{The analog holds for greatest common left divisor and least common right multiple.} Note that they are unique and can be calculated using the known left and right division algorithms, for more details see \cite{Ore33}.

\begin{definition}
 Let $f(x), g(x) \in \Fq[x;\sigma]$. 
 \begin{enumerate}
     \item A monic polynomial $d(x)$ is called the \emph{greatest common right divisor} ($ \gcrd $) of $f(x)$ and $g(x)$ if  $d(x)$ is a right divisor of $f(x)$ and $g(x)$, and any right divisor  $d^{'}(x)$ of $f(x)$ and $g(x)$ is also a  right divisor of $d(x)$.
\item In particular, $f(x)$ and $ g(x) $ are said to be \emph{right coprime} if $\gcrd(f(x), g(x))=1$. 
In this case, there exist polynomials  $a(x), b(x) \in  \Fq[x;\sigma]$ such that $a(x) f(x)+b(x) g(x)=1.$
 \item Similarly, the \emph{least common left multiple} $(\lclm)$ of $f(x)$ and $g(x)$   is defined as the monic polynomial  $h(x)\in  \Fq[x;\sigma]$  of least degree such that $f(x)$ and $ g(x)$ are both  right  divisors of $h(x).$
\end{enumerate}
\end{definition}

\begin{definition}[\cite{Ore33}] Let $a(x), b(x) \in \Fq[x,\sigma]$ two skew polynomials. We say  $a(x)$ is \emph{similar} to $b(x)$ if there exists $u(x) \in \Fq[x,\sigma] $  such that $\lclm(a(x), u(x)) = b(x) u(x)$  and $\gcrd(a(x), u(x)) = 1.$
\end{definition}

\begin{theorem} \cite[Theorem 1]{Ore33}  \label{Factorization}
Every non-zero polynomial $f(x) \in \Fq[x;\sigma]$ factorizes as $f(x)=f_1(x) \cdots f_n(x)$ where $f_i(x) \in \Fq[x;\sigma]$ is irreducible for all $i \in\{1, \ldots, n\}$. Furthermore, if $f(x)=g_1(x) \cdots g_s(x)$ is any other factorization of $f(x)$ as a product of irreducible polynomials $g_i(x) \in \Fq[x;\sigma]$, then $s=n$ and there exists a permutation $\pi:\{1, \ldots, n\} \rightarrow$ $\{1, \ldots, n\}$ such that $f_i$ is similar to $g_{\pi(i)}$. In particular, $f_i$ and $g_{\pi(i)}$ have the same degree for all $i \in\{1, \ldots, n\}$.
\end{theorem}

It is important to note that in the skew polynomial ring $\Fq[x;\sigma]$, there are different polynomial evaluations, out of which we will use the following. 

\begin{definition} \cite[p. 310]{Lam1988}\label{def.2}
For $\alpha \in \Fq,$ and $f(x)\in \Fq[x;\sigma]$,
  the \emph{(right) evaluation} (also called  the \emph{(right) remainder evaluation})  of $f(x)$ at $\alpha$ is the element $f(\alpha ):=r \in \Fq$ such that $f(x)= q(x)(x-\alpha)+r.$
\end{definition}

In \cite[Proposition 2.9]{Lam1988}, Lam and Leroy  proved that, for any $\alpha\in \Fq$ we have   
\begin{equation}\label{e1}
f(\alpha)= a_0 + a_1 N^{\sigma}_1(\alpha) +\ldots+ a_{n-1} N^{\sigma}_{n-1}(\alpha)=  \dsum{i=0}{n} a_i N^{\sigma}_i(\alpha)
\end{equation}
where $  N^{\sigma}_{0}(\alpha):=1 ,$ and  for  every $i$ in $\mathbb{N}^*$  
$$
N^{\sigma}_i(\alpha):=\sigma^{i-1}(\alpha) \ldots \sigma(\alpha) \alpha.
$$
By taking  $\sigma(a)=a^{p^r} $ for all $a\in \Fq,$ we obtain  for each $ i \in \mathbb{N}^*$  that
\begin{equation}\label{Eqsigma}
N_i^{\sigma}(\alpha)=\alpha^{p^{r (i-1)}} \alpha^{p^{r (i-2)}} \ldots  \alpha^{p^{r}} \alpha =  \alpha^{ \sum_{j=0}^{i-1} p^{rj}}= \alpha^{{\frac{ p^{ri}-1}{p^r-1}}}= \alpha^{[i]_r},
\end{equation}
where $[i]_{r}:=  \dfrac{ p^{ri}-1}{p^r-1}$. Then the  right evaluation of $f(x)=\dsum{i=0}{n} a_i x^i$ at $\alpha\in \Fq$ is given by 
\begin{equation*}
f(\alpha)= \dsum{i=0}{n} a_i N^{\sigma}_i(\alpha)= \dsum{i=0}{n} a_i \alpha^{[i]_r}.
\end{equation*}

A skew polynomial $f(x) \in \Fq[x;\sigma]$ is called an \emph{invariant polynomial} if $f(x)\cdot \Fq[x;\sigma] = \Fq[x;\sigma]\cdot f(x),$ i.e., there exist $g(x)$ and $ g'(x) $ such that $ f(x)g(x)= g'(x)f(x).$ In particular,   each polynomial in  $ \Z(\Fq[x;\sigma])=\Fq^{\sigma}[x^{\mu}]$ is   invariant, where $\mu$ is the order of $\sigma$. Note that when $f$ is an invariant polynomial in $\Fq[x;\sigma]$  then $R_f:=\Fq[x;\sigma]/\langle f(x)\rangle$ is an associative algebra over  $ \Fq^{\sigma} $ of dimension $ \deg(f) \mu.$   If $f$ is not invariant, then the center of $ R_f$ is $ \Fq.$

\begin{remark}
There is a well-known connection between skew polynomials and linearized polynomials (precisely $\sigma$-polynomials), and the latter have been studied extensively in the context of rank-metric codes, see, e.g.,
\cite{Gabidulin1985,ko08,ore1933special,wu2013linearized}. 
For this let $\sigma$ be an automorphism of order $\mu,$ and $\mathbb{F}_{q_0}=\Fq^{\sigma}$ its fixed subfield.
The (non-commutative) ring of   $\sigma$-polynomials, $ (\mathcal{L}_{q,\sigma}[x],+,\circ),$  is  the subset  of the polynomial ring $\Fq[x]$ defined by
$$
\mathcal{L}_{q,\sigma}[x]:=\left\{\sum_{i=0}^n a_i x^{\sigma^i} := \sum_{i=0}^n a_i x^{q_0^{i}}\mid n \in \mathbb{N}, a_i \in \Fq \right\},
$$
endowed with  the usual addition $"+"$ and  composition $"\circ"$ of polynomials. Note that  $\Fq[x ; \sigma]$ and $\mathcal{L}_{q,\sigma}[x]$ are
isomorphic rings 
via 
$$
\Lambda: \Fq[x ; \sigma] \longrightarrow \mathcal{L}_{q,\sigma}[x], \sum_{i=0}^n a_i x^i \longmapsto \sum_{i=0}^n a_i x^{\sigma^i}=\sum_{i=0}^n a_i x^{q_0^{i}} .
$$
Analogously, we get the following isomorphism 
$$ \Fq[x ; \sigma]/\langle x^{\mu}-1\rangle \cong \mathcal{L}_{q,\sigma}[x]/\langle x^{q_0^{\mu}}-x\rangle .$$ 
The above correspondence allows us to characterize skew cyclic codes of length $\mu $ over $ \Fq$ as left ideals of $\mathcal{L}_{q,\sigma}[x]/\langle x^{q_0^{\mu}}-x\rangle,$ for more details see \cite{gabidulin2009rank} and more recently \cite[Section 2.4]{martinez2017roots}. More generally, we can prove 
for any  skew polynomial $f \in \mathbb{F}_{q_0}[x,\sigma]$ of degree $ \mu$  ($f$ is defined over $\mathbb F_{q_0}$ to generate a two sided ideal in $\Fq[x;\sigma])$ the following   ring isomorphism 
$$ \Fq[x ; \sigma]/\langle f(x)\rangle \cong \mathcal{L}_{q,\sigma}[x]/\langle f(x^{q_0}) \rangle.$$ 
Note that, seen as an algebra isomorphism, this relation holds over $ \mathbb{F}_{q_0}.$

\end{remark}

\vspace{0.5cm}

Next we recall some important facts about algebraic sets and the class of \emph{Wedderburn (\textnormal{or} W-)polynomials}. For more details, we refer the reader to \cite{Leroy2004, Leroy2008}. 

\begin{definition}
\begin{enumerate}
    \item 
    For a skew polynomial $g \in \mathbb{F}_q[x, \sigma]$, the \emph{right vanishing set} of $g$, also called the \emph{set of right roots} of $g$, is defined by  
$
V(g) := \left\{ a \in \mathbb{F}_q : g(a) = 0 \right\}.
$
\item
For any subset $A \subseteq \mathbb{F}_q$, the set $I(A) := \{ g \in \mathbb{F}_q[x, \sigma] : g(A) = 0 \}$ is a left ideal of $\mathbb{F}_q[x, \sigma]$, and the set $A$ is called \emph{$\sigma$-algebraic} if $I(A) \neq 0$. In particular, the set $V(g)$ is $\sigma$-algebraic because $g(V(g)) = 0$. Moreover, any finite set is $\sigma$-algebraic \cite[Section 6]{Lam2000}.
\item If $A$ is $\sigma$-algebraic, then the monic generator of $I(A)$ is called the \emph{minimal polynomial} of $A$, and is denoted by $m_A$. 
\end{enumerate}
\end{definition}

The \emph{minimal polynomial} $m_A$ is the monic least left common multiple of the linear polynomials $\{x - a : a \in A\}$. Thus, it has the form  
$
m_A = (x - a_0)(x - a_1) \cdots (x - a_{n-1}),
$
where each $a_i$ is $\sigma$-conjugate to some element of $A$, and $n = \deg m_A$. Furthermore, any right root of $m_A$ is also $\sigma$-conjugate to some element of $A$.
The \emph{rank} of a set $A \subseteq \mathbb{F}_q$ is defined by $\operatorname{rank}(A) := \deg(m_A)$.

\begin{definition}
\begin{enumerate}
    \item 
A monic polynomial $g \in \mathbb{F}_q[x, \sigma]$ is called a \emph{W-polynomial} if the minimal polynomial of $V(g) = \{ a \in \mathbb{F}_q : g(a) = 0 \}$ is $g$ itself. 
\item 
An element $a \in \mathbb{F}_q$ is called \emph{$P$-dependent} on an algebraic set $A\subseteq \Fq$ if $g(a) = 0$ for every $g \in I(A)$. A $\sigma$-algebraic set $A$ is \emph{$P$-independent} if no element $b \in A$ is $P$-dependent on $A \setminus \{b\}$. 
\end{enumerate}    
Note that the minimal polynomial of a $P$-independent set is a W-polynomial.
\end{definition}

Lastly, we define \emph{exponents} of skew polynomials corresponding to \cite[Definition 3.1.]{Cherchem2016}, which will be crucial to prove the Roos-like bound on the Hamming and rank distance in the next section. (For the definition of exponents of skew polynomials over period rings, see \cite{bouzidi2021exponents}.)

\begin{definition}\cite[Definition 3.1.]{Cherchem2016}
    Let $ f(x) \in \Fq[x;\sigma], $ be a polynomial of degree $n$ with $ f(0)\neq 0 $.  Then there exists  a positive integer $ e $ such that $ f(x) $ right divides $ x^{e}-1.$  The smallest integer $e$ with this property is called the \emph{right exponent}  of $f(x),$ and we denote it by $ \ord_r(f(x))=e.$ 
    If $f(0)=0$, we write $f(x)=x^h g(x)$, where $h \in \mathbb{N}$ and $g \in \mathbb{F}_q[x ; \sigma]$ with $g(0) \neq 0$ are uniquely determined, and define $ \ord_r(f(x)):=\ord_r(g(x))$. The \emph{left exponent} is defined analogously.
\end{definition}

\subsection{Skew polycyclic codes}

We will now turn to polycyclic codes over skew polynomial rings, defining some of the necessary notions for the rest of this paper. We begin with the definition of skew polycyclic codes, based on \cite[Definition 8]{Ouazzou2023}.

  \begin{definition}
  Let $ \vec{a}= (a_0,a_1,\ldots, a_{n-1})\in \Fq^n$   and  $\sigma \in \aut(\Fq)$  be an automorphism of $\Fq.$  A linear code $C\subseteq \Fq^n$ is called a
  \begin{enumerate}
      \item  right \emph{skew  $(\sigma,\vec{a})$-polycyclic   code}   if for each codeword  $(c_0,c_1,\ldots, c_{n-1})$ of $C,$  we have 
$$
\left( 0, \sigma(c_0),\sigma(c_1),\ldots, \sigma(c_{n-2})\right) + \sigma(c_{n-1}) \vec{a}  \in C.
$$
\item left \emph{skew  $(\sigma,\vec{a})$-polycyclic   code}   if for each codeword  $(c_0,c_1,\ldots, c_{n-1})$ of $C,$  we have 
$$
\left(\sigma(c_1),\ldots, \sigma(c_{n-1},0)\right) + \sigma(c_{0}) \vec{a}  \in C.
$$
 \end{enumerate}
    If $\vec{a}=\left(a_{_0},0,\ldots, 0,a_{\ell}, 0, \ldots, 0\right)\in \Fq^{n},$ with $a_0\neq 0 $ and $ a_{\ell}\neq 0$, for some $ 0<\ell <n $, we call a right (resp. left)    $(\sigma,\vec{a})$-polycyclic code with  associated vector $\vec{a}$  a \emph{skew $(\ell,\sigma)$-trinomial code}  of length $n$ over $\Fq$.
 \end{definition}
 In this paper we will focus on right skew polycyclic codes, however, the same result can be proved for the case of left polycyclic codes. For more details about the difference between left and right (skew) polycyclic codes see \cite{Lopez2009,Poly,ou2024linear,Ouazzou2023, Hassan2025Advances}.
 
 Let us recall that a map $T: \Fq^n \to \Fq^n,$ is called a \emph{($\sigma$-)semi-linear transformation} if it is an additive map satisfying $ T(\alpha u)=\sigma( \alpha ) T(u)$ for $\alpha \in \Fq$ and $u\in \Fq^n$. In the particular case $\sigma=\id, \ T$ is a linear transformation. 
 
 \begin{remark}\label{remark1}
  Let $\vec{a}=(a_0,a_1,\ldots,a_{n-1}) \in \Fq^n$   and  $\sigma \in \aut(\Fq)$  be an automorphism of $\Fq.$ The following follows directly from the definition of (skew) polycyclic codes.
 \begin{enumerate}
 \item  Skew $(\sigma,\vec{a})$-polycyclic   codes are invariant under the $\sigma$-semi-linear transformation $  T_{\sigma,\vec{a}} ,$  called the \emph{$(\sigma,\vec{a})$-polycyclic  shift}, defined by:
 $$  T_{\sigma,\vec{a}} (v_0,v_1,\ldots,v_{n-1} )=  \left( 0, \sigma(v_0), \sigma(v_1), \ldots, \sigma(v_{n-2})\right)+ \sigma(v_{n-1})\vec{a}, $$
 for all $ v=(v_0,v_1,\ldots, v_{n-1}) \in \Fq^n.$
 \item Polycyclic  codes are invariant under the linear transformation $ T_{\vec{a}}$ defined by:
$$ T_{\vec{a}}(v_0,v_1,\ldots,v_{n-1} )=  \left( 0, v_0, v_1, \ldots, v_{n-2}\right)+ v_{n-1}\vec{a}, $$

for all $ v=(v_0,v_1,\ldots, v_{n-1}) \in \Fq^n.$
 \item  When $\sigma= \id,$ we get  $ T_{\vec{a}}= T_{\id,\vec{a}}.$
 \end{enumerate}
 \end{remark}

 Let $R_f:=\Fq[x;\sigma]/ \langle f(x)\rangle,$ with $ \langle f(x)\rangle $ being the left ideal of $\Fq[x,\sigma]$ generated by $g(x).$ In this work, we mainly work with right polycyclic codes, which we henceforth simply refer to as polycyclic codes. Under the usual identification of vectors with polynomials, 
   	\begin{equation}\label{Realization}
  	\begin{array}{rccc}
  	\Phi \ : \ & \ \Fq^n &\longrightarrow & R_f \\ & & & \\
  	  &v=(v_0, v_{_1},\ldots,v_{_{n-1}}) & \longmapsto &  v(x)=\dsum{i=0}{n-1} v_{_i} x^{i}  
  	\end{array}
  	\end{equation} 
each polycyclic code $C$ of length $n$  associated with a vector $\vec{a}$ is seen as a left $ \Fq[x;\sigma]$-submodule (or as a left ideal if $f$ is central)   in  $ R_{f}.$

In the following result we  recall more  characterizations of skew $(\sigma,\vec{a})$-polycyclic   codes from \cite[Theorem 4, Theorem 5]{Ouazzou2023}.

 \begin{theorem}\label{T2.3}
 A linear code  $ C\subseteq \Fq^{n} $ is a skew $(\sigma,\vec{a})$-polycyclic   code if and only if $\Phi(C)$  is a left $\Fq[x;\sigma]$-submodule of $  R_f.$ Moreover:
 \begin{enumerate}
 \item There is a monic polynomial of least degree $ g(x)\in \Fq[x;\sigma] $ such that $ g(x) $  right divides $ f(x) $
 and $ \Phi(C)=\langle g(x)\rangle. $
 \item The set $ \{ g(x),x g(x),\ldots,x^{k-1}g(x)\} $ forms a basis of $ C $ and the dimension of $ C $ is $ k=n-\deg(g). $

 \item A generator matrix $ G $ of $ C $ is given by:
 \begin{equation}\label{eq10}
 G= \left(
 \begin{array}{cccccccc}
 g_{_0} &g_{_1} &\cdots&g_{_{n-k}}& 0 &\cdots &\cdots & 0\\
 0 & \sigma(g_{_0}) & \sigma(g_{_1}) &\cdots & \sigma(g_{_{n-k}}) & 0 &\cdots & 0\\
 \vdots &\ddots &\ddots &\ddots & &\ddots & &\vdots\\
 \vdots & &\ddots &\ddots &\ddots & &\ddots &\vdots\\
 0 &\ldots & &0 & \sigma^{k-1}(g_{_0}) & \sigma^{k-1}(g_{_1}) &\ldots & \sigma^{k-1}(g_{_{n-k}})\\
 \end{array}
 \right) 
 \end{equation}
 where $k=n-\deg(g) $ and $ g(x)= \displaystyle\sum_{i=0}^{n-k} g_ix^i. $
\item There is a one-to-one correspondence between  right divisors of $f(x)$ and skew $(\sigma,\vec{a})$-polycyclic   codes of length $n$ over $\Fq$.
\end{enumerate}
 \end{theorem}

From here to the end of this paper, instead of saying ''$\Phi(C)$ is generated by $g(x)$'', we will just say $C$ is generated by $g(x)$ and write $ C= \Fq[x;\sigma] g(x).$ 
     Moreover, for any  $ \sigma \in \aut(\Fq)$ of order $ \mu $ such that $ n = \mu m$, we define the  set
$$
\mathcal{A}_{\sigma} := \left\{ (a_0, 0, \ldots, 0, a_{\mu}, 0, \ldots, 0, a_{2\mu}, \ldots, a_{(m-1)\mu}, 0, \ldots, 0) \in \mathbb{F}_q^n \mid a_i \in \mathbb{F}_{q}^{\sigma}, \ i = 0,  \ldots, (m-1)\mu \right\}.
$$

\begin{theorem}\cite[Theorem 5.9.]{Bag2025}\\
 Let $ \sigma \in \aut(\Fq)$ be of order $ \mu $ such that $ n = \mu m .$ Then a linear code $ C \subseteq \Fq^n$ is a skew polycyclic code induced by $ \vec{a} \in \mathcal{A}_{\sigma} $  if and only if $ C $ is a left ideal of $ \mathbb{F}_q[x; \sigma] / \langle x^n - \vec{a}(x) \rangle .$
\end{theorem}

From the above theorem we can easily deduce the following result regarding skew trinomial codes.

\begin{corollary}
Let $0<\ell<n$ be an integer such that $\ell$ and $n$ are both multiples of $\mu,$ and $a, b$ two non-zeros elements of $ \Fq^{\sigma}.$ Then $C\subseteq \Fq^n$ is a skew $(\ell,\sigma)$-trinomial code associated with $x^n-ax^{\ell}-b$  if and only if $ C $ is a left ideal of $ \mathbb{F}[x;\sigma]/\langle x^n-ax^{\ell}-b \rangle. $
\end{corollary}

Studying the equivalence between two classes of codes with a cyclic structure reduces to analyzing possible isometries, with respect to the desired metric, between the ambient spaces. In the case of skew polycyclic codes, this amounts to studying isometries between 
$
R_{f_1} := \mathbb{F}_q[x, \sigma]/ \langle f_1(x)\rangle $ and $ R_{f_2} := \mathbb{F}_q[x, \sigma]/\langle f_2(x)\rangle,
$
which preserve the parameters of the codes (length, dimension, and minimum distance) as well as their algebraic structure. For skew polycyclic codes, various approaches can be employed depending on the nature of the ambient spaces under consideration. In the following remark we collect some possible approaches, out of which we will use the second in this work.

\begin{remark}
Let $R_{f_i} := \mathbb{F}_q[x, \sigma]/\langle f_i(x)\rangle$ for $i = 1,2$. 
\begin{enumerate}
    \item $R_{f_1}$ and $R_{f_2}$ are non-associative $\mathbb{F}_q^{\sigma}$-algebras (they are associative if $f_i \in \mathcal{Z}(\mathbb{F}_q[x, \sigma])$). We say that $f_1$ and $f_2$ are \emph{equivalent} if there exists an $ \mathbb{F}_{q_0}$-algebra isomorphism between $R_{f_1}$ and $R_{f_2}$ that preserves the desired distance. In the case of skew constacyclic codes, where $x^n - \lambda \in \mathcal{Z}(\mathbb{F}_q[x, \sigma])$ this was done in \cite[Section 6]{Lobillo2017}. For the commutative case $\sigma = \mathrm{id}$, see \cite{Chen2014,Chen2012,Equiv2025}.
    
    \item Another approach is to study equivalence more simply by requiring only the existence of an $\mathbb{F}_q$-vector space isomorphism $\varphi$ between $R_{f_1}$ and $R_{f_2}$ that preserves the desired distance and satisfies $\varphi(a(x)  b(x)) = \varphi(a(x)) \varphi(b(x))$. This approach was used in \cite{Boulanouar2021} to study equivalence between classes of skew constacyclic, cyclic, and negacyclic codes. We will adopt this approach in our setting; see Definition~\ref{FqIso}.
    \item We can also characterize $\sigma$-skew polycyclic codes induced by a polynomial $f(x)$ as left $\mathcal{L}_{q,\sigma}[x]/\langle f(x^{q_0}) \rangle$-modules, or as left ideals if $f \in \mathcal{Z}(\mathcal{L}_{q,\sigma}[x])$. Studying the equivalence between two such classes reduces to studying isometries between algebras with respect to the desired distance, namely between the $ \mathbb{F}_{q_0}$-algebras $\mathcal{L}_{q,\sigma}[x]/\langle f_1(x^{q_0}) \rangle$ and $\mathcal{L}_{q,\sigma}[x]/\langle f_2(x^{q_0}) \rangle$. For more details on the algebraic structure of skew cyclic codes from this point of view, see \cite{gabidulin2009rank,martinez2017roots}.
\end{enumerate}
\end{remark}

We end this section by recalling the definition of the component-wise product of two vectors that we use throughout this paper. Given two vectors of length $n$, $x = (x_0, x_1, \ldots, x_{n-1})$ and $y = (y_0, y_1, \ldots, y_{n-1})$, the \emph{Schur product} is defined as
$$
x \star y := (x_0 y_0, x_1 y_1, \ldots, x_{n-1} y_{n-1}) .
$$

\section{ Skew Roos-like bound for skew polycyclic codes  }
In \cite{Alfarano2021}, the authors proved a Roos-like bound for the minimum Hamming distance and the minimum rank distance of skew cyclic codes. In their work, they took  $f=x^n-1 \in \Fq[x;\sigma]$, where  $n$ is a multiple of the order of $\sigma$, which means that $x^n-1 $ is a central polynomial.
In this section, we prove this bound for skew \emph{poly}cyclic codes (with respect to the Hamming distance and the rank distance) without assuming that the modulus $f(x)$ is a two-sided polynomial. 

Let us first fix some notation that we use throughout this section.  To decompose a commutative polynomial into a product of linear factors, we usually search for a big extension field which contains all the roots of that polynomial. The same idea is used to decompose skew polynomials, but  in addition we need to define an automorphism over the big field extension, for more details we refer the interested reader to \cite{Lobillo2017}.   Let $\sigma\in \aut(\mathbb{F})$ be an automorphism of $ \mathbb{F},$ of order $\mu.$ 
We say that $\sigma$ has an extension $\theta$ of degree $s$ if there exists a field extension $\mathbb{F} \subseteq \mathbb{L}$ and $\theta \in \aut(\mathbb{L})$ of order $s \mu$ such that $\theta_{\mid \mathbb{F}}=\sigma$ and $ \mathbb{L}^{\theta}=\mathbb{F}^{\sigma}.$ 
Note that if $\mathbb{F}$ is a finite field, extensions of any degree always exist.\footnote{However,  if $\mathbb{F}$ is the field of rational functions over a finite field, they exist when $\mu$  order of $\sigma,$ and $s$ are coprime, see \cite[Examples 2.4 and 2.5]{Lobillo2017}.}

Let  $f(x)$ be a skew polynomial of degree $n$, with right exponent $e ,$ such that $f(0)\ne 0.$ To construct the  extension $\theta$ of $\sigma,$ we have to consider the following two cases; 
        \begin{itemize}
            \item[•] \textbf{If $e$ is a multiple of the order  $\mu,$} i.e., $e=m\mu$ then the   $\mathbb{F}_{q_0}$-Frobenius automorphism of $ \mathbb{F}_{q_0^e}$ is an   extension of degree $m$ of $\sigma.$ 
          
           \item[•] \textbf{If  $e$ is not a multiple of $\mu,$}  we set  $e':=\lcm(e,\mu)$ and  verify that $ f(x)$ is a right divisor of $ x^{e'}-1,$ since $x^e-1$ right divides $ x^{e'}-1.$  So in this case  the   $\mathbb{F}_{q_0}$-Frobenius automorphism of $ \mathbb{F}_{q_0^{e'}}$ is  an extension of degree 
           $m =\frac{e'}{\mu}$ of $\sigma.$ 
            \end{itemize}
In the sequel, we will assume that $ e$ is a multiple of the order  $\mu,$   and the  $\mathbb{F}_{q_0}$-Frobenius automorphism of $ \mathbb{F}_{q_0^e},$ say $\theta,$ is an extension of degree $m:= \frac{e}{\mu}$ of $\sigma.$  According to  Hilbert's Theorem 90 (see \cite[Chapter VI; Theorem 6.1]{Lang2002}),  $\beta \in \mathbb{F}_{q_0^e} $ is a root of $x^e-1$ if and only if there exist $ \alpha \in \mathbb{F}_{q_0^{e}}^* $ such that  
   $\beta=\theta(\alpha) \alpha^{-1},$ i.e.,   $\beta$ is a $\sigma$-conjugate of $1$.   
   Let us now assume that  $\alpha \in \mathbb{F}_{q_0^e}$ is a \emph{normal} element, i.e., $ B:=\left\{\alpha, \theta(\alpha), \ldots, \theta^{e-1}(\alpha)\right\}$ is a basis of $\mathbb{F}_{q_0^e}$ as an $\mathbb F_{q_0}$ vector space.  For $\beta =\theta(\alpha)\alpha^{-1},$ we can prove that the set
   $$  \{ \beta, \theta(\beta), \ldots,  \theta^{e-1}(\beta)\}$$ forms a $P$-independent set and  a $P$-basis  of  the vanishing set  $V(x^e-1) $  of $x^e-1.$ It follows that the rank of $V(x^e-1) $ equals $e,$ and so we can decompose $x^e-1$ over $ \mathbb{F}_{q_0^{e}} $ as follows:     
      $$
x^e-1=\lclm(\left\{x-\theta^i(\beta) \mid 0 \leq i \leq e-1\right\}).
$$

Now we recall the idea used in \cite{Cathryn2024, Lobillo2025} to decompose the skew polynomial $x^n-a.$

\begin{remark} \label{RmqConsta}
In \cite{Cathryn2024, Lobillo2025} the authors studied 
$a$-constacyclic codes of length $n$ over $\Fq,$ which correspond in our case to skew polycyclic codes associated 
    with the polynomial $f(x)=x^n-a$. 
    Their  idea was to find a decomposition of $x^n-a$ in a big extension  $ \mathbb{F}_{q_0^{\mu t}},$ where $ n=\mu t$ under the  condition that $ a=N_n^{\sigma}(\gamma)$ for some $\gamma\in \Fq.$ This latter condition ensures that $x^n-a$ has at least a root in $\Fq. $ Let $\alpha$ be a normal element in $\mathbb{F}_{q^t}$ and  $ \beta=\theta(\alpha) \alpha^{-1},$ with $\theta \in \aut(\mathbb{F}_{q^t})$ an extension of $\sigma$ of degree $t.$ Then  $\{ \gamma \beta, \gamma \theta(\beta),\ldots,  \gamma \theta^{t-1}(\beta )  \}$ forms a $P$-independent set and  a $P$-basis  of  the vanishing set  $V(x^n-a) $  of $x^n-a.$ It follows that the rank of $V(x^e-1) $ equals $n,$ and so they  decomposed $x^n-a$ over $ \mathbb{F}_{q^{t}} $ as follows:     
      $$
x^n-a=\lcm(\left\{x- \gamma \theta^i(\beta) \mid 0 \leq i \leq n-1\right\}).
$$
\end{remark}

Keeping the same notation as above,  we give the definition of the $\beta$-defining set of a given skew polycyclic code associated with a polynomial  $f\in \Fq[x;\sigma].$
\begin{definition} 
Let $g \in \Fq[x;\sigma]$ be  a right divisor of $f(x),$ denoted $ g(x) ~|_r~ f(x),$ and $ C= \Fq[x,\sigma] g(x)$ and $\widehat{C}= \mathbb{F}_{q_0^e}[x,\sigma] g(x) .$
The \emph{$\beta$-defining set} of $g$ is
$$
T_\beta(g)=\left\{ 0 \leq i \leq e-1  ~:~ x-\theta^i(\beta) ~~|_r~~ g \right\}.
$$
\end{definition}
In particular, $\lclm(\left\{x-\theta^i(\beta) ~ \mid~ i \in T_\beta(g) \right\}) ~ ~|_r~~ g$ in $ \mathbb{F}_{q_0^e}[x;\sigma].$ In the sequel we will prove that this division is in  $ \mathbb{F}_{q}[x;\sigma].$

The following lemma will be useful to prove Theorem \ref{thm:Roos} and Theorem \ref{Rank_bound}. 
\begin{lemma}[Lemma 12, \cite{Alfarano2021}]\label{Circulant_lemma}
Let $E/K$ be an extension field of finite degree $n$. Let $\alpha_1, \ldots, \alpha_{t+r} \in E$ be linearly independent elements over $K$, and  $\left\{k_0, \ldots, k_r\right\}$ $\subseteq\{0, \ldots, n-1\}$ be such that $k_r-k_0 \leq t+r-1$ and $k_{j-1}<k_j$ for $1 \leq j \leq r$. Let

$$
A_0=\left(\begin{array}{cccc}
\theta^{k_0}\left(\alpha_1\right) & \theta^{k_1}\left(\alpha_1\right) & \cdots & \theta^{k_r}\left(\alpha_1\right) \\
\vdots & \vdots & \ddots & \vdots \\
\theta^{k_0}\left(\alpha_{t+r}\right) & \theta^{k_1}\left(\alpha_{t+r}\right) & \cdots & \theta^{k_r}\left(\alpha_{t+r}\right)
\end{array}\right) .
$$
Let $b \in\{0, \ldots, n-1\}$ such that $(b, n)=1$ and
$$
A_i=\left(A_0\left|\theta^b\left(A_0\right)\right| \ldots \mid \theta^{b i}\left(A_0\right)\right)
$$
for $0 \leq i \leq r$. Then $\rank\left(A_{t-1}\right)=t+r$.
\end{lemma}

\subsection{Hamming distance}

The following theorem gives a skew Roos-like bound on Hamming distance for skew polycyclic codes.
    \begin{theorem}[Skew Roos-like bound for the Hamming distance]\label{thm:Roos}
    Let $g \in \Fq[x;\sigma]$ be such that $ g(x) ~|_r~ f(x)$, let $ C:= \Fq[x;\sigma] g(x)$ be a skew polycyclic code generated by $g(x)$ and $\widehat{C}:= \mathbb{F}_{q_0^e}[x;\sigma] g(x) $ be its extension code over $ \mathbb{F}_{q_0^e}.$
    Let $\beta:= \theta(\alpha)\alpha^{-1}$ where $\alpha$ is a normal element of $ \mathbb{F}_{q_0^e}$. 
    Suppose that there exist integers $ a, b, \delta, r, k_0, \dots , k_r$ such that $\gcd( b,e)=1,\  k_j<k_{j+1}\ \text { for } 0 \leq j \leq r-1, \ k_r-k_0 \leq \delta+r-2, $ and  
    $$ (x-{\theta^{a+ib+k_j}(\beta)}) \mid_r g(x)  ~\text{for}~ i=0,\ldots,{\delta-2}, j=0,\ldots, r-1,$$ 
    i.e., $$ \left\{ a+ib+k_j   ~:~ i=0,\ldots,{\delta-2}, j=0,\ldots, r-1 \right\} \subseteq T_{\beta}(g). $$ Then  the minimum Hamming distance $d_H(C)$ of  $C$ satisfies     $d_H(C)\geq   d_H(\widehat{C}) \geq {\delta}+{r}. $
\end{theorem}

\begin{proof}
Let $c$ be a codeword of weight $ w <\delta+r, $ then 
				$$c(x) =\dsum{h=0}{w-1} c_{h}x^{l_h}\;\mbox{ for some }\;\{ l_{0},\ldots,l_{w-1}\}\subseteq\{ 0,\ldots,n-1\}. $$
As $ (x-{\theta^{a+ib+k_j}(\beta)}) \mid_r g(x)$, we have 
				$$ 0=c(\beta^{a+ib+k_j}) = \dsum{h=0}{w-1} c_{h} N_{l_h}(\theta^{a+ib+k_j}(\beta)) = \theta^{a+ib+k_j}(\alpha)^{-1}\dsum{h=0}{w-1} c_{h} \theta^{a+ib+k_j+l_h}(\alpha)$$
It follows that  $ c=(c_1,\ldots,c_w)$ is in the kernel of the matrix $B$ where 

            $$ B= \left[A, \theta^b(A),\theta^{2b}(A),\ldots, \theta^{b(\delta-2)}(A) \right], $$ 
            and  
            $$ 
            A=\left( \theta^{k_j+l_h}(\theta^a(\beta)) \right)_{\substack{0 \leq h \leq w-1 \\ 0 \leq j \leq r-1}}.$$
According to Lemma \ref{Circulant_lemma} , we obtain  that $\rank(B)=w$, and so, $c=0.$  
Therefore, $d_H(C) \geq d_H(\widehat{C}) \geq \delta+r.$ 
\end{proof} 


Then general bound in Theorem \ref{thm:Roos} includes the known bounds for special cases like cyclic and constacyclic codes. 
 For this let  $f\in \Fq[x;\sigma]$ be a skew polynomial of degree $n$ and of right exponent $e ,$ with $f(0)\ne 0.$ 
\begin{enumerate}
    \item If $f=x^n-1$, with $n$ being a multiple of $ \mu ,$  we find the case of skew cyclic codes \cite{Alfarano2021,Lobillo2017}.
      \item If $f=x^n-\lambda $, with $n$ being a multiple of $ \mu $ and $\lambda \in \Fq^{\sigma},$  we find the case of skew constacyclic codes \cite{Lobillo2025}.
      \item  If $f=x^n-\lambda $, with $\lambda= N_n^{\sigma}(\gamma) $ for some $\gamma \in \Fq, $  we find the case of skew constacyclic  codes  \cite{Cathryn2024}.
\end{enumerate}

The above bound includes particular cases, namely the skew BCH-like bound and the Hartmann–\,Tzeng-like bound (in the Hamming distance), as summarized in the following: 

\begin{corollary}[\cite{Alfarano2021, Lobillo2017, Cathryn2024, Lobillo2025}]\label{CorBch}
Let $g \in \Fq[x,\sigma]$ be such that $ g(x) ~|_r~ f(x),$ let $ C= \Fq[x,\sigma] g(x)$ be a skew polycyclic code generated by $g(x)$ and let $\widehat{C}= \mathbb{F}_{q_0^e}[x,\sigma] g(x) $ be its extension code over $ \mathbb{F}_{q_0^e}.$ Let $\beta= \theta(\alpha)\alpha^{-1}$ where $\alpha$ is a normal element of $ \mathbb{F}_{q_0^e}$. 
\begin{enumerate}
    \item (Skew BCH-like bound) If there exist integers $a,b$ such that  
    $$ \{(x-\theta^{a+jb}(\beta)) \mid_r g(x) ~:~ j=0,\ldots, \delta-2\},$$ 
 i.e., $$ \left\{ a+jb   ~:~ j=0,\ldots,{\delta-2},  \right\} \subseteq T_{\beta}(g), $$
 with  $\gcd(e, b)=1$,  then $C$ has  minimum Hamming distance  $ d_H(C)\geq   d_H(\widehat{C}) \geq {\delta}.$
 
 \item  (Skew Hartmann-Tzeng-like bound)      If there exist integers $a,b$ such that 
 $$ \{ (x-{\theta^{a+ib+jc}(\beta)}) \mid_r g(x)  ~:~ i=0,\ldots,{\delta-2}, j=0,\ldots, {r}\},$$ 
  i.e., $$ \left\{ a+ib+jc  ~:~ i=0,\ldots,{\delta-2}, j=0,\ldots, {r} \right\} \subseteq T_{\beta}(g), $$
	with   $\gcd(e, b)=1 $ and $\gcd(e, c)<\delta$, then  the minimum Hamming distance $d_H(C)$ of  $C$ satisfies   $d_H(C)\geq   d_H(\widehat{C})  \geq {\delta}+ {r}. $ 
\end{enumerate}
\end{corollary}

\subsection{Rank distance}

We now turn to the rank distance and prove the analog bound. In this section  we consider the rank distance $ \Fq/ \mathbb{F}_{q_0}$ where $ \mathbb{F}_{q_0} =\Fq^{\sigma}$ is the fixed subfield of $ \Fq$ by $\sigma$.

\begin{theorem}[Skew Roos-like bound for the rank distance]\label{Rank_bound}
Let $g \in \Fq[x,\sigma]$ be such that $ g(x) ~|_r~ f(x),$ let $ C:= \Fq[x,\sigma] g(x)$ be a skew polycyclic code generated by $g(x)$ and let $\widehat{C}:= \mathbb{F}_{q_0^e}[x,\sigma] g(x) $ be its extension code over $ \mathbb{F}_{q_0^e}.$ Let $\beta: = \theta(\alpha)\alpha^{-1}$ where $\alpha$ is a normal element of $ \mathbb{F}_{q_0^e}$. 
    Suppose there exists integers $ a, b, \delta, r, k_0, \dots , k_r$ such that $\gcd( b,e)=1,\  k_j<k_{j+1}\ \text { for } 0 \leq j \leq r-1, \ k_r-k_0 \leq \delta+r-2, $ and  
    $$\{ (x-{\theta^{a+ib+k_j}(\beta)}) \mid_r g(x)  ~\text{ for }~ i=0,\ldots,{\delta-2}, j=0,\ldots, r-1.\}$$ 
    Then  the minimum rank distance of $C$ satisfies     $d_R(C)\geq   d_R(\widehat{C}) \geq {\delta}+ {r}. $
\end{theorem}
\begin{proof}
According to Proposition \ref{Prop_Hamming_rank}, it suffices to prove $d_H(C M) \geq \delta+r$ for all $M \in \GL(\mathbb{F}_{q_0})$. Let $M \in \GL(\mathbb{F}_{q_0})$, and assume that there is a codeword $c \in C M^{-1}$ of Hamming weight $w < \delta+r$, i.e., $\widehat{c} := cM \in C$. Let $\widehat{c}(x) := \dsum{h=0}{n-1} \widehat{c}_{h}x^{h}\in C $. For  $i = 0, \ldots, {\delta - 2}, j = 0, \ldots, r,$ we have  that  $x - {\theta^{a+ib+k_j}(\beta)}~|_r~ g(x).$ Then it right divides $\widehat{c}(x)$, and so
$$
0 = \widehat{c}(\beta^{a+ib+k_j}) = \dsum{h=0}{n-1} \widehat{c}_{h} N_{h}(\theta^{a+ib+k_j}(\beta)) = \theta^{a+ib+k_j}(\alpha)^{-1} \dsum{h=0}{n-1} \widehat{c}_{h} \theta^{a+ib+k_j+h}(\alpha) = 0.
$$
It follows that $c = \widehat{c}M^{-1}$ is in the kernel of the matrix $MB$, where
$$
B = \left[A, \theta^b(A), \theta^{2b}(A), \ldots, \theta^{b(\delta-2)}(A)\right],
$$
and 
$$
A = \left( \theta^{k_j+h}(\theta^a(\beta)) \right)_{\substack{0 \leq h \leq n-1 \\ 0 \leq j \leq r}}.
$$
Now let $S := \{l_0, \ldots, l_{w-1}\} \subseteq \{0, \ldots, n-1\}$ be the set of the non-zero components of $c$, since we assumed that $w_H(c) = w$. Set $\tilde{c} = (c_{l_0}, \ldots, c_{l_{w-1}}) \in \Fq^{w}$ and let $M_S$ be the rows of $M$ indexed by $S$.  
Then $cM = \tilde{c} M_{S}$, and so $\tilde{c}$ is in the kernel of $M_{S} B$. As $M_S \in \mathbb{M}_{w\times n}(\mathbb{F}_{q_0})$, we obtain that
$$
M_{S} B = \left[M_{S}A, \theta^b(M_{S} A), \theta^{2b}(M_{S}A), \ldots, \theta^{b(\delta-2)}(M_{S}A)\right].
$$
Next, we observe that the matrix $M_SA$ is of the form
$$
M_SA = \left( \theta^{k_j}(\theta^a(\beta_h)) \right)_{\substack{0 \leq h \leq w-1 \\ 0 \leq j \leq r}},
$$
where $\beta_h := M_{\{l_h\}} \star\alpha^{[\theta]},$ with $\star$ means component-wise product of the vectors $ M_{\{l_h\}}$ and $ \alpha^{[\theta]}, $ and  $\alpha^{[\theta]} := (\alpha, \theta(\alpha), \ldots, \theta^{n-1}(\alpha)) \in \mathbb{F}_{q_0^e}^n, $ are linearly independent over $\mathbb{F}_{q_0}$. In fact,
let $t_0, \ldots, t_{w-1} \in \mathbb{F}_{q_0}$ be such that $t_0 \beta_0 + t_1 \beta_1 + \ldots + t_{w-1} \beta_{w-1} = 0$,
then
$$
0 = \dsum{h=0}{w-1} t_h \beta_h 
= \dsum{h=0}{w-1} t_h \left(\dsum{i=0}{n-1} m_{l_h, i} \theta^{i}(\alpha) \right) 
= \dsum{h=0}{w-1} \dsum{i=0}{n-1} t_h m_{l_h, i} \theta^{i}(\alpha) 
= \dsum{i=0}{n-1} \left(\dsum{h=0}{w-1} t_h m_{l_h, i} \right) \theta^{i}(\alpha).
$$
Since $\left\{\alpha, \theta(\alpha), \ldots, \theta^{e-1}(\alpha)\right\}$ is a basis of $\mathbb{F}_{q_0^e}/\mathbb{F}_{q_0}$ and $n \leq e$,  $\alpha, \theta(\alpha), \ldots, \theta^{n-1}(\alpha)
$ are linearly independent over $\mathbb{F}_{q_0}$. It follows that $\dsum{h=0}{w-1} t_h m_{l_h, i} = 0$ for all $i = 0, \ldots, n-1$.  This forces $\left(t_0, \ldots, t_{w-1}\right) M_S = 0$. Now, since $M$ is invertible, we know $\rank(M_S) = w$, so $t_h = 0$ for all $0 \leq h \leq w-1$. Thus, $\beta_0, \ldots, \beta_{w-1}$ are linearly independent over $\mathbb{F}_{q_0}$. Therefore, by Lemma \ref{Circulant_lemma}, we have $\rank(M_S B) = w$. This forces $\tilde{c} = 0$ and hence $c = 0$. Therefore, $d_R(C) \geq d_R(\widehat{C}) \geq \delta + r$. 
\end{proof}

\begin{remark}\label{Rmq_gamma}
    In \cite[Example 5.1.5]{Cathryn2024}, K. Hechtel gave an example where the Roos bound on the rank metric for skew $ a $-constacyclic codes does not hold for $ \gamma \in (\mathbb{F}_{q_0} \setminus \mathbb{F}_{q}) $ such that $ a = N_n^{\sigma}(\gamma) .$ Following our approach, which uses the decomposition of $ x^e - 1 $ instead of $ x^n - a, $ where $e$ is the right order of $x^n-a$ (see Remark \ref{RmqConsta}), we obtain that the bound holds without any restriction. The difference between our approach and hers lies in the proof of the above theorem. Specifically, in the proof that  $ \beta_0, \beta_1, \ldots, \beta_{w-1} $ are linearly independent. However, in her argument, linear independence holds only if $ \gamma \in \mathbb{F}_{q_0} ;$ for more details, see \cite[ Theorem 5.1.4]{Cathryn2024}.

\end{remark}

As above we can recover the particular cases of the skew BCH-like bound and the skew Hartmann-Tzeng-like bound on the rank metric for skew polycyclic codes. 
\begin{corollary} [\cite{Alfarano2021, Cathryn2024,  Lobillo2025}] \label{CorRank_Bch}
Let $C$ be a skew polycyclic code associated with $f,$ of length $n$ over $\Fq$ generated by $  g(x) \in \mathbb{F}_{q}[x;\sigma]$  and let $ \widehat{C} $ be its extension code of the same length and generator over $ \mathbb{F}_{q_0^e}$ i.e., $\widehat{C}= \mathbb{F}_{q_0^e}[x;\theta] g$.  Let $\beta= \theta(\alpha)\alpha^{-1}$ where $\alpha$ is a normal element of $ \mathbb{F}_{q_0^e}$.
\begin{enumerate}
    \item (Skew BCH-like bound) If   $$ \{(x-\theta^{a+jb}(\beta)) \mid_r g(x) ~\text{ for }~ j=0,\ldots, \delta-2\},$$
 such that  $\gcd(e, b)=1$,  then $ d_R(C)\geq   d_R(\widehat{C}) \geq {\delta}.$ 
 
 \item  (Skew Hartmann-Tzeng-like bound)    If   $$ \{ (x-{\theta^{a+ib+jc}(\beta)}) \mid_r g(x)  ~\text{ for }~ i=0,\ldots,{\delta-2}, j=0,\ldots, {r}\},$$ 
	where   $\gcd(e, b)=1 $ and $\gcd(e, c)<\delta$, 
    then  the minimum rank distance $d_R(C)$ of  $C$ satisfies $d_R(C)\geq   d_R(\widehat{C})  \geq {\delta}+ {r}. $ 
\end{enumerate}
\end{corollary}

\subsection{Constructions of skew polycyclic codes with designed distance. }

To construct classical cyclic (constacyclic, polycyclic) codes of length $n$ over $\mathbb{F}_q$, one studies the factorization of $x^n - 1$ over $\mathbb{F}_q$, obtained by considering its roots in an extension field of $\mathbb{F}_q$. These roots form a multiplicative group of order $n$, generated by a primitive element $\zeta$. The irreducible factors of $x^n - 1$ correspond to cyclotomic cosets that are $\langle q \rangle$-orbits on $\mathbb{Z}_n$. Each coset $C_i$ defines a minimal polynomial  $
M_{\zeta^i}(x) = \prod_{k \in C_i} (x - \zeta^k), $ which is irreducible over $\mathbb{F}_q$.

In the case of skew codes, an alternative to the classical concept of cyclotomic cosets is the notion of  \emph{$\mu$-closed sets}, introduced by Gómez-Torrecillas et al.\ in \cite{Lobillo2017} and further developed in \cite{Alfarano2021, Cathryn2024, Lobillo2025}. We adapt this concept to construct skew polycyclic codes in our setting. Let $\mathbb{Z}_e = \{0,1,\ldots,e-1\}$ be the cyclic group of order $e$, and suppose $e = m\mu$. Then,
$$
\mu \mathbb{Z}_e = \{0, \mu, \ldots, (m-1)\mu\}
$$
is a subgroup of order $m$, and the quotient group $\mathbb{Z}_e / \mu \mathbb{Z}_e$ partitions $\mathbb{Z}_e$ as
$$
\mathbb{Z}_e = \bigcup_{i=0}^{m-1} (i + \mu \mathbb{Z}_e).
$$
Let $\alpha \in \mathbb{F}_{q_0^e}$ be such that $\{\alpha, \theta(\alpha), \ldots, \theta^{e-1}(\alpha)\}$ forms a normal basis of $\mathbb{F}_{q_0^e}/\mathbb{F}_{q_0}$, and set $\beta := \alpha^{-1} \theta(\alpha)$. Then, as previously discussed, the polynomial $x^e - 1$ factors over $\mathbb{F}_{q_0^e}$ as:
$$
x^e - 1 = \lclm\left\{ (x - \beta), (x - \theta(\beta)), \ldots, (x - \theta^{e-1}(\beta)) \right\}.
$$
Grouping roots according to $\mu$-closed sets, we obtain
$$
x^e - 1 = \lclm\left\{ M_{\beta}(x), M_{\theta(\beta)}(x), \ldots, M_{\theta^{\mu - 1}(\beta)}(x) \right\},
$$
where for each $i \in \{0, \ldots, \mu - 1\}$,
$$
M_{\theta^i(\beta)}(x) := \lclm \left\{ x - \theta^{i + j\mu}(\beta) \mid 0 \leq j < m \right\}.
$$

We now recall the following result, which will be used in the construction of skew polycyclic codes with some designed distance.

\begin{lemma}(\cite[Lemma 5]{Lobillo2025}).\label{Skew_Roots}
Let $ \gamma  \in \mathbb{F}_{q_0^e} $ and $g(x)\in \mathbb{F}_{q}[x;\sigma].$ If $ \gamma   $ is a right root of $g,$ then $ \theta^{\mu}(\gamma)$ is also a right root of $g.$ 
\end{lemma}

\begin{proposition}\label{mu-closed}
Let $f\in \Fq[x;\sigma]$ be of order $e$ such that $e=m \mu,$ and  let $\left\{\alpha, \theta(\alpha), \ldots, \theta^{e-1}(\alpha)\right\}$ be a normal basis of $ \mathbb{F}_{q_0^e}/\mathbb{F}_{q_0}$. Define $\beta:=\theta(\alpha) \alpha^{-1}.$ 
\begin{enumerate}
    \item For each $ i=0,\ldots, \mu-1,$   the polynomial $ M_{\theta^i(\beta)}(x)\in \Fq[x;\sigma] $ and $ \deg(M_{\theta^i(\beta)}(x) )= m.$
   Moreover it is  the minimal polynomial of $ T_{\beta}^{(i)}:= \{ \theta^i(\alpha), \theta^{i+\mu}(\alpha), \ldots, \theta^{i+\mu(m-1)}(\alpha)\}$ over $\Fq.$ 
    \item If $ \theta^{i}(\beta)$ is a right root of $g\in \Fq[x,\sigma]$ then $ M_{\theta^i(\beta)}(x) $ right divides $ g(x)$ in $\Fq[x,\sigma].$
    \end{enumerate}
\end{proposition}
\begin{proof}
\begin{enumerate}
    \item  We extend $\theta $ to $\mathbb{F}_{q_0}[x;\theta]$ in the natural way i.e., $ \theta\left( \dsum{i=0}{d-1} a_i x^i\right)= \dsum{i=0}{d-1} \theta(a_i) x^i.$ As the fixed field of $ \mathbb{F}_{q_0^e}$ by $ \theta^{\mu}$ is $ \Fq,$  it suffices  to prove that $ \theta^{\mu}( M_{\theta^i(\beta)}(x))=  M_{\theta^i(\beta)}(x).$ For each $ i=0,\ldots, \mu-1,$  we have 
    $$
    \begin{array}{rl}
        \theta^{\mu}\left( M_{\theta^i(\beta)}(x)\right)
        &= \theta^{\mu}\left( \lclm\left((x-\theta^{i}(\beta)), (x-\theta^{i+\mu}(\beta)), \ldots, (x-\theta^{i+(m-1)\mu}(\beta)) \right)  \right) \\
         & = \lclm\left((x-\theta^{i+\mu}(\beta)), (x-\theta^{i+2 \mu}(\beta)), \ldots, (x-\theta^{i+m\mu}(\beta)) \right)  \\
         &= M_{\theta^i(\beta)}(x), \ \text{since the order of $\theta$ is $e=m\mu.$} \\  
    \end{array}
     $$
     Then $ M_{\theta^i(\beta)}(x) \in \Fq[x;\sigma].$ As the set $ T_{\beta}^{(i)} $ is a subset of the $P$-independent set $ \left\{\alpha, \theta(\alpha), \ldots, \theta^{e-1}(\alpha)\right\}$, $\deg(M_{\theta^i(\beta)}(x))=m.$ Therefore, it is the minimal polynomial of $ T_{\beta}^{(i)}.$ 
     \item Let  $ \theta^{i}(\beta)$ be a right root of $g\in \Fq[x;\sigma]$ then according to  Lemma \ref{Skew_Roots}, $ \theta^{\mu}(\theta^i(\beta))=\theta^{i+\mu}(\beta)$ is also a right root of $g.$ Consequently, for any $j=0,\ldots,m-1,$ $\theta^{j\mu}(\theta^i(\beta))= \theta^{i+j\mu}(\beta) $ is a root of $g.$ Hence $ M_{\theta^i(\beta)}(x) $ right divides $ g(x)$ in $\Fq[x;\sigma],$
     since it is the minimal polynomial of $\{ \theta^i(\beta),  \theta^{\mu}(\theta^i(\beta)),  \ldots, \theta^{(m-1)\mu}(\theta^i(\beta)) \}. $
    \end{enumerate}

\end{proof}
Now we recall the definition of $ \mu$-closed set and its representative  elements. 
\begin{definition}
Let $ T\subseteq \mathbb{Z}_e$ be a non-zero set. 
\begin{enumerate}
    \item We say that   $ T\subset \mathbb{Z}_e$ is \emph{$\mu$-closed} if   $i \in T$ implies $i+\mu \in T$  for every  $i \in T$. 
    \item If $T$ is $\mu$-closed  and $\{i_1,i_2,\ldots, i_l\} \subseteq T$  are integers such that
    $$T=T^{i_1} \cup \cdots \cup T^{i_l} , \text{ and } T^{i_k}\cap  T^{i_j}= \{0\} $$
    then $ S_{T}:= \{i_1,i_2,\ldots, i_l\}$ is  called the \emph{$\mu$-representative set} of $T.$ 
   \end{enumerate} 
\end{definition}
 Following  \cite[Proposition 17]{Alfarano2021}, we can easily prove the following result.

\begin{proposition}
 Let $g \in \mathbb{F}_q[x; \sigma]$ be such that $g ~|_r~ f(x)$. Let $\alpha \in \mathbb{F}_{q_0^e}$ be such that $\left\{ \alpha, \theta(\alpha), \ldots, \theta^{e-1}(\alpha) \right\}$ forms a normal basis of $\mathbb{F}_{q_0^e}/\mathbb{F}_{q_0}$. Set  $\beta := \theta(\alpha) \alpha^{-1}$, and let $ T_\beta(g) := \left\{ i \in \mathbb{Z}_e : x - \theta^i(\beta) \mid_r g \right\}$ the $\beta$-defining set of $g.$
 \begin{enumerate}
     \item If $ g \in \Fq[x;\sigma]$ then $ T_\beta(g) $ is $\mu$-closed.
     \item Conversely, if $ T_{\beta}(g) \subseteq \mathbb Z_e$ is $\mu$-closed, then $ 
    g':= \lclm( \{x-\theta^i(\beta) ~:~ i \in T_{\beta}\}) \in \Fq[x,\sigma].$
 \end{enumerate}
\end{proposition}
\begin{proof}
\begin{enumerate}
     \item  Let $  g \in \Fq[x;\sigma] $ and  $j\in T_\beta(g) ,$  then  $g(\theta^j(\beta))=0$, i.e.,  $ \gamma:=\theta^{j}(\beta)$ is a right root of $g.$  By Lemma \ref{Skew_Roots},  $ \theta^{\mu}(\gamma)=\theta^{j+\mu}(\beta) $   is also a right root of $g,$ and so $ j+\mu \in  T_{\beta}(g). $   Hence  $ T_\beta(g) $ is $\mu$-closed.
     \item Suppose that  $ T_{\beta} \subseteq \mathbb Z_e$  is $\mu$-closed.   Let $ \{i_1,i_2,\ldots, i_l\}$  be  the representative set of $T_{\beta}(g),$ and $ M_{\theta^{i_j}(\beta)}(x):= \lclm\{(x-\theta^{i_1}(\beta)), (x-\theta^{i_1+\mu}(\beta)),\ldots, (x-\theta^{i_1+(m-1)\mu}(\beta))$ for $ j=1,\ldots,r.$ It follows that $g'=\lclm\{ M_{\theta^{i_1}(\beta)}(x), M_{\theta^{i_2}(\beta)}(x),\ldots, M_{\theta^{i_r}(\beta)}(x) \}. $ So by Proposition \ref{mu-closed},  we deduce that $g'\in\Fq[x;\sigma].$  
 \end{enumerate}
\end{proof}

We know that if $g \in \mathbb{F}_{q}[x ; \sigma]$ is a right  divisor of $f$,  then it is in fact the minimal polynomial for its $\beta$-defining set $T_{\beta}(g)$. In the following result, we will prove that the size of the $\mu$-representative set of $T_{\beta}(g)$ has implications for when  $ C := \mathbb{F}_q[x;\sigma] g(x) $  is an MRD code. The proof of the following proposition is analogous to \cite[Proposition 26]{Alfarano2021} (the case of skew cyclic codes) and \cite[Proposition 5.2.4]{Cathryn2024} (the case of skew constacyclic codes), so we omit it.

\begin{proposition}\label{SKEW_MRD}
 Let $g \in \mathbb{F}_q[x; \sigma]$ be such that $g ~|_r~ f(x)$. 
 Suppose that the defining set $T_\beta(g)$ satisfies a skew Roos bound as in Theorem \ref{Rank_bound} for some $\delta \geq 2$ and $r \geq 0$. Then the minimum rank distance of the code $C=\mathbb{F}_{q}[x;\sigma] g$ satisfies $$\delta+r \leq d_R(C) \leq \left| S_{T_{\beta}(g)} \right|+1$$ where $  S_{T_{\beta}(g)} \subset \mathbb{Z}_{\mu} $ is the $\mu$-representative set of $T_{\beta}(g).$ 
 In particular, if $\left|  S_{T_{\beta}(g)} \right|=\delta+r-1$, then $ C$ is an MRD code.
\end{proposition}

In the following remark, we discuss how to construct an MRD code from the class of skew polycyclic codes using Proposition \ref{SKEW_MRD}, and present some existing results in the cases of skew cyclic and constacyclic codes.

Keeping the same notations as in Theorem \ref{Rank_bound}, we define
$T_{\beta}(f) := \{ i \in \mathbb{Z}_e \mid (x - \theta^i(\beta)) \mid_r f(x) \}, $ i.e., the set of exponents corresponding to the roots of $f$ in $\mathbb{F}_{q_0^e}[x;\theta]$.
According to Proposition \ref{SKEW_MRD}, we observe that the construction of MRD codes corresponds to finding a $\beta$-defining set $T \subset \mathbb{Z}_e \cap T_{\beta}(f)$ such that
$$|S_{\overline{T}}| = \delta + r - 1. $$ Here, $\overline{T} \subset \mathbb{Z}_e$ is the $\mu$-closed set generated by $T$, and $S_{\overline{T}} \subset \mathbb{Z}_{\mu}$ is its $\mu$-representative set. It follows that $$ g:= \lclm_{ \{ i\in S_{\overline{T}} \}}\left( M_{\theta^i(\beta)}(x)\right)  \in \Fq[x, \sigma] $$
is a right divisor of $f. $ Hence it is a generator of  an $[n, n-m(\delta + r - 1), \delta + r - 1]_q $ MRD code. 
\begin{remark}
In \cite[Theorem 5.2.7]{Cathryn2024}, in the case of constacyclic codes, it was proved that if the above equality holds, that is, $ S_{\overline{T}} = \delta + r - 1, $
then $S_{\overline{T}}$ must be an arithmetic progression with common difference $b$, for any value of $\mu.$ This means that
$$ S_{\overline{T}} =\{ a_0+ k b\ ~\text{with }~ a_0 \in \mathbb{Z}_{\mu} ~\text{ and }~ 0\leq k\leq \mu-1 \}.$$
This question was also addressed in \cite[Proposition 30]{Alfarano2021}, in the context of skew cyclic codes, when $\mu$ is a prime number.
\end{remark}

We finish this section with an example of a skew polycyclic code and its designed as well as actual minimum Hamming and rank distance.
\begin{example} 
Let $ \mathbb{F}_{2^6}= \mathbb{F}_{2}(w), $ with $ w^6 + w^4 + w^3 +  w+ 1=0 $ and  $f(x)=x^{10} + w^{40}x^9 + w^{39}x^8 + w^{12} x^6 + w^{46}x^5 + w^{42}x^4 + w^{60} x^2 + w^7x +
    w^{54}\in \mathbb{F}_{2^6}[x;\sigma]$ with $\sigma$ the Frobenius automorphism of $ \mathbb{F}_{2^6}$ then its order is $\mu=6.$ Using Magma and SageMath we found that the left exponent  of $f$ is $e:= 12.$  As $e$ is a multiple of $\mu$, the Frobenius automorphism $\theta$ of $\mathbb{F}_{2^{12}}$ is an extension of degree $ m=2 $ of $\sigma.$ Let $\gamma$ be a generator of $ \mathbb{F}_{2^{12}}$ such that $ \gamma^12 + \gamma^7 + \gamma^6 + \gamma^5 + \gamma^3 + \gamma + 1, $   and $ \alpha=\gamma^5 $ be a normal element of  $\mathbb{F}_{2^{12}}.$  Set $ \beta:=\theta(\alpha)\alpha^{-1},$ then the skew polynomial  $ x^{ 12}-1$ can be decomposed in $ \mathbb{F}_{2^{12}}[x;\sigma]$ as 
    $$  x^{ 84}-1 = \lclm \left( x-\beta, x-\theta(\beta), \ldots, x-\theta^{11}(\beta)\right)= \lclm \left( M_{\beta}(x), M_{\theta(\beta)}(x), \ldots, M_{\theta^5(\beta)}(x)\right)  $$
with
    $$\begin{array}{l}
      M_{\beta}(x)=  \lclm ( x-\beta, x-\theta^6(\beta))= x^2 + w^{38}x + w^{58} \\    
      M_{\theta(\beta)}(x)=  \lclm ( x-\theta(\beta), x-\theta^7(\beta))= x^2 + w^{13} x + w^{53}  \\
      
      M_{\theta^2(\beta)}(x)=  \lclm ( x-\theta^2(\beta), x-\theta^8(\beta))= x^2 + w^{26} x + w^{43}\\
      M_{\theta^3(\beta)}(x)=  \lclm ( x-\theta^3(\beta), x-\theta^9(\beta))= x^2 + w^{52}x + w^{23} \\ M_{\theta^4(\beta)}(x)=  \lclm ( x-\theta^4(\beta), x-\theta^10(\beta))= x^2 + w^{41} x + w^{46}\\
      M_{\theta^5(\beta)}(x)=  \lclm ( x-\theta^5(\beta), x-\theta^11(\beta))= x^2 + w^{19}x + w^{29}. 
    \end{array}
    $$
  Let now $g= x^4 + w^{52}x^3 + w^{46}x^2 + w^{23}x + w^{33 }$ be a right divisor of $f(x).$ For  $ a=2,\ b=1,\ \delta=3,\  r=1 ,$ we find that the $\beta$-defining set of $g$ is   $T_{\beta}(g)=\{ 2,3,8,9 \}$ and so $g= \lclm(x-\theta^2(\beta), x-\theta^3(\beta), x-\theta^8(\beta), x-\theta^9(\beta)) . $ Let $C$ be the skew polycyclic code induced by $f$ and generated by $g.$
  According to Theorem \ref{thm:Roos}, $C$ has minimum Hamming distance $d_H(C)\geq 3+1=4, $ and by the aid of magma we found that the exact Hamming  distance is $d_H(C)=4,$ and hence that $C$ is an $[10,6,4]_{2^6}$ code.
  
  Similarly, according to Theorem \ref{Rank_bound},  $C$  has rank distance (over $\mathbb{F}_2$) $d_R(C)\geq 3+1=4, $ and with the aid of Magma we found that the exact rank distance $d_R(C)=4,$ since it is a codeword $c=(0, 0, 1, 0, 0, 0, w^{37}, w^{57}, 0,    w^7 )$ of rank weight $4.$ Therefore $C$ is an $[10,6,4]_{2^6}$ rank-metric code.

\end{example}

\section{Equivalence  of  skew  polycyclic codes over $\Fq$}

Generalizing the usual equivalence between codes, we consider equivalence maps between two ambient spaces of skew polycyclic codes which preserve the algebraic structure of the code. This notion was first introduced in \cite{Chen2012}, to study the equivalence between two classes of classical constacyclic  codes over $\Fq.$ Recently it was extended to the case of skew constacyclic codes \cite{Boulanouar2021,Ouazzou2025}, \cite[Section 6]{Lobillo2025},  and to the case of polycyclic codes \cite{Equiv2025, Chibloun2025polycyclic}.  

\begin{definition}\label{FqIso}  
Let $f_1$ and $f_2$ be two skew polynomials from $\mathbb{F}_q[x;\sigma]$. We say that $f_1$  is \emph{Hamming} (resp. \emph{rank}) \emph{equivalent} to $f_2$, and we denote this by $ f_1 \cong_{(n,\sigma)-\text{Hamming}} f_2, $ (resp. $\cong_{(n,\sigma)-\text{Rank}}$)  if there is an $\mathbb{F}_q$-vector space isomorphism 
$$\varphi ~:~  \mathbb{F}_q[x;\sigma] /\langle f_2(x)\rangle   \quad \longrightarrow  \quad  \mathbb{F}_q[x;\sigma] /\langle 
f_1(x)\rangle,$$
such that, for all $ f(x), g(x) \in \Fq[x;\sigma] /\langle f_2(x) \rangle,$ we have 
\begin{enumerate}  
    \item $\varphi(g(x)f(x))= \varphi(g(x)) \varphi(f(x)) $.  
    \item $d_H(\varphi(f(x)), \varphi(g(x)))=d_H(f(x), g(x))$ \\ (resp. $ d_R(\varphi(f(x)), \varphi(g(x)))=d_R(f(x), g(x))$).
\end{enumerate}  
We will call such an isomorphism $\varphi$ an \emph{$\mathbb{F}_q$-morphism isometry} with respect to the Hamming (resp. rank)   distance.
\end{definition}  

We remark that, in this section,  we consider the rank distance defined over $ \Fq/ \mathbb{F}_{q'}$, where $ \mathbb{F}_{q'} $ is an arbitrary subfield of $ \Fq.$

\begin{proposition}\label{ProofEquiv}
Let $ R_{ n}$ be the set of monic skew polynomials of degree $n.$ 
Then the relations $\cong_{(n,\sigma)-Hamming}$ and  $\cong_{(n,\sigma)-Hamming}$  are equivalence relations on $ R_{n}.$ 
\end{proposition}
\begin{proof}
    See Appendix.
\end{proof}

As commonly done in the literature, we will focus on $\Fq$-morphism isometries of the form  $\varphi_{\alpha}(f(x)) = f(\alpha x)$, for some $\alpha \in \mathbb{F}_q^*$. Note that any such map preserves the code length $n$, the dimension $k$, and the Hamming distance; however, it only preserves the rank distance if $\alpha\in \mathbb{F}_{q'}$.\footnote{Similarly, any $ \mathbb{F}_{q'}$-vector space isomorphism $\varphi$ preserves the code length $n$, the dimension $k$, and the rank distance, but not necessarily the Hamming metric.}
However, such maps do not necessarily preserve the algebraic structure of the codes, i.e., the skew polycyclical property  of the code. To ensure this, we must verify that
$$
\varphi(f(x)  g(x)) = \varphi(f(x))  \varphi(g(x)).
$$
Verifying this will be the main objective in the following sections. For simplicity we will first focus on trinomial codes, first in the Hamming and then in the rank metric, before we generalize the results to general polycyclic codes in the third subsection.

\subsection{ Hamming equivalence  of  skew  trinomial codes}

Now, we give the following special case of equivalence,  which we refer to as  {Hamming $(n,\ell,\sigma)$-equivalence}, between two classes of $(\ell,\sigma)$-trinomial codes. 

 \begin{definition}\label{Def_Iso}
Let $a_0, a_{\ell}, b_0, b_{\ell}$ be nonzero elements of $\mathbb{F}_q$, and let $\ell$ be an integer such that $0 < \ell < n$. We say that $(a_0, a_{\ell})$ and $(b_0, b_{\ell})$ are { Hamming $(n,\ell,\sigma)$-equivalent} 
in $\mathbb{F}_q^* \times \mathbb{F}_q^*$, denoted by 
$$
(a_0, a_{\ell}) \sim_{(n,\ell,\sigma)} (b_0, b_{\ell}),
$$
if there exists an $\alpha \in \mathbb{F}_q^*$ such that the following map:

  \begin{equation}
  	\begin{array}{cccc}
  	\varphi_{\alpha}:& \mathbb{F}_q[x;\sigma] /\langle x^n-b_{\ell}x^{\ell}-b_0 \rangle  &\longrightarrow & \mathbb{F}_q[x;\sigma] /\langle x^n-a_{\ell}x^{\ell}-a_0 \rangle, \\ 
  	& f(x) & \longmapsto &  f(\alpha x),
  	\end{array}
  	\end{equation} 
is an  $\mathbb{F}_q$-morphism isometry with respect to the Hamming distance,  i.e., 
$$\varphi_{\alpha}( f(x) g(x))= \varphi_{\alpha}( f(x)) \varphi_{\alpha}( g(x)), \quad \text{and } \quad d_H(\varphi_{\alpha}(f(x)), \varphi_{\alpha}(g(x)))=d_H(f(x), g(x)), $$ 
for all $ f(x), g(x) \in \Fq[x;\sigma] /\langle x^n-b_{\ell}x^{\ell}-b_0 \rangle.$
\end{definition}

\begin{remark}\label{Remark.3}
\begin{enumerate}
 \item  Note that the $(n,\ell,\sigma)$-equivalence relation in the above definition generalizes the $(n,\sigma)$-equivalence of constacyclic codes studied in \cite[Section 4]{Ouazzou2025}, which was denoted by $ \lambda \sim_{(n,\sigma)} \mu .$
 \item  If $\sigma=\id$ we recover  the case of $(n,\ell)$-equivalence for  $\ell$-trinomial   codes studied in \cite{Equiv2025}.

\end{enumerate}
\end{remark}

In the following theorem we give essential characterizations of  the Hamming $(n,\ell,\sigma)$-equivalence   between two classes of skew $(\ell,\sigma)$-trinomial codes of length $n$ over $\mathbb{F}_q$.

\begin{theorem}\label{Th.1}
 Let $0<\ell<n$ be an integer,  $(a_0, a_{\ell})$ and $(b_0, b_{\ell})$ be elements of $\mathbb{F}_q^{*} \times \mathbb{F}_q^{*}$,   $\xi$ be  a primitive element of $\Fq$   and  $ \sigma$ be an automorphism of $\Fq$. The following statements are equivalent:

\begin{enumerate}
    \item $(a_0, a_{\ell}) \sim_{(n,\ell,\sigma)} (b_0, b_{\ell}),$ i.e.,  $(a_0, a_{\ell}) $ and  $ (b_0, b_{\ell})$ are Hamming \textcolor{purple}{$(n,\ell,\sigma)$} equivalent.  

\item   There exists $\alpha\in \Fq^*$ such that $$ a_0 N_n^{\sigma}(\alpha)  = b_0   \quad  \text{and} \quad   a_{\ell} N^{\sigma}_{n-\ell}(\sigma^{\ell}(\alpha))=  b_{\ell}. $$
\item  There exists $\alpha\in \Fq^*$ such that $$ (a_0, a_{\ell}) \star
\left( N_n^{\sigma}(\alpha), N_{n-\ell}^{\sigma}(\sigma^{\ell} (\alpha)) \right)= (b_0,b_{\ell}) .$$
    
 \item   $ (a_0,a_{\ell})^{-1}\star (b_0,b_{\ell})\in H_{\ell,\sigma},$ where   $H_{\ell,\sigma}$ is the cyclic subgroup of $\Fq^*\times \Fq^*$ generated by $\left(N^{\sigma}_{n}(\xi), N_{n-\ell}^{\sigma}(\sigma^{\ell} (\xi))   \right).$
\end{enumerate}
The equivalence between (1) and (4) implies that the number of $(n,\ell,\sigma)$-equivalence classes is $$N_{(n,\ell,\sigma)}:= \dfrac{ (q-1)^2}{ \lcm\left(
\frac{q-1}{\gcd([n]_r,q-1)}, \frac{q-1}{ \gcd([n-\ell]_r,q-1)}\right)} .$$
\end{theorem}

\begin{proof}\;
\begin{itemize}
    \item[(1) $\Rightarrow$ (2):]
   Suppose that $(a_0, a_{\ell}) \sim_{(n,\ell,\sigma)} (b_0, b_{\ell})$. Then by Definition \ref{Def_Iso}, there exists $\alpha \in \mathbb{F}_q^*$ such that the map
    $$
    \varphi_{\alpha} : \Fq[x;\sigma] /\langle x^n - b_{\ell} x^{\ell} - b_0 \rangle \to \Fq[x;\sigma] /\langle x^n - a_{\ell} x^{\ell} - a_0 \rangle, \quad f(x) \ \mapsto \ f(\alpha x)
    $$
    is an  $\mathbb{F}_q$-morphism isometry.  It follows that
$$ 
\varphi_{\alpha}(x^i)= (\alpha x)^i= N_i^{\sigma}(\alpha) x^{ i}\ \text{mod}_r (x^n-a_{\ell} x^{\ell}-a_0), \ \forall  \ i=0,1, \ldots, n-1.
$$

As $\varphi_{\alpha}$ is an $\Fq$-morphism and  $ \varphi(x^n-b_{\ell}x^{\ell}-b_0)=0 \ \text{mod}_r (x^n- a_{\ell} x^{\ell}-a_0) ,$ we have
\begin{equation}
\varphi( x^n)=   b_{\ell}  N_{\ell}^{\sigma}(\alpha) x^{\ell} +b_0.
 \end{equation}

On the other hand,
\begin{equation}
\varphi( x^n)=  N_n^{\sigma}(\alpha) x^n= N_n^{\sigma}(\alpha) ( a_{\ell}x^{\ell}+a_0) = N_n^{\sigma}(\alpha) a_{\ell} x^{\ell}+ N_n^{\sigma}(\alpha) a_0.
\end{equation}

Comparing term by term, we obtain  that $$ a_0 N_n^{\sigma}(\alpha)  = b_0, \quad    \text{and} \quad   a_{\ell} N_n^{\sigma}(\alpha) =  N_{\ell}^{\sigma}(\alpha) b_{\ell}. $$ It follows that
$$ a_0 N_n^{\sigma}(\alpha)  = b_0   \quad  \text{and} \quad   a_{\ell} N_n^{\sigma}(\alpha)  N_{\ell}^{\sigma}(\alpha^{-1})=  b_{\ell}. $$
 On the other hand we have 
\begin{equation}\label{Eq_formulae}
\begin{array}{rl}
    N_n^{\sigma}(\alpha) N_{\ell}^{\sigma}(\alpha^{-1})  & =  N_{\ell+(n-\ell)}^{\sigma}(\alpha) N_{\ell}^{\sigma}(\alpha^{-1})\\
     &=   \sigma^{\ell}(N_{n-\ell}^{\sigma}(\alpha)) N^{\sigma}_{\ell}(\alpha) N_{\ell}^{\sigma}(\alpha^{-1})  \\
     &=  \sigma^{\ell}(N^{\sigma}_{n-\ell}(\alpha))\\
     & = N^{\sigma}_{n-\ell}(\sigma^{\ell}(\alpha))
\end{array}
\end{equation}
Hence we obtain that 
$$ a_0 N_n^{\sigma}(\alpha)  = b_0   \quad  \text{and} \quad   a_{\ell} N^{\sigma}_{n-\ell}(\sigma^{\ell}(\alpha))=  b_{\ell}. $$

 \item[(2) $\Rightarrow$ (3):] Immediate.
 \item[(3) $\Rightarrow$ (4):] Suppose that

$$  (b_0, b_{\ell} )= \left( N_{n}^{\sigma}(\alpha) a_0, N_{n-\ell}^{\sigma}(\sigma^{\ell}(\alpha) a_{\ell}\right)=\left(N_{n}^{\sigma}(\alpha), N_{n-\ell}^{\sigma}(\sigma^{\ell}(\alpha))\right)\star (a_0, a_{\ell}).$$ 
It follows that
$$  (a_0,a_{\ell})^{-1}\star (b_0, a_{\ell})=  \left( N_{n}^{\sigma}(\alpha), N_{n-\ell}^{\sigma}( \sigma^{\ell}(\alpha)) \right)= \left( N_{n}^{\sigma}(\xi^j),
N_{n-\ell}^{\sigma}(\sigma^{\ell}(\xi^j))\right)= 
\left(N_{n}^{\sigma}(\xi), N_{n-\ell}^{\sigma}(\sigma^{\ell}(\xi))\right)^j .$$
Then $ (a_0,a_{\ell})^{-1}\star (b_0, a_{\ell}) $ belongs to the cyclic subgroup  $H_{\ell,\sigma}$ generated by $  \left(N_{n}^{\sigma}(\xi), N_{n-\ell}^{\sigma}(\sigma^{\ell}(\xi))\right) $  as a subgroup of $ \Fq^{*}\times \Fq^{*}.$

    \item[(4) $\Rightarrow$ (1):] 
     Suppose that $ (a_0,a_{\ell})^{-1}\star (b_0, b_{\ell}) $ is an element of the cyclic subgroup  $H_{\ell,\sigma}$ generated by 
     $  (N_{n}^{\sigma}(\xi), N_{n-\ell}^{\sigma}(\sigma^{\ell}(\xi)))  $  as a subgroup of $ \Fq^{*}\times \Fq^{*}.$ Then there exists an integer $ h$ such that 
     $$ (a_0,a_{\ell})^{-1}\star (b_0, b_{\ell})= \left(N_{n}^{\sigma}(\xi), N_{n-\ell}^{\sigma}(\sigma^{\ell}(\xi))\right)^h = \left( N_{n}^{\sigma}(\xi^h), N_{n-\ell}^{\sigma}(\sigma^{\ell}(\xi^h))\right).$$
    For $ \alpha=  \xi^{h}, $ we obtain  that $  a_i N_{n-i}^{\sigma}( \sigma^{i}(\alpha))= b_i, $ for any $ i\in \{ 0,\ell\}.$
Now, let us consider the map $ \tilde{\varphi}_{\alpha}$  
  \begin{equation}
        \begin{array}{cccc}
        \tilde{\varphi}_{\alpha} : & \mathbb{F}_q[x;\sigma]   & \longrightarrow & \mathbb{F}_q[x;\sigma] /\langle x^n - a_{\ell} x^{\ell} - a_0 \rangle, \\ 
        & f(x) & \longmapsto & f(\alpha x).
        \end{array}
    \end{equation} 
	$ \tilde{\varphi}_{\alpha}$ is a surjective $\Fq$-vector space homomorphism. Indeed, for  all $0\leq j \leq n-1$, $x^{j} = \tilde{\varphi}_{\alpha}(N_{j}^{\sigma}(\alpha^{-1}) x^{j}).$ 
	Moreover, 
 $$ 
    \begin{array}{rl}
    \tilde{\varphi}_{\alpha} (x^n - b_{\ell} x^{\ell} - b_0) & = N^{\sigma}_n(\alpha) x^n - N^{\sigma}_{\ell}(\alpha) b_{\ell} x^{\ell} - b_0 \\
    & = N^{\sigma}_n(\alpha) x^n - N^{\sigma}_{\ell}(\alpha)  N_{n-\ell}^{\sigma}( \sigma^{\ell}(\alpha)) a_{\ell} x^{\ell} - N^{\sigma}_n(\alpha) a_0 \\
    & = N^{\sigma}_n(\alpha) x^n -  \sigma^{\ell}(N_{n-\ell}^{\sigma}( \alpha))  N^{\sigma}_{\ell}(\alpha)  a_{\ell} x^{\ell} - N^{\sigma}_n(\alpha) a_0 \\
    & = N^{\sigma}_n(\alpha) x^n - N_{\ell+ n-\ell}^{\sigma}( \alpha)  a_{\ell} x^{\ell} - N^{\sigma}_n(\alpha) a_0 \\
     & = N^{\sigma}_n(\alpha) x^n - N^{\sigma}_n(\alpha)  a_{\ell} x^{\ell} - N^{\sigma}_n(\alpha) a_0 \\
    & = N^{\sigma}_n(\alpha) (x^n - a_{\ell} x^{\ell} - a_0) \\
    & = 0 \mod (x^n - a_{\ell} x^{\ell} - a_0).
    \end{array}
    $$ 
	So $\langle x^n - b_{\ell} x^{\ell} - b_0\rangle \subseteq \ker\tilde{\varphi}_{\alpha}$. And for all $f(x)\in \ker\tilde{\varphi}_{\alpha}$,
    $$f(\alpha x)= 0~ (\text{mod } x^{n}- a_{\ell} x^{\ell} - a_0).$$ Therefore, there exists $g(x)\in \Fq[x;\sigma]$ such that $f(\alpha x)=g(x)(x^{n}- a_{\ell} x^{\ell} - a_0)$. Thus, 
 $$
  \begin{array}{rl}
 f(x)&=g(\alpha^{-1}x)( N^{\sigma}_n(\alpha^{-1}) x^{n}- N^{\sigma}_{\ell}(\alpha^{-1}) a_{\ell} x^{\ell} - a_0) \\
 & =N^{\sigma}_n(\alpha^{-1}) g(\alpha^{-1}x)(x^{n}- N^{\sigma}_{n-\ell}(\alpha) a_{\ell} x^{\ell} - N^{\sigma}_n(\alpha) a_0)\\
 &= N^{\sigma}_n(\alpha^{-1}) g(\alpha^{-1}x)( x^n-  b_{\ell} x^{\ell} - b_0), \ \text{since  $a_k N^{\sigma}_{n-k}( \sigma^k(\alpha)) =b_k, \ \forall k \in \{0, \ell\}.$} \\ 
 \end{array}$$
	So $\ker\tilde{\varphi}_{\alpha} \subseteq \langle x^{n}- a_{\ell} x^{\ell} - a_0\rangle $ and hence $$\ker\tilde{\varphi}_{\alpha} = \langle x^{n}-  b_{\ell} x^{\ell} - b_0 \rangle.$$
	Therefore  by the  first isomorphism  theorem, the map 
      \begin{equation}
        \begin{array}{cccc}
        \varphi_{\alpha} : & \mathbb{F}_q[x;\sigma]/\langle x^n - b_{\ell} x^{\ell} - b_0 \rangle  & \longrightarrow & \mathbb{F}_q[x;\sigma] /\langle x^n - a_{\ell} x^{\ell} - a_0 \rangle, \\ 
        & f(x) & \longmapsto & f(\alpha x).
        \end{array}
    \end{equation} 
	is an  isomorphism. As  the weights of $f(x)$ and $f(\alpha x)$ are the same,  $\varphi_{\alpha}$  is an isometry.
    
    We now need  to verify that $\varphi_{\alpha} $ is a $\Fq$-morphism.
    Let $ f(x)=\dsum{i=0}{n-1} f_ix^i$ and $g(x)=  \dsum{i=0}{n-1} g_ix^i$  in $ \Fq[x;\sigma]/\langle x^n- b_{\ell}x^{\ell}-b_0 \rangle.$ On the one hand we have 

           $$
           \begin{array}{rl}
f(x)g(x) &=\dsum{j=0}{n-1}\dsum{i=0}{j} f_i \sigma^i\left(g_{j-i}\right)  x^j+ 
\dsum{j=0}{n-1} \dsum{i=j+1}{n-1} f_i \sigma^i\left(g_{n-i+j}\right)  x^{j+n} \\
 
\end{array}
$$
as $x^{j+n}= \sigma^{j}(b_{\ell}) x^{j+\ell} + \sigma^{j}(b_{0}) x^{j+\ell} \mod \ (x^n-b_{\ell}x^{\ell} -b_0).$ Then
$$
f(x)g(x) = \sum_{j=0}^{n-1} \left( \sum_{i=0}^{j} f_i \sigma^i(g_{j-i}) \right) x^j +
\sum_{j=0}^{n-1} \left( \sum_{i=j+1}^{n-1} f_i \sigma^i(g_{n-i+j}) \sigma^j(b_\ell)x^{j+\ell} + \sum_{i=j+1}^{n-1} f_i \sigma^i(g_{n-i+j}) \sigma^j(b_0)x^j \right).
$$
It follows that
$$
f(x)g(x) = \sum_{j=0}^{n-1} \left( \sum_{i=0}^{j} f_i \sigma^i(g_{j-i}) + \sum_{i=j+1}^{n-1} f_i \sigma^i(g_{n-i+j}) \sigma^j(b_0) \right) x^j +
\sum_{j=0}^{n-1} \left( \sum_{i=j+1}^{n-1} f_i \sigma^i(g_{n-i+j}) \sigma^j(b_\ell) \right) x^{j+\ell}.
$$

Next, as $\varphi_\alpha(x^j) = N_j^\sigma(\alpha)x^j$ and $\varphi_\alpha(x^{j+\ell}) = N_{j+\ell}^\sigma(\alpha)x^{j+\ell},$ we have  
\begin{equation}\label{eq11}
\begin{array}{rl}
\varphi_\alpha(f(x)g(x)) &= \dsum{j=0}{n-1} \left( \dsum{i=0}{j} f_i \sigma^i(g_{j-i})  N_j^\sigma(\alpha) +
 \dsum{i=j+1}{n-1} f_i \sigma^i(g_{n-i+j}) \sigma^j(b_0)  N_j^\sigma(\alpha) \right)x^j \\

&+\dsum{j=0}{n-1} \left( \dsum{i=j+1}{n-1} f_i \sigma^i(g_{n-i+j}) \sigma^j(b_\ell) N_{j+\ell}^\sigma(\alpha) \right) x^{j+\ell}.
\end{array}
\end{equation}

On the other hand, we   compute $\varphi_\alpha(f(x))\varphi_\alpha(g(x)).$
$$
\begin{array}{rl}
 \varphi_\alpha(f(x))\varphi_\alpha(g(x))    &=  \left( \dsum{i=0}{n-1} f_i N_i^\sigma(\alpha)x^i \right) \left( \dsum{j=0}{n-1} b_j N_j^\sigma(\alpha)x^j \right) \\
    & = \dsum{j=0}{n-1} \left( \dsum{i=0}{j} f_i N_i^\sigma(\alpha) \sigma^i(g_{j-i}N_{j-i}^\sigma(\alpha)) \right) x^j \\
    &  +\dsum{j=0}{n-1} \left( \dsum{i=j+1}{n-1} f_i N_i^\sigma(\alpha) \sigma^i(g_{n-i+j}N_{n-i+j}^\sigma(\alpha)) \right)x^{j+n}.
\end{array}
$$
As $x^{j+n} \equiv \sigma^j(a_\ell)x^{j+\ell} + \sigma^j(a_0)x^j \mod (x^n - a_\ell x^\ell - a_0)$,  we have 
\begin{equation}\label{eq12}
\begin{array}{rl}
\varphi_\alpha(f(x))\varphi_\alpha(g(x)) &= \dsum{j=0}{n-1} \left( \dsum{i=0}{j} f_i\sigma^i(g_{j-i})  N_i^\sigma(\alpha) \sigma^i(N_{j-i}^\sigma(\alpha))\right)x^j \\
&+ \dsum{j=0}{n-1} \left( \dsum{i=j+1}{n-1} f_i  \sigma^i(g_{n-i+j})  \sigma^j(a_0) N_i^\sigma(\alpha)\sigma^i(N_{n-i+j}^\sigma(\alpha))  \right)x^j 
\\
& +\dsum{j=0}{n-1} \left( \dsum{i=j+1}{n-1} a_i  \sigma^i(g_{n-i+j}) N_i^\sigma(\alpha) \sigma^i(N_{n-i+j}^\sigma(\alpha))\sigma^j(a_\ell) \right)x^{j+\ell}.\\
\end{array}
\end{equation}

By \cite[Proposition 2.1]{Cherchem2016},  
$
N_{i+j}^\sigma(\alpha) = N_i^\sigma(\alpha)\sigma^i(N_j^\sigma(\alpha)).
$
Therefore, 
$$ N_i^\sigma(\alpha) \sigma^i(N_{j-i}^\sigma(\alpha))= N_j^{\sigma}(\alpha) .$$ 
So, for any $k\in \{ 0, \ell\}$, we have 

$$ 
\begin{array}{rl}
& N_i^\sigma(\alpha)\sigma^i(N_{n-i+j}^\sigma(\alpha))  \sigma^j(a_{k})   \\
 & = N_{n+j}^{\sigma}(\alpha)  \sigma^j(a_{k})  \\
  &=  \sigma^{ j} ( N^{\sigma}_{n}(\alpha)) N^{\sigma}_j(\alpha)  \sigma^{ j}(a_{k}) \\
  &= \sigma^{ j} ( N^{\sigma}_{n}(\alpha))  N^{\sigma}_j(\alpha) \sigma^{ j} (b_{k}) \sigma^{ -j} ( N^{\sigma}_{n-k}(\sigma^{k}(\alpha))), \ \ 
  \text{since } \ 
  a_{k}  N^{\sigma}_{n-k}(\sigma^{k}(\alpha))  = b_{k}\\
  &= \sigma^{ j} ( N^{\sigma}_{n}(\alpha)) N^{\sigma}_{j}(\alpha)   \sigma^{ k-j} ( N^{\sigma}_{n-k}(  \alpha)) \sigma^{ j} (b_{k}) \\
&=  \sigma^{ j} ( N^{\sigma}_{n}(\alpha)) N^{\sigma}_{j}(\alpha)   \sigma^{ k-j} ( N^{\sigma}_{n-k}(  \alpha))     N_{k-j}^{\sigma}(\alpha)  N_{k-j}^{\sigma}(\alpha^{-1}) \sigma^{ j} (b_{k})   \\ 
&=  \sigma^{ j} ( N^{\sigma}_{n}(\alpha)) N^{\sigma}_{j}(\alpha)   N_{n-j}^{\sigma}(\alpha) N_{k-j}^{\sigma}(\alpha^{-1}) \sigma^{ j} (b_{k})   \\ 
&=  \sigma^{ j} ( N^{\sigma}_{n}(\alpha)) N^{\sigma}_{j}(\alpha)  \sigma^{-j}(N_n^{\sigma}(\alpha)) N_{-j}^{\sigma}(\alpha) N_{k-j}^{\sigma}(\alpha^{-1}) \sigma^{ j} (b_{k})   \\ 
&=   N^{\sigma}_{j}(\alpha)  N_{-j}^{\sigma}(\alpha) N^{\sigma}_{-j}(\alpha^{-1}) \sigma^{-j}(N^{\sigma}_{k}(\alpha^{-1})) \sigma^{ j} (b_{k})   \\ 
&=   N^{\sigma}_{j}(\alpha)  \sigma^{-j}(N^{\sigma}_{k}(\alpha^{-1})) \sigma^{ j} (b_{k})   \\ 
 & =N^{\sigma}_{j}(\alpha)  \sigma^{j}(N^{\sigma}_{k}(\alpha)) \sigma^{ j} (b_{k})   \\ 
 & =N^{\sigma}_{k+j}(\alpha)  \sigma^{ j} (b_{k})   
\end{array}
$$
By comparing equations  (\ref{eq11}) and  (\ref{eq12}), we get 
$$
\varphi_\alpha(f(x)g(x)) = \varphi_\alpha(f(x))\varphi_\alpha(g(x)).
$$

Hence we have shown that $\varphi_\alpha$ is an $\mathbb{F}_q$-morphism which preserves the Hamming weight,  so
$(a_0, a_{\ell}) \sim_{(n,\ell,\sigma)} (a_0, a_{\ell}).$ Hence,  (4) $\Rightarrow$ (1).
\end{itemize}

Now, by the equivalence between (1) and (4) we deduce that the number of $(n,\ell,\sigma)$-equivalence classes on $\Fq^*\times \Fq^*$  corresponds to the order of the group $ \left( \Fq^*\times \Fq^* \right)/H_{\ell,\sigma},$ which equals 
 $$ N_{(n,\ell,\sigma)}= \dfrac{ (q-1)^2}{ \lcm \left(\frac{q-1}{ \gcd([n]_r,q-1)}, \frac{q-1}{ \gcd(p^{r\ell}[n-\ell]_r,q-1)}\right)}=
 \dfrac{ (q-1)^2}{ \lcm \left(\frac{q-1}{ \gcd([n]_r,q-1)}, \frac{q-1}{ \gcd([n-\ell]_r,q-1)}\right)}
 .$$ 

\end{proof}

\begin{remark}\label{Remark_TwoSided}
The reader may wonder what happens if we assume that the polynomials
$ x^n - a_{\ell}x^{\ell} - a_0 $ and $ x^n - b_{\ell}x^{\ell} - b_0 $
are elements of the center $ Z(\mathbb{F}_q[x; \sigma]) = \mathbb{F}_q^{\sigma}[x^{\mu}] $,
where $ \mu $ is the order of $ \sigma .$ In this case  $ \Fq[x;\sigma]/\langle  x^n-a_{\ell}x^{\ell}-a_0\rangle $ and $ \Fq[x;\sigma]/\langle  x^n-b_{\ell}x^{\ell}-b_0\rangle $ are $ \Fq^{\sigma}$-algebras (associative), and the definition of the $(n,\sigma)$-equivalence remains the same as in Definition \ref{Def_Iso}.  Moreover, all the statement of the Theorem \ref{Th.1}, remain the same except for the last one on the number of $(n,\ell,\sigma)$-equivalence classes which will be given by
    $$ \dfrac{ (\vert \Fq^{\sigma}\vert -1)^2}{ \lcm \left(\frac{\vert \Fq^{\sigma}\vert-1}{ \gcd([n]_r,\vert \Fq^{\sigma}\vert-1)}, \frac{\vert \Fq^{\sigma}\vert-1}{ \gcd([n-\ell]_r,\vert \Fq^{\sigma}\vert-1)}\right)}= \dfrac{ (q_0-1)^2}{ \lcm \left(\frac{q_0-1}{ \gcd([n]_r,q_0-1)}, \frac{q_0-1}{ \gcd([n-\ell]_r,q_0-1)}\right)}
 .$$ 
\end{remark}

Using the equivalence between the assertions $(1)$ and $(4)$ of Theorem \ref{Th.1} we derive the following characterization regarding the equivalence between the class of skew $(\ell,\sigma)$-trinomial codes  associated with  $x^n - a_{\ell}x^{\ell} - a_0$ and the class  associated with $x^n - x^{\ell} - 1.$
\begin{corollary}  \label{Equiv_1}
Let $\sigma$ be and automorphism of $ \Fq,$ $\ell$ be an integer such that  $0<\ell<n$,  and  $(a_0, a_{\ell})$ be an element of $\mathbb{F}_q^{*} \times \mathbb{F}_q^{*}$. 
Then the following statements are equivalent.
    \begin{enumerate}
        \item  The class of skew  $(\ell,\sigma)$-trinomial codes  associated with  $x^n - a_{\ell}x^{\ell} - a_0$ is $(n,\ell,\sigma)$-equivalent (Hamming equivalence) to the class of   skew $(\ell,\sigma)$-trinomial codes  associated with $x^n - x^{\ell} - 1.$
        \item  There exists $ \alpha \in \Fq^*$ such that $(a_0, a_{\ell})\star\left(N^{\sigma}_{n}(\alpha), N^{\sigma}_{n-\ell}(\sigma^{\ell}(\alpha)) \right)= (1, 1)$. 
        \item  There exists  an $(n,\sigma)$-th root  $\alpha \in \Fq^* $  of $a_0$     such that $   a_0 =  N^{\sigma}_{\ell}(\alpha) a_\ell $.
        \end{enumerate}
\end{corollary}
\begin{proof}\;
\begin{itemize}
    \item[(1) $\Rightarrow$ (2):] Follows from the equivalence of  (1) and (3) of Theorem \ref{Th.1}.
 \item[(2) $\Rightarrow$ (3):]  
 If $(a_0, a_{\ell})\star\left(N_n^{\sigma}(\alpha), N_{n-\ell}^{\sigma}(\sigma^{\ell}(\alpha))\right) = (1, 1)$, then 
$$ a_0 N_n^{\sigma}(\alpha)=1 \quad \quad 
\text{and} \quad  
 \quad a_\ell N_{n-\ell}^{\sigma}(\sigma^{\ell}(\alpha))=1.
$$
 It follows that
 $$  a_0 = N_n^{\sigma}(\alpha^{-1}) \quad \text{ and } \quad  a_\ell N_{n-\ell}^{\sigma}(\sigma^{\ell}(\alpha^{-1}) ) N_n^{\sigma}(\alpha^{-1}) = a_0.
  $$
 $$ \iff a_0 = N_n^{\sigma}(\alpha^{-1}) \quad \text{ and } \quad  a_{\ell} N_{\ell}^\sigma(\alpha^{-1})= a_0 $$
 
  i.e., $\beta:=\alpha^{-1}$ is an $(n,\sigma)$-th root of $a_0$ and $  a_{\ell} N_{\ell}^\sigma(\beta)= a_0  $.

 \item[(3) $\Rightarrow$ (1):] Suppose that $\alpha $ is an $(n,\sigma)$-th root of $a_0$ such that $ a_0 =  N^{\sigma}_{\ell}(\alpha) a_\ell.$

 $$ N_{n}^{\sigma}(\alpha)  = a_0\ \ \text{ and } \ \   
   a_{\ell}=  N^{\sigma}_{\ell}(\alpha^{-1})a_0= N^{\sigma}_{\ell}(\alpha^{-1}) N_{n}^{\sigma}(\alpha). $$ 
From (\ref{Eq_formulae}), we have $ N^{\sigma}_{\ell}(\alpha^{-1}) N_{n}^{\sigma}(\alpha)=  N_{n-\ell}^{\sigma}(\sigma^{\ell}(\alpha)).$
Then 
$$ N_{n}^{\sigma}(\alpha)  = a_0\ \ \text{ and } \ \   
   a_{\ell}= N_{n-\ell}^{\sigma}(\sigma^{\ell}(\alpha)). $$ 
 It follows that  
$$(a_0,a_{\ell} )\star \left(N_{n}^{\sigma}(\alpha^{-1}), N_{n-\ell}^{\sigma}(\sigma^{\ell}(\alpha^{-1})) \right)= (1,1) .$$
 By the second statement of Theorem \ref{Th.1}, we conclude that  $ (a_0, a_{\ell}) \sim_{(n,\ell,\sigma)}(1,1) ,$ and the result holds.

\end{itemize}
 
 \end{proof}

  In the following result, we describe the associated polynomials of possible skew $(\ell,\sigma)$-trinomial codes based on the $(n,\ell,\sigma)$-equivalence relation defined above in Definition \ref{Def_Iso}.
\begin{theorem} \label{Equiv_2}
Let $n, \ell$ be two integers such that $0<\ell<n,$ let $\xi$ be  a primitive element of $\Fq$, and $\sigma$ be an automorphism of $ \Fq.$  Set  $ 
 d:=\gcd \left(\frac{q-1}{ d_0}, \frac{q-1}{ d_{\ell}}\right)$  and $ d_i:=\gcd([n-i]_r,q-1)$, for $ \ i\in \{0,\ell\}$.  Then the class of skew  $(\ell,\sigma)$-trinomial  codes associated with $x^n-a_{\ell}x^{\ell}-a_0$ is $(n,\ell,\sigma)$-equivalent (Hamming equivalence) to the class of skew  $(\ell,\sigma)$-trinomial  codes associated with  $x^n-\xi^{j} x^{\ell}- \xi^{i+h[n]_r},$ for some $i\in\{0,1,\ldots, d_0-1\}$,  $j\in \{0,1,\ldots, d_{\ell}-1\}$  and $ h\in \{0,\ldots, d-1\}$.  
\end{theorem}

\begin{proof}
\begin{enumerate}
    \item  If $ d = 1 $, then the cyclic group $ H_{\ell,\sigma} $ generated by $ \left( N_n(\xi), N_{n-\ell}(\sigma^{\ell}(\xi))\right) $ is isomorphic to the group $ \langle  N_n(\xi) \rangle \times \langle N_{n-\ell}(\sigma^{\ell}(\xi)) \rangle $ and has order $ \frac{(q-1)^2}{d_0 d_{\ell}} $. By Theorem \ref{Th.1}, the number of $ (n, \ell) $-equivalence classes is $ d_0 d_{\ell} $. Therefore, we can partition $ \Fq^* \times \Fq^* $ as follows

$$
\Fq^* \times \Fq^* =  \bigcup_{i=0}^{d_0-1} \bigcup_{j=0}^{d_{\ell}-1} (\xi^i, \xi^j) H_{\ell,\sigma}.
$$
Then any pair $(a_0,a_{\ell})$ is $(n,\ell,\sigma)$-equivalent to one of the pairs   $(\xi^i, \xi^j), $ for  $i=0,1,\ldots, d_0-1$, and  $j=0,1,\ldots, d_{\ell}-1$.
 \item 
 If $d\neq 1, $ the number of  $(n,\ell,\sigma)$-equivalence classes is $ d d_0 d_{\ell} ,$ and so we partition $ \Fq^* \times \Fq^* $ as follows:
$$  \Fq^{*}\times \Fq^{*} = \bigcup_{h=0}^{d-1}  \bigcup_{i=0}^{d_0-1} \bigcup_{j=0}^{d_{\ell}-1} (\xi^{i+h[n]_r} ,\xi^{j}) H_{\ell,\sigma}, $$
which implies the second statement, similarly to the first case.
\end{enumerate}

\end{proof}

In the following proposition, we deduce additional characterizations of the $(n,\ell,\sigma)$-equivalence between skew $(\ell,\sigma)$-trinomial codes, linking this $(n,\ell,\sigma)$-equivalence to that of skew  constacyclic codes \cite{Ouazzou2025,Boulanouar2021,Lobillo2025}.

\begin{proposition}\label{Eq_Consta}
Let $ (a_0, a_{\ell}) $ and $ (b_0, b_{\ell}) $ be elements of $ \Fq^{*} \times \Fq^{*} $ such that $ (a_0, a_{\ell}) \sim_{(n, \ell,\sigma)} (b_0, b_{\ell}) $. Then:
\begin{enumerate}
    \item For each $ i \in \{0, \ell\} $, $ a_i^{-1} b_i \in \langle N_{n-i}^{\sigma}(\sigma^i(\xi)) \rangle $, where $ \xi $ is a primitive element of $ \Fq^* $.
    \item For each $ i \in \{0, \ell\} $, $ (a_i^{-1} b_i)^{d_i} = 1 $, where $ d_i := \frac{q-1}{\gcd([n-i]_r, q-1)} $.
    \item For each $ i \in \{0, \ell\} $, $ a_i \sim_{(n-i,\sigma)} b_i $.
    \item For each $ i \in \{0, \ell\} $ the class of skew  $(a_i,\sigma)$-constacyclic codes of length $n-i$ is Hamming equivalent to the class of skew   $(b_i,\sigma)$-constacyclic codes of length $n-i$  over $\Fq.$
\end{enumerate}
\end{proposition}
\begin{proof}
\begin{enumerate}
\item  As $ (a_0, a_{\ell}) \sim_{(n, \ell,\sigma} (b_0, b_{\ell})$, by the second  statement of Theorem \ref{Th.1}, 
there exists $\alpha\in \Fq^*$ such that $$ a_0 N_n^{\sigma}(\alpha)  = b_0   \quad  \text{and} \quad   a_{\ell} N^{\sigma}_{n-\ell}(\sigma^{\ell}(\alpha))=  b_{\ell}. $$
This gives
$$ a_0^{-1}  b_0 = N_n^{\sigma}(\alpha)  \quad  \text{and} \quad   a_{\ell}^{-1}   b_{\ell}=  N^{\sigma}_{n-\ell}(\sigma^{\ell}(\alpha)) . $$
Which means that  
$$ a_0^{-1}  b_0 \in   \langle  N_n^{\sigma}(\xi) \rangle  \quad  \text{and} \quad   a_{\ell}^{-1}   b_{\ell} \in  
\langle   N^{\sigma}_{n-\ell}(\sigma^{\ell}(\xi))  \rangle. $$

\item Let   $i\in \{0,\ell\}$ then the order of $ \langle   N^{\sigma}_{n-i}(\sigma^{i}(\xi))  \rangle =
\langle   \xi^{p^{ri}[n]_r}  \rangle $ is  $$d_i=\dfrac{q-1}{\gcd(p^{ri}[n-i]_r, q-1)}= \dfrac{q-1}{\gcd([n-i]_r, q-1)} , \text{ since $ \gcd(p^{ri},q-1)=1.$ }$$  It follows that  $(b_i a_i^{-1})^{d_i} =1 ,\ i\in \{0,\ell\}.$
\item By the second  statement of Theorem \ref{Th.1}, 
there exists $\alpha\in \Fq^*$ such that $$ a_0 N_n^{\sigma}(\alpha)  = b_0   \quad  \text{and} \quad   a_{\ell} N^{\sigma}_{n-\ell}(\sigma^{\ell}(\alpha))=  b_{\ell}. $$  As in the proof of \cite[Theorem 6]{Ouazzou2025}. For $i\in \{0,\ell\},$ we verify
 that the  map  $\varphi_i$ given by 
\begin{equation}
  	\begin{array}{cccc}
  	\varphi_{i}:& \mathbb{F}_q[x] /\langle x^{n-i}-b_0 \rangle  &\longrightarrow & \mathbb{F}_q[x] /\langle x^{n-i}-a_0 \rangle, \\
  	& f(x) & \longmapsto &  f(\sigma^i(\alpha) x). 
  	\end{array}
  	\end{equation} 
    is an $\Fq$-morphism isometry with respect to the Hamming  metric.
    \item Follows straight-forwardly from the third statement.
\end{enumerate}
\end{proof}

\subsection{Rank equivalence  of  skew  trinomial codes}
From the beginning of this section, we know that an $\mathbb{F}_{q}$-morphism is an isometry of the form $\varphi_{\alpha}$, where $\alpha \in \mathbb{F}_{q'}$, preserves the rank distance. Not really, we are just checking maps of the form $\varphi_\alpha$ for simplicity. Moreover, $\varphi_\alpha$ preserves the distance, not $\alpha$ itself. Recall here that the rank of $v=(v_0,v_1,\ldots,v_{n-1})\in \Fq^n$ is the dimension of the $\mathbb{F}_{q'}$-vector space $ \langle v_0,v_1,\ldots,v_{n-1}\rangle_{\mathbb{F}_{q'}}.$ We  now give the definition of rank equivalence in this case, along with a complete characterization, as in the case of Hamming equivalence.

\begin{definition}\label{DefFqSigma}
Let $a_0, a_{\ell}, b_0, b_{\ell}$ be nonzero elements of $\mathbb{F}_q$, and let $\ell$ be an integer such that $0 < \ell < n$. We say that $(a_0, a_{\ell})$ and $(b_0, b_{\ell})$ are {$( n,\ell,\sigma, \mathbb{F}_{q'})$-equivalent} in $\mathbb{F}_q^* \times \mathbb{F}_q^*$ with respect to the rank metric, and we  denote
$$ (a_0, a_{\ell}) \sim_{( n,\ell,\sigma, \mathbb{F}_{q'})} (b_0, b_{\ell}),$$
if there exists an $\alpha \in \mathbb{F}_{q'}^* $ such that the map

  \begin{equation}
  	\begin{array}{cccc}
  	\varphi_{\alpha}:& \mathbb{F}_q[x;\sigma] /\langle x^n-b_{\ell}x^{\ell}-b_0 \rangle  &\longrightarrow & \mathbb{F}_q[x;\sigma] /\langle x^n-a_{\ell} x^{\ell}-a_0 \rangle, \\ 
  	& f(x) & \longmapsto &  f(\alpha x),
  	\end{array}
  	\end{equation} 
is an  $\mathbb{F}_q$-morphism isometry with respect to the rank metric. 
\end{definition}

We give now the following result regarding the rank equivalence between two classes of skew trinomial codes. Because the proof is analogous to the one of Theorem \ref{Th.1}---with the small change that we require $\alpha \in \mathbb{F}_{q'}\subseteq \Fq$---we omit it here. 
\begin{theorem}\label{Th.RankEquiv}
Let $0<\ell<n$ be an integer,  $(a_0, a_{\ell})$ and $(b_0, b_{\ell})$ be elements of $\mathbb{F}_q^{*} \times \mathbb{F}_q^{*}$,  and   $\zeta$ be  a primitive element of $\mathbb{F}_{q'}$. 
Then the following statements are equivalent:

\begin{enumerate}
    \item   $ (a_0, a_{\ell})  \sim_{( n,\ell,\sigma, \mathbb{F}_{q'})} (b_0, b_{\ell}) .$ 
    \item There exists $\alpha \in  \mathbb{F}_{q'}^*$ such that the map

  \begin{equation}
  	\begin{array}{cccc}
  	\varphi_{\alpha}:& \mathbb{F}_q[x;\sigma] /\langle x^n-b_{\ell}x^{\ell}-b_0 \rangle  &\longrightarrow & \mathbb{F}_q[x;\sigma] /\langle x^n-a_{\ell} x^{\ell}-a_0 \rangle, \\ 
  	& f(x) & \longmapsto &  f(\alpha x),
  	\end{array}
  	\end{equation} 
    is an  $\mathbb{F}_q$-morphism isometry with respect to the rank metric.
       
     \item  There exists $\alpha\in  \mathbb{F}_{q'}^* $ such that $ (a_0,a_{\ell})\star (N_n^{\sigma}(\alpha), N_{n-\ell}^{\sigma}(\sigma^{\ell}(\alpha)))= (b_0,b_{\ell}) $.

      \item  $ (a_0, a_{\ell})^{-1} \star(b_0,b_{\ell})\in H_{\ell,  \mathbb{F}_{q'}},$ where   $ H_{\ell, \mathbb{F}_{q'}}$ is the cyclic subgroup of $  \mathbb{F}_{q'}^*\times  \mathbb{F}_{q'}^*$ generated by $( N_n^{\sigma}(\zeta), N_{n-\ell}^{\sigma}(\sigma^{\ell}(\zeta)) ) .$
\end{enumerate}
The equivalence between (1) and (4) implies that the number of $( n,\ell,\sigma, \mathbb{F}_{q'})$-equivalence classes is $$ 
N_{( n,\ell,\sigma, \mathbb{F}_{q'})}=\dfrac{ (q-1)^2}{ \lcm(\frac{q'-1}{ \gcd([n]_r,q'-1)}, \frac{q'-1}{ \gcd([n-\ell]_r,q'-1)})} .$$
\end{theorem}

In the case where we calculate the rank over the fixed subfield $ \mathbb{F}_q^{\sigma} = \mathbb{F}_{q_0}, $ we obtain the following characterizations of rank equivalence. In this case, we provide more characterizations compared to Theorem~\ref{Th.RankEquiv}, using the right roots of our skew trinomial polynomials.

\begin{corollary} \label{Th.1FqSigma} 
Let $0<\ell<n$ be an integer,  $(a_0, a_{\ell})$ and $(b_0, b_{\ell})$ be elements of $\mathbb{F}_q^{*} \times \mathbb{F}_q^{*}$,  and  
$\zeta$ be  a primitive element of $\mathbb{F}_{q_0}:=\Fq^{\sigma}$. 
Then the following statements are equivalent:
\begin{enumerate}
    \item   $ (a_0, a_{\ell})  \sim_{( n,\ell, \Fq^{\sigma})} (b_0, b_{\ell}) ,$ i.e.,  $(a_0, a_{\ell}) $ and $ (b_0, b_{\ell})$ are rank equivalent. 
    \item There exists an $\alpha \in (\mathbb{F}_q^{\sigma})^*$ such that the map

  \begin{equation}
  	\begin{array}{cccc}
  	\varphi_{\alpha}:& \mathbb{F}_q[x;\sigma] /\langle x^n-b_{\ell}x^{\ell}-b_0 \rangle  &\longrightarrow & \mathbb{F}_q[x;\sigma] /\langle x^n-a_{\ell} x^{\ell}-a_0 \rangle, \\ 
  	& f(x) & \longmapsto &  f(\alpha x),
  	\end{array}
  	\end{equation} 
    
is an  $\mathbb{F}_q$-morphism isometry with respect to the rank metric.

    \item  There exists $\alpha\in (\Fq^{\sigma})^* $ that is  a common right  root of the polynomials  $\displaystyle a_i x^{n-i } - b_i \in \Fq[x;\sigma], i \in \{0,\ell\}. $ 

    \item   There exists $\alpha\in (\Fq^{\sigma})^* $ that is a right root of  $\gcrd(a_0 x^n - b_0, a_{\ell} x^{n - \ell} - b_{\ell}) \in \Fq[x;\sigma].$
       \item There exists $\alpha\in (\Fq^{\sigma})^* $ that is a right root of $\gcrd(x^n - b_0 a_0^{-1}, x^{n - \ell} - b_{\ell} a_{\ell}^{-1})\in \Fq[x;\sigma] .$ 
     \item  There exists $\alpha\in (\mathbb{F}_q^{\sigma})^* $ such that $ (a_0,a_{\ell})\star ( \alpha^n,   \alpha^{n-\ell})= (b_0,b_{\ell}) $.

      \item  $ (a_0, a_{\ell})^{-1} \star(b_0,b_{\ell})\in H_{\ell,\Fq^{\sigma}},$ where   $ H_{\ell,\Fq^{\sigma}}$ is the cyclic subgroup of $ (\Fq^{\sigma})^*\times (\Fq^{\sigma})^*$ generated by $(\zeta^n,\zeta^{n-\ell}) .$
\end{enumerate}
The equivalence between (1) and (7) implies that the number of $(n,\ell,\Fq^{\sigma})$-equivalence classes is $$ 
N_{(n,\ell,\Fq^{\sigma})}=\dfrac{ (q-1)^2}{ \lcm(\frac{q_0-1}{ \gcd(n,q_0-1)}, \frac{q_0-1}{ \gcd(n-\ell,q_0-1)})} .$$
\end{corollary}
\begin{proof}
     See Appendix.
\end{proof}

In the following result we give characterizations of rank equivalence for the class of  $(\ell,\sigma)$-trinomial codes  associated with  $x^n - a_{\ell}x^{\ell} - a_0$ to the class of $(\ell,\sigma)$-trinomial codes  associated with  $x^n - x^{\ell} - 1.$ The proof is similar to the case of Hamming equivalence (Corollary \ref{Equiv_1}), so we omit it.

\begin{corollary} \label{Caract_(1,1)}
Let $\sigma$ be and automorphism of $ \Fq,$ $\ell$ be an integer such that  $0<\ell<n$,  and  $(a_0, a_{\ell})$ be an element of $\mathbb{F}_q^{*} \times \mathbb{F}_q^{*}$. Then the following statements are equivalent.
    \begin{enumerate}
        \item  The class of skew  $(\ell,\sigma)$-trinomial codes  associated with  $x^n - a_{\ell}x^{\ell} - a_0$ is  $(n,\ell,\sigma, \mathbb{F}_{q'})$-equivalent (rank equivalence) to the class of   skew  $(\ell,\sigma)$-trinomial codes  associated with $x^n - x^{\ell} - 1.$
        \item  There exists $ \alpha \in \mathbb{F}_{q'}^*  $ such that $(a_0, a_{\ell})\star\left(N^{\sigma}_{n}(\alpha), N^{\sigma}_{n-\ell}(\sigma^{\ell}(\alpha)) \right)= (1, 1)$. 
        \item  There exists  an $(n,\sigma)$-th root  $ \alpha \in \mathbb{F}_{q'}^* $  of $a_0$  such that 
        $   a_0 =  N^{\sigma}_{\ell}(\alpha) a_\ell.$ 
        \end{enumerate}
\end{corollary}

 In the following result, we describe the associated polynomials of possible skew $(\ell,\sigma)$-trinomial codes based on the  $( n,\ell,\sigma, \mathbb{F}_{q'})$-equivalence relation defined above in Definition \ref{DefFqSigma}.
\begin{theorem}\label{RankClasses}
Let $n, \ell$ be two integers such that $0<\ell<n,$ let $\xi$ be  a primitive element of $\Fq$, $ \displaystyle{ \zeta:=\xi^{\frac{q-1}{q'-1}}}$  a primitive element of $\mathbb F_{q'}$ and $\sigma$ be an automorphism of $ \Fq.$  Set  $n':=\frac{q-1}{q'-1} ,$
  $  d:=\gcd \left(\frac{q'-1}{ d_0}, \frac{q'-1}{ d_{\ell}}\right)$  and $ d_i:=\gcd([n-i]_r,q'-1)$, for $ \ i\in \{0,\ell\}$. Then   the class of skew  $(\ell,\sigma)$-trinomial  codes associated with $x^n-a_{\ell}x^{\ell}-a_0$ is $(n,\sigma,\mathbb{F}_{q'})$-equivalent (rank equivalence) to the class of skew  $(\ell,\sigma)$-trinomial  codes associated with  $x^n-\xi^{j} x^{\ell}- \xi^{i+hn'},$ for some $i\in\{0,1,\ldots, d_0-1\}$,  $j\in \{0,1,\ldots, d_{\ell}-1\}$, and $ h\in \{0,\ldots, {n'}^2 d-1\}.$
  \end{theorem}
  \begin{proof}
      The proof here is the same as that one of Theorem \ref{Equiv_2}, the unique difference between them is in the number of equivalence classes. For the complete proof  see Appendix.
  \end{proof}

In the following proposition, we link the rank equivalence of skew trinomial codes to the one of skew constacyclic codes. The proof is again similar to the Hamming case (Proposition \ref{Eq_Consta}), details can be found in the Appendix.
\begin{proposition}\label{RankEquiv_Consta}
Let $ (a_0, a_{\ell}) $ and $ (b_0, b_{\ell}) $ be elements of $ \Fq^{*} \times \Fq^{*} $ such that $ (a_0, a_{\ell}) \sim_{(n, \ell, \mathbb{F}_{q'})} (b_0, b_{\ell}) $. Then:
\begin{enumerate}
    \item There exists $\alpha\in  \mathbb{F}_{q'}^{*} ,$ such that for each $ i \in \{0, \ell\} $, $ a_i^{-1} b_i \in
    \langle N_{n-i}^{\sigma}(\sigma^i(\alpha)) \rangle.$
    \item For each $ i \in \{0, \ell\} $, $ (a_i^{-1} b_i)^{d_i} = 1 $, where $ d_i := \dfrac{q'-1}{\gcd([n-i]_r, q'-1)} $.
    \item For each $ i \in \{0, \ell\} $, the map 
    \begin{equation*}
  	\begin{array}{cccc}
  	\varphi_{i}:& \mathbb{F}_q[x;\sigma] /\langle x^{n-i}-b_i\rangle  &\longrightarrow & \mathbb{F}_q[x;\sigma] /\langle
    x^{n-i}-b_i\rangle, \\ 
  	& f(x)& \longmapsto &  f(\sigma^i(\alpha) x)  
  	\end{array}
  	\end{equation*} 
is an $\mathbb{F}_q$-morphism isometry  with respect to the rank  distance.
    \item For each $ i \in \{0, \ell\} $ the class of skew  $(a_i,\sigma)$-constacyclic codes of length $n-i$ is { rank equivalent} to the class of skew   $(b_i,\sigma)$-constacyclic codes of length $n-i$  over $\Fq.$ 
\end{enumerate}
\end{proposition}

\begin{remark}\label{ConstacyclicFq_sigma}
In the case where we calculate the rank over the fixed subfield $ \mathbb{F}_q^{\sigma} = \mathbb{F}_{q_0}$ (i.e., we take $\alpha \in \Fq^{\sigma}$ such that $ \varphi_{\alpha}(x)= \alpha x$ is an $\Fq$-morphism isometry with respect to the rank metric (Definition \ref{DefFqSigma})), we have $ N_n(\sigma^i(\alpha))= \alpha^n$ for any $i$, since $ \alpha$ will be fixed by any power of $\sigma.$ In the case described here we simply get that $[n-i]_r$ is equal to $n-i.$ For more details see Appendix Proposition \ref{PropFq_sigma}.  
\end{remark}

\subsection{Generalization to skew polycyclic codes}
 In this subsection we generalize the results on {Hamming and rank}  equivalence to general skew  polycyclic codes. We start with the generalized definition of {Hamming} equivalence.

 
 
 \begin{definition}
Let $ \vec{a} = (a_0, a_1, \ldots, a_{n-1}) $ and $ \vec{b} = (b_0, b_1, \ldots, b_{n-1}) $ be elements in $\mathbb{F}_q^n .$
We say that $ \vec{a} $ and $ \vec{b} $ are {Hamming $ (n,\sigma)$-equivalent}, and we denote this by 
$$ \vec{a} \sim_{(n,\sigma)} \vec{b}, $$ 
if there exists \( \alpha \in \mathbb{F}_q^* \) such that the following map

  \begin{equation}
  	\begin{array}{cccc}
  	\varphi_{\alpha}:& \mathbb{F}_q[x;\sigma] /\langle x^n-\vec{b}(x) \rangle  &\longrightarrow & \mathbb{F}_q[x;\sigma] /\langle x^n-\vec{a}(x) \rangle, \\ 
  	& f(x) & \longmapsto &  f(\alpha x),
  	\end{array}
  	\end{equation} 
is an $\mathbb{F}_q$-morphism isometry with respect to the  Hamming distance.    Moreover, we can easily verify that ``$ \sim_{(n,\sigma)}$'' is an equivalence relation.
 \end{definition}

\begin{remark}
 \begin{enumerate}
     \item  If $ \vec{a} = (\lambda, 0, \ldots, 0) $ and $ \vec{b} = (\mu, 0, \ldots, 0) $, then we recover the case of $(n,\sigma)$-equivalence for skew constacyclic codes studied in \cite{Ouazzou2025}.
     \item  If $ \vec{a} = (a_0, 0, \ldots,a_{\ell},0,\ldots, 0) $ and  $ \vec{b} = (b_0, 0, \ldots,b_{\ell},0,\ldots, 0) $, then we recover the case of $(n,\ell,\sigma)$-equivalence for  skew  $(\ell,\sigma)$-trinomial  codes studied in Section 3.
     \item If $\sigma=\id$, we recover  the case of $n$-equivalence for   polycyclic   codes studied in \cite{Equiv2025}.
     \end{enumerate}
 \end{remark}

 In the following lemma, we show that this notion of equivalence automatically implies that the vectors $\vec{a}$ and $\vec{b}$ have zero entries in exactly the same positions. This implies that skew  $(\ell,\sigma)$-trinomial code families can only be equivalent to other skew $(\ell,\sigma)$-trinomial code families.
 
 \begin{lemma}
    Let $ \vec{a} = (a_0, a_1, \ldots, a_{n-1}) $ and $ \vec{b} = (b_0, b_1, \ldots, b_{n-1}) $ be elements in $\mathbb{F}_q^n $ such that  $ \vec{a} \sim_{(n,\sigma)}\vec{b} $. Then $ a_i \neq 0 $ if and only if $ b_i \neq 0 $, for any $ 0\leq  i\leq n-1,$  and so $\vec{a}$ and $\vec{b}$ have the same Hamming weight.
 \end{lemma}
\begin{proof}
Suppose that $ \vec{a} \sim_n \vec{b},$ then there is $\alpha \in \Fq^*$ such that $ \varphi_{\alpha}$ is an $\mathbb{F}_q$-morphism isometry between $ \Fq[x;\sigma] /\langle x^n-\vec{b}(x) \rangle $  and $ \Fq[x;\sigma] /\langle x^n-\vec{a}(x) \rangle$. Then, as in the proof of [Theorem \ref{Th.1}, (1) $\Rightarrow$ (2)], we obtain that 
$$ b_i = N_{n-i}^{\sigma}(\alpha) a_i, \quad \forall i=0, \ldots, n-1.$$
Hence the result holds.
    
\end{proof}

We now generalize Theorem \ref{Th.1} to the general skew polycyclic case: 

\begin{theorem}\label{Th_general}
Let $ \vec{a}=(a_0,a_1,\ldots, a_{n-1}), \vec{b}=(b_0,b_1,\ldots, b_{n-1})\in \mathbb{F}_q^{n}$ have non-zero entries in the same $m$ positions, i.e., $a_{i_j}$ and $b_{i_j}$ are non-zero  for $0\leq i_0<\dots<i_{m-1}\leq n-1$. Moreover, let $\xi $ be a primitive element of $\Fq$. Then the following statements are equivalent:
\begin{enumerate}
\item $ \vec{a} \sim_{(n,\sigma)} \vec{b}$,  i.e., $ \vec{a} $ and  $ \vec{b}$ are  Hamming $ (n,\sigma)$-equivalent. 

\item   There exists $\alpha\in \Fq^*$ such that 
 $$  a_{i_j} N^{\sigma}_{n-i_j}(\sigma^{i_j}(\alpha)) = b_{i_j}, \   \text{for  each $ j \in \{0,1,\ldots,m-1\}$,}$$   where the  $ a_{i_j}$'s and $  b_{i_j}$'s are the non-zero components of $\vec{a}$ and $\vec{b}.$
 
 \item  There exists $\alpha\in \Fq^*$ such that 
 $$ (a_{i_0}, a_{i_1},\ldots, a_{i_{m-1}}) \left( N^{\sigma}_{n-i_0}(\sigma^{i_0}(\alpha)), N^{\sigma}_{n-i_1}(\sigma^{i_1}(\alpha)) , \ldots, N^{\sigma}_{n-i_{m-1}}(\sigma^{i_{m-1}}(\alpha))\right)=  ( b_{i_0}, b_{i_1},\ldots, b_{i_{m-1}}),$$
 where $m$ is the number of the non-zero entries  of $\vec{a}.$
\item $ (a_{i_0}, a_{i_1},\ldots, a_{i_{m-1}})^{-1}\star ( b_{i_0}, b_{i_1},\ldots, b_{i_{m-1}}) \in H,$ where $H$ is the cyclic subgroup of $ (\Fq^*)^m$ generated by 
$  \left( N^{\sigma}_{n-i_0}(\sigma^{i_0}(\xi)), N^{\sigma}_{n-i_1}(\sigma^{i_1}(\xi)) , \ldots, N^{\sigma}_{n-i_{m-1}}(\sigma^{i_{m-1}}(\xi))\right),$  and   $j=0,\ldots,m-1.$
\end{enumerate}
In particular,the number of $(n,\sigma)$-equivalence classes is
$$ N_{(n,\sigma)}=\dfrac{ (q-1)^m}{ \lcm_{ j\in \{i_0,i_1,\ldots, i_{m-1} \} }\left(\frac{q-1}{ \gcd([n-j]_r,q-1)}\right)}. $$
\end{theorem}

\begin{proof}\;
\begin{itemize}
    \item[(1) $\Rightarrow$ (2):] 
    Suppose that $ \vec{a} \sim_{(n,\sigma)} \vec{b} $, then there is $\alpha \in \mathbb{F}_q^*$ such that

$$\varphi_{\alpha} : \Fq[x;\sigma] /\langle x^n - \vec{b}(x) \rangle \to \Fq[x;\sigma] /\langle x^n - \vec{a}(x) \rangle, \quad f(x) \mapsto f(\alpha x) $$
is an $\mathbb{F}_q$-vector space isometry.  It follows that

$$ 
\varphi_{\alpha}(x^{k}) = (\alpha x)^k = N^{\sigma}_k(\alpha) x^{k} \mod (x^n - \vec{a}(x)), \quad \forall  \ k = 0, 1, \ldots, n-1.
$$

As $\varphi_{\alpha}$ is an $\mathbb{F}_q$-algebra isometry and $ \varphi(x^n - \vec{b}(x)) = 0 \mod (x^n - \vec{a}(x))$, we have
\begin{equation}
\varphi_{\alpha}(x^n) =\varphi_{\alpha} (\vec{b}(x)) = b_0 + N_1^{\sigma}(\alpha) b_1 x + \ldots + N_{n-1}^{\sigma}(\alpha)b_{n-1} x^{n-1} \mod (x^n - \vec{a}(x)).
\end{equation}

On the one hand,
\begin{equation}
\varphi(x^n) = N^{\sigma}_n(\alpha) x^n = N^{\sigma}_n(\alpha) ( a_0 + a_1 x + \ldots + a_{n-1} x^{n-1}) \mod (x^n - \vec{a}(x)).
\end{equation}

Comparing term by term, we deduce that   
$$ a_i N_{n}^{\sigma}(\alpha) = N_i^{\sigma}(\alpha)b_i, \ \text{ for any }  i \in \{0, 1, \ldots, n-1\},$$ 
and so 
$$ a_i N_{n}^{\sigma}(\alpha)  N_i^{\sigma}(\alpha^{-1}) =b_i, \ \text{ for any }  i \in \{0, 1, \ldots, n-1\}.$$ 
On the other hand, for any $  i \in \{0, 1, \ldots, n-1\} , $ we have 
$$ \begin{array}{rl}
    N_n^{\sigma}(\alpha) N_{i}^{\sigma}(\alpha^{-1})  & =  N_{i+(n-i)}^{\sigma}(\alpha) N_{i}^{\sigma}(\alpha^{-1})\\
     &=   \sigma^{i}(N_{n-i}^{\sigma}(\alpha)) N^{\sigma}_{i}(\alpha) N_{i}^{\sigma}(\alpha^{-1})  \\
     &=  \sigma^{i}(N^{\sigma}_{n-i}(\alpha))\\
     & = N^{\sigma}_{n-i}(\sigma^{i}(\alpha)).
\end{array}
$$

Hence,
$$ a_i N_{n-i}^{\sigma}( \sigma^i(\alpha))  =b_i, \ \text{ for any }  i \in \{0, 1, \ldots, n-1\}.$$ 

As  $ a_{i_j}$'s and $  b_{i_j}$'s are the non-zeros components of $\vec{a}$ and $\vec{b},$  we have
$$ a_{i_j} N_{n-{i_j}}^{\sigma}(\sigma^{i_j}(\alpha)) = b_{i_j},\  \ j=0,\ldots,m-1.$$

 \item[(2) $\Rightarrow$ (3):]  is immediate. 

 \item[(3) $\Rightarrow$ (4):]
Suppose that there is $\alpha\in \Fq^{*}$ such that 
$$  (b_{i_0}, b_{i_1},\ldots, b_{i_{m-1}} )=  \left( N_{n-i_0}^{\sigma}(\sigma^{i_0}(\alpha)), 
N^{\sigma}_{n-i_1}(\sigma^{i_1}(\alpha)), \ldots, N^{\sigma}_{n-i_{m-1}}(\sigma^{i_{m-1}}(\alpha))\right)\star ( a_{i_0},  a_{i_1}, \ldots, a_{i_{m-1}}).$$ 
For $\alpha= \xi^h,$ we obtain 
$$  ( a_{i_0}, a_{i_1},\ldots, a_{i_{m-1}})^{-1}\star (b_{i_0}, b_{i_1},\ldots, b_{i_{m-1}} ) =   \left( N_{n-i_0}^{\sigma}(\sigma^{i_0}(\xi)), 
N^{\sigma}_{n-i_1}(\sigma^{i_1}(\xi)), \ldots, N^{\sigma}_{n-i_{m-1}}(\sigma^{i_{m-1}}(\xi))\right)^h .$$
It follows that $  ( a_{i_0},  a_{i_1}, \ldots, a_{i_{m-1}})^{-1}\star ( b_{i_0},  b_{i_1}, \ldots, b_{i_{m-1}}) $ belongs to the cyclic subgroup  $H$ of $ (\Fq^{*})^m$ generated by $ \left( N_{n-i_0}^{\sigma}(\sigma^{i_0}(\xi)), 
N^{\sigma}_{n-i_1}(\sigma^{i_1}(\xi)), \ldots, N^{\sigma}_{n-i_{m-1}}(\sigma^{i_{m-1}}(\xi))\right).$
   
    \item[(4) $\Rightarrow$ (1):]  
     Suppose that $ (a^{n-{i_0}}, a^{n-{i_1}} , \ldots, a^{n-i_{m-1}})^{-1}\star (b^{n-{i_0}}, b^{n-{i_1}} , \ldots, b^{n-i_{m-1}}) $ is an element of the cyclic subgroup  $H$ of $ (\Fq^{*})^m$ generated by
$$ \left( N_{n-i_0}^{\sigma}(\sigma^{i_0}(\xi)), 
N^{\sigma}_{n-i_1}(\sigma^{i_1}(\xi)), \ldots, N^{\sigma}_{n-i_{m-1}}(\sigma^{i_{m-1}}(\xi))\right) .$$
Then there exists an integer $ h$ such that 
     $$  ( a_{i_0},  a_{i_1}, \ldots, a_{i_{m-1}})^{-1}\star ( b_{i_0},  b_{i_1}, \ldots, b_{i_{m-1}}) =  \left( N_{n-i_0}^{\sigma}(\sigma^{i_0}(\xi)), 
N^{\sigma}_{n-i_1}(\sigma^{i_1}(\xi)), \ldots, N^{\sigma}_{n-i_{m-1}}(\sigma^{i_{m-1}}(\xi))\right)^h $$
    For $ \alpha= \xi^{h}, $ we obtain  that $  a_j N_{n-j}^{\sigma}(\sigma^j(\alpha))= b_j, $ for any $ j\in \{ i_0,i_1,\ldots, i_{m-1}\}.$
As in the proof of Theorem \ref{Th.1}, we verify that  $ \varphi_{\alpha}, $ defined as  
  \begin{equation}
        \begin{array}{cccc}
        \varphi_{\alpha} : & \Fq[x;\sigma]/\langle x^n - \vec{b}(x) \rangle,  & \longrightarrow & \Fq[x;\sigma] /\langle x^n - \vec{a}(x) \rangle, \\ 
        & f(x) & \longmapsto & f(\alpha x),
        \end{array}
    \end{equation} 
    is an $\Fq$-vector space  isometry with respect to the Hamming distance.
To verify that $\varphi_{\alpha}$ is an $\mathbb{F}_q$-morphism, the key idea is to use the fact that  
$$ 
a_{i_k} N^{\sigma}_{n-i_k}(\sigma^{i_k}(\alpha)) = b_{i_k}, \quad \forall k \in \{0,1,\ldots,m-1\}, \quad \text{and} \quad N_{i+j}^\sigma(\alpha) = N_i^\sigma(\alpha)\sigma^i(N_j^\sigma(\alpha)),
$$  
to  obtain 
$$ 
N_i^\sigma(\alpha)\sigma^i(N_{n-i+j}^\sigma(\alpha)) \sigma^j(a_{k}) = N^{\sigma}_{k+j}(\alpha) \sigma^{j}(b_{k}).
$$  

From this, we can conclude that  
$$ 
\varphi_{\alpha}(f(x)g(x)) = \varphi_{\alpha}(f(x))\varphi_{\alpha}(g(x)).
$$  

We gave full details of the computation in the proof of Theorem \ref{Th.1},  so we do  not repeat the same process here.
\end{itemize}

\noindent By the equivalence between (1) and (4), we deduce that the number of $(n,\sigma)$-equivalence classes  on $(\Fq^*)^m$ corresponds to the order of the group $ (\Fq^*)^m/H,$ which equals 
$$ N_{(n,\sigma)}=\dfrac{ (q-1)^m}{ \lcm_{ j\in \{i_0,i_1,\ldots, i_{m-1} \} }\left(\frac{q-1}{ \gcd(p^{r j}[n-j]_r,q-1)}\right)}=\dfrac{ (q-1)^m}{ \lcm_{ j\in \{i_0,i_1,\ldots, i_{m-1} \} }\left(\frac{q-1}{ \gcd([n-j]_r,q-1)}\right)}. $$

\end{proof}

Similarly to the case of skew $(\ell,\sigma)$-trinomial codes, Theorem \ref{Th_general} implies the following results regarding the equivalence of skew polycyclic codes. The proofs are analogous to the skew $(\ell,\sigma)$-trinomial case. 

\begin{corollary}\label{Equiv_General2}
As before let $ \vec{a}=(a_0,a_1,\ldots, a_{n-1})$ and $ \vec{b}=(b_0,b_1,\ldots, b_{n-1})$ be elements of $\mathbb{F}_q^{n}$  
of the same weight $m$  and denote by $ a_{i_j}$ and $  b_{i_j}$  the non-zero components of $\vec{a}$ and $\vec{b}$. 
\begin{enumerate}
    \item     The class of skew polycyclic  codes associated with $ x^n-\dsum{j=0}{m-1} a_{i_j} x^{i_j} $ is $(n,\sigma)$-equivalent (Hamming equivalence) to the class of  skew polycyclic  codes associated with the vector  $x^n-\dsum{j=0}{m-1}  x^{i_j}$  if and only if there exists $\alpha \in \Fq^*$ such that
    $$ ( a_{i_0}, a_{i_1},\ldots, a_{i_{m-1}}) 
    \left( N_{n-i_0}^{\sigma}(\sigma^{i_0}(\alpha)),  N^{\sigma}_{n-i_1}(\sigma^{i_1}(\alpha)), \ldots, N^{\sigma}_{n-i_{m-1}}(\sigma^{i_{m-1}}(\alpha))\right) = (1, 1,\ldots, 1 ).$$

     \item    Let  $ d:= \gcd_{ j\in \{i_0,i_1,\ldots, i_{m-1} \} }\left(\frac{q-1}{ \gcd([n-j]_r,q-1)}\right) $, then the  class of  skew polycyclic  codes associated with  $ x^n-\dsum{j=0}{m-1} a_{i_j} x^{i_j} $ is $(n,\sigma)$equivalent (Hamming equivalence) to the class of  skew polycyclic  codes associated with the polynomial  $ x^n-\dsum{j=1}{m-1} \xi^{k_j} x^{i_j}- \xi^{k_0+h[n]_r} ,$ with $k_j=0,1,\ldots, \gcd(n-i_j,q-1) , \ \ h=0,\ldots, d-1.$ 
\end{enumerate}
\end{corollary} 

\hspace{0.3cm}

Now we return to the case of rank equivalence, and we start by the following definition. 
\begin{definition} Let $ \vec{a} = (a_0, a_1, \ldots, a_{n-1}) $ and $ \vec{b} = (b_0, b_1, \ldots, b_{n-1}) $ be elements in $\mathbb{F}_q^n .$ We say that $\vec{a}$ and $\vec{b}$ are {$( n,\sigma, \mathbb{F}_{q'})$-equivalent} in $\mathbb{F}_q^* \times \mathbb{F}_q^*$ with respect to the rank metric, and we  denote
$$ \vec{a} \sim_{( n,\sigma, \mathbb{F}_{q'})} \vec{b},$$
if there exists an $\alpha \in \mathbb{F}_{q'}^* $ such that the map
  \begin{equation}
  	\begin{array}{cccc}
  	\varphi_{\alpha}:& \mathbb{F}_q[x;\sigma] /\langle x^n-\vec{b}(x) \rangle  &\longrightarrow & \mathbb{F}_q[x;\sigma] /\langle x^n- x^n-\vec{a}(x) \rangle, \\ 
  	& f(x) & \longmapsto &  f(\alpha x),
  	\end{array}
  	\end{equation} 
is an  $\mathbb{F}_q$-morphism isometry with respect to the rank metric. 
\end{definition}

\begin{theorem}
Let $ \vec{a}=(a_0,a_1,\ldots, a_{n-1}), \vec{b}=(b_0,b_1,\ldots, b_{n-1})\in \mathbb{F}_q^{n}$ have non-zero entries in the same $m$ positions, i.e., $a_{i_j}$ and $b_{i_j}$ are non-zero  for $0\leq i_0<\dots<i_{m-1}\leq n-1$. Moreover, let $ \mathbb{F}_{q'}$ be a subfield of $\Fq$ and 
$\zeta $ its primitive element. Then the following statements are equivalent:

\begin{enumerate}
    \item $ \vec{a} \sim_{(n,\sigma,\mathbb{F}_{q'})} \vec{b} $ i.e, $\vec{a} $ and $ \vec{b}$ are rank equivalent.
    \item There exists an $\alpha \in \mathbb{F}_{q'}^*$ such that the map

  \begin{equation}
  	\begin{array}{cccc}
  	\varphi_{\alpha}:& \mathbb{F}_q[x;\sigma] /\langle x^n-\vec{b}(x) \rangle  &\longrightarrow & \mathbb{F}_q[x;\sigma] /\langle x^n-\vec{a}(x) \rangle, \\ 
  	& f(x) & \longmapsto &  f(\alpha x),
  	\end{array}
  	\end{equation} 
    
is an  $\mathbb{F}_q$-morphism isometry.

     \item  There exists $\alpha\in \mathbb{F}_{q'}^* $ such that 
     $$ (a_{i_0}, a_{i_1},\ldots, a_{i_{m-1}})\star ( N_{n-i_0}^{\sigma}(\sigma^{i_0}(\alpha)), N_{n-i_1}^{\sigma}(\sigma^{i_1}(\alpha)),\ldots, N_{n-i_{m-1}}^{\sigma}(\sigma^{i_{m-1}}(\alpha)) ) =   ( b_{i_0}, b_{i_1},\ldots, b_{i_{m-1}}).$$
     
\item  $ (a_{i_0}, a_{i_1},\ldots, a_{i_{m-1}})^{-1}\star( b_{i_0}, b_{i_1},\ldots, b_{i_{m-1}}) \in H_{n, \mathbb{F}_{q'}},$ where $H_{n,\mathbb{F}_{q'}}$ is the cyclic subgroup of $\mathbb{F}_{q'}^{* m}$ generated by 
$( N_{n-i_0}^{\sigma}(\sigma^{i_0}(\zeta)), N_{n-i_1}^{\sigma}(\sigma^{i_1}(\zeta)),\ldots, N_{n-i_{m-1}}^{\sigma}(\sigma^{i_{m-1}}(\zeta)) ).$
\end{enumerate}
The equivalence between (1) and (4) implies that the number of $(n,\sigma,\mathbb{F}_{q'})$-equivalence classes is 
$$ 
N_{(n,\sigma,\mathbb{F}_{q'})}= \dfrac{ (q-1)^m}{ \lcm_{ j\in \{i_0,i_1,\ldots, i_{m-1} \} }\left(\frac{q'-1}{ \gcd([n-j]_r,q'-1)}\right)}. $$ 
\end{theorem}

When the rank is calculated over the fixed subfield $\Fq^{\sigma}$, we obtain additional characterizations of the rank equivalence equivalence.
\begin{corollary}
Let $ \vec{a}=(a_0,a_1,\ldots, a_{n-1}), \vec{b}=(b_0,b_1,\ldots, b_{n-1})\in \mathbb{F}_q^{n}$ have non-zero entries in the same $m$ positions, i.e., $a_{i_j}$ and $b_{i_j}$ are non-zero  for $0\leq i_0<\dots<i_{m-1}\leq n-1$. Moreover, let $\zeta $ be a primitive element of $\Fq^{\sigma}$. Then the following statements are equivalent:

\begin{enumerate}
    \item $ \vec{a} \sim_{( n,\Fq^{\sigma})} \vec{b} .$
    \item There exists an $\alpha \in (\mathbb{F}_q^{\sigma})^*$ such that the map

  \begin{equation}
  	\begin{array}{cccc}
  	\varphi_{\alpha}:& \mathbb{F}_q[x;\sigma] /\langle x^n-\vec{b}(x) \rangle  &\longrightarrow & \mathbb{F}_q[x;\sigma] /\langle x^n-\vec{a}(x) \rangle, \\ 
  	& f(x) & \longmapsto &  f(\alpha x),
  	\end{array}
  	\end{equation} 
    
is an  $\mathbb{F}_q$-morphism isometry with respect to the rank metric.

    \item  There exists $\alpha\in (\Fq^{\sigma})^* $ that is  a common right  root of the polynomials  $\displaystyle a_i x^{n-i } - b_i \in \Fq[x;\sigma], i \in \{0,\ldots,m-1\}. $ 

    \item   There exists $\alpha\in (\Fq^{\sigma})^* $ that is a right root of  $\gcrd_{\{0 \leq j \leq m-1\}}( a_{i} x^{n - i} - b_{i}) \in \Fq[x;\sigma].$
    
     \item  There exists $\alpha\in (\mathbb{F}_q^{\sigma})^* $ such that 
     $$ (a_{i_0}, a_{i_1},\ldots, a_{i_{m-1}})\star(\alpha^{n-{i_0}}, \alpha^{n-{i_1}} , \ldots, \alpha^{n-{i_{m-1}}})=  ( b_{i_0}, b_{i_1},\ldots, b_{i_{m-1}}).$$
\item  $ (a_{i_0}, a_{i_1},\ldots, a_{i_{m-1}})^{-1}\star( b_{i_0}, b_{i_1},\ldots, b_{i_{m-1}}) \in H_{n,\Fq^{\sigma}},$ where $H_{n,\Fq^{\sigma}}$ is the cyclic subgroup of $ (\Fq^{\sigma})^{* m}$ generated by 
$ ( N_{n-i_0}^{\sigma}(\sigma^{i_0}(\zeta)), N_{n-i_1}^{\sigma}(\sigma^{i_1}(\zeta)),\ldots, N_{n-i_{m-1}}^{\sigma}(\sigma^{i_{m-1}}(\zeta)) ) .$
\end{enumerate}
The equivalence between (1) and (7) implies that the number of $(n,\Fq^{\sigma})$-equivalence classes is 
$$ 
N_{(n,\Fq^{\sigma})}= \dfrac{ (q-1)^m}{ \lcm_{ j\in \{i_0,i_1,\ldots, i_{m-1} \} }\left(\frac{q_0-1}{ \gcd(n-j,q_0-1)}\right)}. $$ 
\end{corollary}

Similarly to  the case of skew $(\ell,\sigma)$-trinomial codes, Theorem \ref{Th_general} implies the following results regarding the rank equivalence of skew polycyclic codes. The proofs are analogous to the skew $(\ell,\sigma)$-trinomial case in Theorem \ref{RankClasses}.

\begin{corollary}
As before let $ \vec{a}=(a_0,a_1,\ldots, a_{n-1})$ and $ \vec{b}=(b_0,b_1,\ldots, b_{n-1})$ be elements of $\mathbb{F}_q^{n}$  
of the same weight $m$  and denote by $ a_{i_j}$ and $  b_{i_j}$  the non-zeros components of $\vec{a}$ and $\vec{b}$. Let $ \mathbb{F}_{q'}$ be a subfield of $\Fq, $ and $n':=\frac{q-1}{q'-1} .$  Then,
\begin{enumerate}
    \item    The class of skew polycyclic  codes associated with   $ x^n-\dsum{j=0}{m-1} a_{i_j} x^{i_j} $ is $(n,\sigma,\mathbb{F}_{q'})$-equivalent (rank equivalence) to the class of  skew polycyclic  codes associated with the vector  $x^n-\dsum{j=0}{m-1}  x^{i_j}$  if and only if there exists $\alpha \in \mathbb{F}_{q'}^*$ such that
    $$ ( a_{i_0}, a_{i_1},\ldots, a_{i_{m-1}})\star( N_{n-i_0}^{\sigma}(\sigma^{i_0}(\alpha)), N_{n-i_1}^{\sigma}(\sigma^{i_1}(\alpha)),\ldots, N_{n-i_{m-1}}^{\sigma}(\sigma^{i_{m-1}}(\alpha)) ) = (1, 1,\ldots, 1 ).$$  
        
     \item  Let   $ d:= \gcd_{ j\in \{i_0,i_1,\ldots, i_{m-1} \} }\left(\frac{q'-1}{ \gcd([n-j]_r,q'-1)}\right)$, then the  class of  skew polycyclic  codes associated with  $ x^n-\dsum{j=1}{m-1} a_{i_j} x^{i_j}-a_0 $ is  $(n,\sigma,\mathbb{F}_{q'})$-equivalent (rank equivalence) to the class of  skew polycyclic  codes associated with the polynomial  $ x^n-\dsum{j=1}{m-1} \xi^{k_j} x^{i_j}- \xi^{k_0+h n'} ,$ with $k_j=0,1,\ldots, \gcd(n-i_j,q_0-1) , \ \ h=0,\ldots, n'^m d-1.$ 
\end{enumerate}
\end{corollary}

\section{Examples for code equivalence}

In this section, we provide some examples related to the Hamming and rank equivalence studied in the previous section. In particular, we characterize and analyze the possible number of Hamming equivalence classes and the representative element of each class, for skew  trinomial  codes of length $n$  over $\mathbb{F}_4$, $\mathbb{F}_8$, and $\mathbb{F}_9$. For rank equivalence, we characterize the number and the representative elements of each equivalence class in the case of trinomial codes of length $n$ over $\mathbb{F}_{2^s}$. Similar methods can be used to construct more examples for different alphabets.
 
\subsection{Hamming metric}

\begin{example}[Skew $(\ell,\sigma)$-trinomial  codes  over $\mathbb{F}_4$] 

Let $q = 4$, and $\xi$ be a primitive element of $\mathbb{F}_4$. In this example, we are interested in skew $(\ell,\sigma)$-trinomials of length $n$ over $\mathbb{F}_4$. Since the degree of the extension field $\mathbb{F}_4/\mathbb{F}_2$ is $2$, we have $\text{Aut}(\mathbb{F}_4) = \{\text{id}, \tau\}$, where $\tau$ is the Frobenius automorphism of $\mathbb{F}_4$, i.e., $\tau(a) = a^2$ for all $a\in \mathbb{F}_4$. 

For $\sigma=\tau,$ we consider skew $(\ell,\sigma)$-trinomial  codes   associated with  a skew polynomial of the form $f(x) =x^{n}-a_{\ell}x^{\ell}-a_0\in \mathbb{F}_4[x,\sigma],$
 where $\ell$ is an integer such that $0<\ell <n$.  Let $\vec{a}(x):=   a_{\ell}x^{\ell}+a_0 ,$  the Hamming weight of $ \vec{a}(x) $ is  $2.$
 As $[i]_1= \dfrac{2^i-1}{2-1}=2^i-1,$  by  Theorem  \ref{Th.1}, the number of Hamming $(n,\ell, \sigma)$-equivalence classes is given by  
 
 $$N:=\dfrac{3^2}{ \lcm\left(\dfrac{3}{\gcd( [n]_1 ,3 )}, \dfrac{3}{\gcd( [n-\ell]_1,3)}  \right)}    
      =\dfrac{9}{ \lcm\left( \frac{3}{\gcd( 2^{n}-1 ,  2^2-1 )}, \frac{3}{\gcd( 2^{n-\ell}-1 ,  2^2-1 )}\right) }  
 $$ 
 Now we analyze the value of $N$  based on the parity of $n$ and $\ell$.
 \begin{enumerate}
\item \textbf{If $n$ and $\ell$ are odd,} then $ \gcd( 2^{n}-1 ,  2^2-1 )=2^{\gcd(n,2)}-1=1 $ and  $ \gcd( 2^{n-\ell}-1 ,  2^2-1 )= 3.$ 
Then the number of $(n,\sigma)-$equivalence classes is $N=3,$ and by Theorem \ref{Equiv_2}, we have 
     $$d_0=\gcd( [n]_1, 3)= 1,\ \ d_{\ell}= \gcd([n-\ell]_1, 3)= 3 \text{ and} \ \  d=\gcd\left( \frac{3}{d_0}, \frac{3}{d_{\ell}}\right)=1.$$
 It follows that each  skew $(\ell,\sigma)$-trinomial  code associated with a skew  polynomial $f$ of the form  $f =x^{n}- a_{\ell} x^{\ell}-a_0$ is equivalent to a skew $(\ell,\sigma)$-trinomial  code  associated with one of the following  skew polynomials: 
 $$f_{0}= x^{n} - x^{\ell} - 1 , \quad f_{1}= x^{n} - \xi x^{\ell} - 1 \quad  f_{2}= x^{n} - \xi^2 x^{\ell} - 1 .  $$
 
     \item \underline{\textbf{If $n$ and $\ell$ are even,}} 
     then $ \gcd( 2^{n}-1 ,  2^2-1 )=2^{\gcd(n,2)}-1=3 $ and  $ \gcd( 2^{n-\ell}-1 ,  2^2-1 )= 3.$ 
Then the number of $(n,\sigma)-$equivalence classes is $N=9.$ In this case the $(n,\sigma)$-equivalence relation has no advantage in the sense that equivalence classes are singleton sets and the number of polynomials to be considered to construct skew polycyclic codes is not reduced.  

  \item \underline{\textbf{If $n$  is odd and $\ell$ is even,}}   then
  $ \gcd( 2^{n}-1 ,  2^2-1 )=2^{\gcd(n,2)}-1=1 $ and
  $ \gcd( 2^{n-\ell}-1 ,  2^2-1 )= 1.$  Then the number of $(n,\sigma)-$equivalence classes is $N=3,$ 
and by Theorem \ref{Equiv_2}, we have 
     $$d_0=\gcd( [n]_1, 3)= 1,\ \ d_{\ell}= \gcd([n-\ell]_1, 3)= 1 \text{ and} \ \  d=\gcd\left( \frac{3}{d_0}, \frac{3}{d_{\ell}}\right)=3$$
    
 It follows that each skew $(\ell,\sigma)$-trinomial  code associated with a skew  polynomial $f$ of the form  $f =x^{n}- a_{\ell} x^{\ell}-a_0,$ is equivalent to a skew $(\ell,\sigma)$-trinomial  code  associated with one of the following  skew polynomials: 
$$f_{0}= x^{n} - x^{\ell} - 1 , \quad f_{1}= x^{n} -  x^{\ell} -   \xi^{[n]_1},  \quad  f_{2}= x^{n} -   x^{\ell} -  \xi^{2[n]_1}.  $$

 \item \underline{\textbf{If $n$  is even and  $\ell$ is odd,}} then,  as above, we obtain that  each skew $(\ell,\sigma)$-trinomial  code associated with a skew  polynomial $f$ of the form  $f =x^{n}- a_{\ell} x^{\ell}-a_0,$ is equivalent to a skew $(\ell,\sigma)$-trinomial  code  associated with one of the following  skew polynomials:
 $$f_{0}= x^{n} - x^{\ell} - 1 , \quad f_{1}= x^{n} - x^{\ell} - \xi,  \quad  f_{2}= x^{n} - x^{\ell} - \xi^2 .  $$
   
 \end{enumerate}
 \end{example}

\begin{example}[Skew $(\ell,\sigma)$-trinomial  codes  over $\mathbb{F}_{8}$] 

Let $q = 8$, and $\xi$ be a primitive element of $\mathbb{F}_8$. In this example, we are interested in skew $(\ell,\sigma)$-trinomial codes of length $n$ over $\mathbb{F}_8$. Since the degree of the extension field $\mathbb{F}_{8}/\mathbb{F}_2$ is 3, we have $\text{Aut}(\mathbb{F}_{8}) = \{\text{id}, \tau,  \tau^{2}\}$, where $\tau$ is the Frobenius automorphism of $\mathbb{F}_{8}$, i.e., $\tau(a) = a^2$ for all $a\in \mathbb{F}_{8}$. Then we consider a number of cases:
\begin{enumerate}
    \item \textbf{ Case 1: }  For $ \sigma=\tau,$ we consider skew $(\ell,\sigma)$-trinomial  codes     associated with  a skew polynomial of the form $f(x) =x^{n}-a_{\ell}x^{\ell}-a_0\in \mathbb{F}_{8}[x,\sigma],$ where $\ell $ is an integer such that $0<\ell <n$.  Let $\vec{a}(x):=   a_{\ell}x^{\ell}+a_0$  with Hamming weight   $2.$
As $ [i]_1=\dfrac{2^{i}-1}{2-1}=2^i-1$ is odd, by Theorem  \ref{Th.1}, the number of  Hamming $(n,\ell, \sigma)$-equivalence classes is given by  
 
 $$N:=\dfrac{7^2}{ \lcm\left(\frac{7}{\gcd( [n]_1 ,7 )}, \frac{7}{\gcd( [n-\ell]_1,7)}  \right)}    
      =\dfrac{49}{ \lcm\left( \frac{7}{\gcd( 2^n-1 ,2^3-1 )}, \frac{7}{ \gcd( 2^{n-\ell}-1 ,2^3-1 )}\right)} .
 $$

 \begin{enumerate}
\item \underline{\textbf{If $ \gcd(n, 3)=\gcd(n-\ell, 3)=1$,}}  then the number of $(n,\sigma)$-equivalence classes is   $  N=7.$ By Theorem \ref{Equiv_2}, we have 
     $$  d_0= \gcd([n]_1, 3)=1, \ \ \ d_{\ell}= \gcd( [n-\ell]_1, 7)= 1 \text{ and} \ \  d=\gcd\left( \frac{ 7}{d_0}, \frac{ 7}{d_{\ell}}\right)= 7.$$
 It follows that each  skew $(\ell,\sigma)$-trinomial  code associated with a skew  polynomial $f$ of the form  $f =x^{n}- a_{\ell} x^{\ell}-a_0,$ is equivalent to a  skew $(\ell,\sigma)$-trinomial  code  associated with one of the following  skew polynomials: 
  $$f_{i}(x)= x^{n} - x^{\ell} - \xi^{ i [n]_1 }  \ \ \text{for } i\in \{0,\ldots, 6\}.  $$ 
 
     \item \underline{\textbf{If  $  \gcd(n, 3)= \gcd(n-\ell, 3)=3 $,}} then the number of $(n,\sigma)$-equivalence classes is   $  N=49,$ and the $(n,\sigma)$-equivalence relation has no advantage.

   \item \underline{\textbf{If  $  \gcd(n, 3)=1$  and  $\gcd(n-\ell, 3)=3 $,}}  then  the number of equivalence classes is  $N=7.$ 
   In this case, $ d_{0}= \gcd([n]_1, 7)= 1,\ \ d_{\ell}= \gcd( [n-\ell]_1, 7)= 7 .$ 
   It follows that each  skew $(\ell,\sigma)$-trinomial  code associated with a skew  polynomial $f$ of the form  $f =x^{n}- a_{\ell} x^{\ell}-a_0$ is equivalent to a  skew $(\ell,\sigma)$-trinomial  code  associated with one of the following  skew polynomials: 
  $$f_{i}(x)= x^{n} - \xi^i x^{\ell} - 1 \ \ \text{for } i\in \{0,\ldots, 6\}.  $$ 
  \item  \textbf{\underline{If  $  \gcd(n, 3)=3$  and  $\gcd(n-\ell, 3)=1 $, } } then, as above, we obtain that each 
  each  skew $(\ell,\sigma)$-trinomial  code associated with a skew  polynomial $f$ of the form  $f =x^{n}- a_{\ell} x^{\ell}-a_0,$ is equivalent to a  skew $(\ell,\sigma)$-trinomial  code  associated with one of the following  skew polynomials: 
  $$f_{i}(x)= x^{n} -  x^{\ell} - \xi^i \ \ \text{for } i\in \{0,\ldots, 6\}.  $$ 

    \end{enumerate}

\item \textbf{Case 2:} For $\sigma=\tau^2,$  we consider skew $(\ell,\sigma)$-trinomial   codes   associated with  a skew polynomial of the form $f(x) =x^{n}-a_{\ell}x^{\ell}-a_0\in \mathbb{F}_{8}[x,\sigma]$, where $\ell $ is an integer such that $0<\ell <n$.  Let $\vec{a}(x):=   a_{\ell}x^{\ell}+a_0 $  with  Hamming weight 2.   By Theorem  \ref{Th.1}, the number of Hamming $(n,\ell, \sigma)$-equivalence classes is given by  
 $$N:=\dfrac{7^2}{ \lcm\left(\frac{7}{\gcd( [n]_2 ,7 )}, \frac{7}{\gcd( [n-\ell]_2,7)}  \right)}.
 $$
 In this case, we performed computations  using Magma software, where we calculated $\operatorname{gcd}([n]_2, 7)$ for $n$ ranging from $1$ to $100,000,000$. We observed two distinct cases: 
$$
\gcd([n]_2, 7) = \begin{cases}
 1 , & \text{if $\gcd(n,3)=1$ }\\
7, & \text{otherwise }\\
  \end{cases} 
$$     

 \begin{enumerate}

 \item \underline{\textbf{If $\gcd(n,3)= \gcd(n,\ell)=1$,}} then $\displaystyle{  \gcd([n]_2, 7)=  \gcd([n-\ell]_2, 7)= 1}$, and so the number of  $(n, \sigma)$-equivalence classes is $ 7$. So the possible trinomials are: 
     $$f_{i}(x)= x^{n} - x^{\ell} - \xi^{ i [n]_2 }  \ \ \text{for  }\  i\in \{0,\ldots, 6\}.  $$ 

    \item \underline{ \textbf{If $\gcd(n,3)= \gcd(n,\ell)=3$,}} then $\gcd([n]_2, 7)= \gcd([n-\ell]_2, 7) = 7$.  In this case the number of possible classes is $N=49,$ and the equivalence relation has no advantage in this case.
    \item \underline{\textbf{If  $\gcd(n,3)= 1 $ and $ \gcd(n,\ell)=3$,}}  then the number of equivalence classes is $N=7,$ and the possible trinomials are 
    $$   x^n-\xi^i x^{\ell}- 1,  \ \ \text{for  }\  i \in \{0,\ldots, 6\}. $$

    \item  \underline{\textbf{If  $  \gcd(n,3)= 3 $ and $ \gcd(n,\ell)=1$,}} then the number of equivalence classes is $N=7,$ and the possible trinomials are 
    $$x^n- x^{\ell}- \xi^i,  \ \ \text{for  }\  i \in \{0,\ldots, 6\}. $$

     \end{enumerate}
 \end{enumerate}
 \end{example}

 \begin{example} [Skew $(\ell,\sigma)$-trinomial  codes  over $\mathbb{F}_9$]
Let $q = 9$, and $\xi$ be a primitive element of $\mathbb{F}_9$. In this example, we are interested in skew $(\ell,\sigma)$-trinomial   codes of length $n$ over $\mathbb{F}_9$. Since the degree of the extension field $\mathbb{F}_9/\mathbb{F}_3$ is $2$, we have $\text{Aut}(\mathbb{F}_9) = \{\text{id}, \tau\}$, where $\tau$ is the Frobenius automorphism of $\mathbb{F}_9$, defined as $\tau(a) = a^3$ for all $a \in \mathbb{F}_9$. 

For $\sigma=\tau,$ we consider skew $(\ell,\sigma)$-trinomial  codes   associated with  a skew polynomial of the form $f(x) =x^{n}-a_{\ell}x^{\ell}-a_0\in \mathbb{F}_4[x,\sigma],$
 where $\ell $ is  an integer such that $0<\ell <n$.  Let $\vec{a}(x):=   a_{\ell}x^{\ell}+a_0 ,$  with Hamming weight 2.
 As $[i]_1= \dfrac{3^i-1}{3-1},$ by Theorem  \ref{Th.1}, the number of the Hamming $(n,\ell, \sigma)$-equivalence classes is given by  
 
 $$N:=\dfrac{8^2}{ \lcm\left(\dfrac{8}{\gcd( [n]_1 ,8 )}, \dfrac{8}{\gcd( [n-\ell]_1,8)}  \right)}    
 $$ 
  In this case, we performed computations  using Magma software, where we calculated $\operatorname{gcd}([n]_1, 8)$ for $n$ ranging from $1$ to $100,000,000$. We observed the following cases:
$$
\gcd([n]_1, 8) = \begin{cases}
 1 , & \text{if $n$ is odd.}\\
4, &  \text{if $n \equiv  2 \mod(4)$.}\\
8, &  \text{if $n \equiv  0 \mod(4)$.}\\
  \end{cases} 
$$ 

\begin{enumerate}
    \item  \underline{\textbf{If $n$ and $\ell$ are odd,}} then $\gcd \left( [n]_1 , 8\right) =\gcd \left( [n-\ell]_1 , 8\right)= 1$, and the number of $(n,\ell,\sigma)$-equivalence classes is $ N= 8.$  Then by Theorem \ref{Equiv_2}, we have 
     $$d_0=\gcd( [n]_1, 8)= 3,\ \ d_{\ell}= \gcd( 3^{\ell}[n-\ell]_1, 3)= 1 \text{ and} \ \  d=\gcd\left( \frac{8}{d_0}, \frac{8}{d_{\ell}}\right)=8.$$

     It follows that each  skew $(\ell,\sigma)$-trinomial  code associated with a skew  polynomial $f$ of the form  $f =x^{n}- a_{\ell} x^{\ell}-a_0$ is equivalent to a  skew $(\ell,\sigma)$-trinomial  code  associated with one of the following  skew polynomials: 
  $$f_{i}(x)= x^{n} - x^{\ell} - \xi^{ i [n]_1 }  \ \ \text{for}, i\in \{0,\ldots, 7\}.  $$

 \item  \underline{\textbf{If $n \equiv n-\ell \equiv 2 \pmod{4},$}} then $\gcd \left( [n]_1 , 8\right) =\gcd \left( [n-\ell]_1 , 8\right)= 4$, and the number of $(n,\ell,\sigma)$-equivalence classes is $ N= 16.$  Then by Theorem \ref{Equiv_2}, we have 
     $$d_0=\gcd( [n]_1, 8)= 4,\ \ d_{\ell}= \gcd( [n-\ell]_1, 3)= 4 \text{ and} \ \  d=\gcd\left( \frac{8}{d_0}, \frac{8}{d_{\ell}}\right)=4.$$

     It follows that each  skew $(\ell,\sigma)$-trinomial  code associated with a skew  polynomial $f$ of the form  $f =x^{n}- a_{\ell} x^{\ell}-a_0,$ is equivalent to a  skew $(\ell,\sigma)$-trinomial  code  associated with one of the following  skew polynomials: 
  $$f_{i}(x)= x^{n} - \xi^i x^{\ell} - \xi^{ j+h [n]_1 }  \ \ \text{for}, i,j,h\in \{0,\ldots, 3\}.  $$

\item \underline{\textbf{ If $n \equiv n-\ell \equiv 0 \pmod{4}$,}} then $\gcd \left( [n]_1 , 8\right) =\gcd \left( [n-\ell]_1 , 8\right)= 8$, and the number of $(n,\ell,\sigma)$-equivalence classes is $ N= 64.$ Hence  the equivalence relation has no advantage in this case.
\end{enumerate}
One can also discuss other cases like \textbf{ $ n$ is odd} and\textbf{ $ n-\ell\equiv 2 \mod(4)$} using a similar approach as above to determine the possible skew polynomials.
\end{example}

\subsection{Rank metric}

\begin{example}[Skew $(\ell,\sigma)$-trinomial codes over $ \mathbb{F}_{2^s}$]

Let $q = 2^s$, and $\xi$ be a primitive element of $\mathbb{F}_{2^s}$. In this example, we are interested in rank equivalence of  skew $(\ell,\sigma)$-trinomial codes of length $n$ over $\mathbb{F}_{2^s}$. The degree of the extension   $\mathbb{F}_{2^s}/\mathbb{F}_2$ is $s$, and  $\text{Aut}(\mathbb{F}_{2^s}) = \{\text{id}, \tau, \ldots, \tau^{s-1}\}$, where $\tau$ is the Frobenius automorphism of $\mathbb{F}_{2^s}$, i.e., $\tau(a) = a^2$ for all $a\in \mathbb{F}_{2^s}$. 

Let $\sigma=\tau^r,$ then the fixed subfield of $ \Fq$ by $ \sigma$ is $ \mathbb{F}_{2^{\gcd(r,s)}}.$
We consider skew $(\ell,\sigma)$-trinomial  codes   associated with  a skew polynomial of the form $f(x) =x^{n}-a_{\ell}x^{\ell}-a_0\in \mathbb{F}_{2^s}[x;\sigma],$ where $\ell$ is an integer such that $0<\ell <n$.  Let $\vec{a}(x):=   a_{\ell}x^{\ell}+a_0 .$ According to Corollary \ref{Th.1FqSigma}, the number of $( n,\ell, \Fq^{\sigma} )$-equivalence classes is given by  
 
 $$ N:=\dfrac{(2^s-1)^2}{ \lcm\left(\dfrac{2^{\gcd(r,s)}-1}{\gcd( n , 2^{\gcd(r,s)}-1)}, \dfrac{2^{\gcd(r,s)}-1}{\gcd( n-\ell, 2^{\gcd(r,s)}-1)}  \right)} . 
 $$ 
 Now we analyze the value of $N$  based on the parity of $n,\ell, r$ and $s.$
 \begin{enumerate}
\item \textbf{If $\gcd(r,s)=1$,} then $ N= (2^s-1)^2$ and each $( n,\ell, \Fq^{\sigma})$-equivalence  class contains one element.
\item \textbf{If $\gcd(r,s)= l\neq 1$ and $n,\ \ell$ are even,} then 

$$ N=\dfrac{(2^s-1)^2}{ \lcm\left(\dfrac{2^{\ell}-1}{1}, \dfrac{2^{\ell}-1}{1}  \right)}= \frac{(2^s-1)^2}{2^{\ell}-1}.$$
By Theorem \ref{RankClasses}, we have 
     $$d_0=\gcd( n, 2^{\ell}-1)= 1,\ \ d_{\ell}= \gcd( n-\ell, 2^{\ell}-1)= 1 \text{ and} \ \  d=\gcd\left( \frac{ 2^{\ell}-1}{d_0}, \frac{2^{\ell}-1}{d_{\ell}}\right)=2^{\ell}-1.$$
 It follows that each  skew $(\ell,\sigma)$-trinomial  code associated with a skew  polynomial $f$ of the form  $f =x^{n}- a_{\ell} x^{\ell}-a_0$ is rank equivalent to a skew $(\ell,\sigma)$-trinomial  code  associated with one of the following  skew polynomials: 
  $$f_i(x)= x^n- x^{\ell}- \xi^{\ell n'},$$
  for $\ell \in \{ 0,\ldots, d n'^2-1\},$ with $ n':= \frac{2^s-1}{2^l-1}.$

\item \textbf{If $\gcd(r,s)= \ell\neq 1$ and $n$ and $ \ell$ are  odd,} then
$$ N=\dfrac{(2^s-1)^2}{ \lcm\left(\dfrac{2^{l}-1}{ \gcd(n, 2^{l}-1) }, \dfrac{2^{l}-1}{1}  \right)}.$$
\begin{enumerate}
    \item If $ \gcd(n, 2^{\ell}-1)=1  $ then the number of $( n,\ell, \Fq^{\sigma})$-equivalence classes is $ \displaystyle{N =\frac{(2^s-1)^2}{2^{\ell}-1}.} $
By Theorem \ref{RankClasses}, we have 
     $$d_0=\gcd( n, 2^{\ell}-1)= 1,\ \ d_{\ell}= \gcd( n-\ell, 2^{\ell}-1)= 1 \text{ and} \ \  d=\gcd\left( \frac{ 2^{\ell}-1}{d_0}, \frac{2^{\ell}-1}{d_{\ell}}\right)=
     2^{\ell}-1.$$
 It follows that each  skew $(\ell,\sigma)$-trinomial  code associated with a skew  polynomial $f$ of the form  $f =x^{n}- a_{\ell} x^{\ell}-a_0$ is rank equivalent to a skew $(\ell,\sigma)$-trinomial  code  associated with one of the following  skew polynomials: 
  $$f_i(x)= x^n- x^{\ell}- \xi^{\ell n'},$$
  for $\ell \in \{ 0,\ldots, d n'^2-1\},$ with $ n':= \frac{2^s-1}{2^{\ell}-1}.$
  \item If $ 2^{\ell}-1  $ is a multiple of $n$, then the number of $( n,\ell, \Fq^{\sigma})-$equivalence classes is $ \displaystyle{N =\frac{(2^s-1)^2}{2^{\ell}-1}.} $
By Theorem \ref{RankClasses}, we have 
     $$d_0=\gcd( n, 2^{\ell}-1)= 2^{\ell}-1,\ \ d_{\ell}= \gcd( n-\ell, 2^{\ell}-1)= 1 \text{ and} \ \  d=\gcd\left( \frac{ 2^{\ell}-1}{d_0}, \frac{2^{\ell}-1}{d_{\ell}}\right)=1.$$
     
 It follows that each  skew $(\ell,\sigma)$-trinomial  code associated with a skew  polynomial $f$ of the form  $f =x^{n}- a_{\ell} x^{\ell}-a_0$ is rank equivalent to a skew $(\ell,\sigma)$-trinomial  code  associated with one of the following  skew polynomials: 
  $$f_i(x)= x^n- x^{\ell}- \xi^{i+\ell n'},$$
  for $i \in \{ 0,\ldots,d_0-1\} ;\quad \ell \in \{ 0,\ldots,n'^2-1\},$ with $ n':= \frac{2^s-1}{2^{\ell}-1}.$
\end{enumerate}

 \end{enumerate}
\end{example}

\section*{Conclusion}
In this paper, we investigated bounds on the (rank and Hamming) minimum distance of and equivalences of classes of skew polycyclic codes over a finite field $\mathbb{F}_q$, associated with a skew polynomial $f(x) \in \mathbb{F}_q[x;\sigma]$. First, we proved the Roos-like bound for both the Hamming and the rank metric for this class of codes. Our bound generalized many known bounds on various types of cyclic codes, both in the skew or commutative setting. Moreover, we analyzed constructions for such code with a designed distance, based on the concept of $\mu$-closed sets as an analog of cyclotomic cosets in the commutative case. 

Next, we focused on the equivalence of ambient spaces for skew polycyclic codes with respect to both metrics. We provided general conditions under which two families of skew polycyclic codes are equivalent. This was used to determine the exact number of equivalence classes. In the particular case of skew trinomial codes, we showed which code families are equivalent to the ones defined by the polynomial $x^n-x^\ell -1$, and more generally, how any equivalence class of a given skew trinomial can be determined. Finally, we gave some examples illustrating applications of the theoretical results on equivalence developed in this paper.

In future work, we will consider more general isometries to reduce the number of equivalence classes and use them for a more refined classification of polycyclic codes. Moreover, we plan to introduce the concept of ``skew cyclotomic cosets'' to facilitate the factorization of $x^n-1$ and $x^n-a$ and trinomial polynomials in $\mathbb{F}_q[x;\sigma]$, which may give an efficient way to classify skew polycyclic and skew constacyclic  codes of length $n$ over a finite field $\Fq$.

\printbibliography

\section{Appendix}
In this section, for clarity we provide some proofs regarding rank metric equivalence that were previously omitted due to their similarity to the Hamming case. 

\subsection{ Proof of Proposition \ref{ProofEquiv} }
\begin{proof}
The proof is the same in both cases; therefore, for simplicity, we denote $\cong_{(n,\sigma)}$
  without specifying the type of equivalence used.
\begin{enumerate}
\item  $f_1(x) \cong_{n,\sigma} f_1(x)  ,$ since 
the identity  automorphism of  $\mathbb{F}_q[x;\sigma] /\langle 
f_1(x)\rangle$  is  an $\Fq$-morphism isometry.
\item  If $ f_1(x) \cong_{n,\sigma} f_2(x) $ 
then there exists  $ \varphi_{\alpha}$ an $\Fq$-morphism isometry
 $$\varphi: \mathbb{F}_q[x,\sigma] /\langle f_2(x) \rangle \longrightarrow \mathbb{F}_q[x,\sigma] /\langle f_1(x) \rangle.$$
 
As $ \varphi$ preserves the Hamming (resp. rank) weight, then for all $f\in \mathbb{F}_q[x,\sigma] /\langle f_2(x) \rangle, $   $$ w(f)=w(\varphi(\varphi^{-1}(f)))=w(\varphi^{-1}(f)).$$
Hence, 
$$
\varphi^{-1}: \mathbb{F}_q[x,\sigma] /\langle f_1(x) \rangle \longrightarrow \mathbb{F}_q[x,\sigma] /\langle f_2(x) \rangle.
$$
is also an  $\Fq$-morphism isometry. Hence, $ f_2(x) \cong_{n,\sigma} f_1(x).$

\item Suppose that $ f_1(x) \cong_{n,\sigma} f_2(x) $  and $ f_2(x) \cong_{n,\sigma} f_3(x),$  then there exist two $\Fq$-morphism isometries
$$
\begin{aligned}
& \varphi_1: \mathbb{F}_q[x,\sigma] /\left\langle f_2(x) \right\rangle \longrightarrow \mathbb{F}_q[x,\sigma] /\langle 
f_1(x) \rangle \\
& \varphi_2: \mathbb{F}_q[x,\sigma] /\langle f_3(x)  \rangle \longrightarrow \mathbb{F}_q[x,\sigma] /\langle f_2(x)  \rangle .
\end{aligned}
$$
Therefore
$$
\varphi_{1}  \circ \varphi_{2} : \mathbb{F}_q[x,\sigma] /\langle f_3(x)  \longrightarrow 
\mathbb{F}_q[x,\sigma] /\langle f_1(x)  \rangle
$$
is an $\Fq$-morphism isometry. Hence, $ f_1(x) \cong_{n,\sigma} f_3(x). $
\end{enumerate}
 \end{proof}
 
 \subsection{Proof of Corollary \ref{Th.1FqSigma}}
 \begin{proof}\;
\begin{itemize}
    \item[(1) $\Rightarrow$ (2):] Immediately follows from the Definition \ref{DefFqSigma}.
 \item[(2) $\Rightarrow$ (3):] Suppose that we have an $\Fq$-morphism  isometry $\varphi_{\alpha}$ given by 
$\varphi_{\alpha}(x) = \alpha x$ for some $ \alpha \in (\mathbb{F}_q^{\sigma})^*$. It follows that

$$ 
\varphi_{\alpha}(x^i)= (\alpha x)^i= N_i^{\sigma}(\alpha) x^{ i} =\alpha^i x^i,$$
since  $\alpha \in (\mathbb{F}_q^{\sigma})^* $. Therefore, $ N_i^{\sigma}(\alpha)=\alpha^i\ \forall  \ i=0,1, \ldots, n-1.$

As $\varphi_{\alpha}$ is an $\Fq$-morphism isometry and  $ \varphi_{\alpha}(x^n-b_{\ell}x^{\ell}-b_0)=0 \ \text{mod}_r (x^n- a_{\ell} x^{\ell}-a_0) ,$ we have
\begin{equation}
\varphi( x^n)=   b_{\ell}  \alpha^{\ell} x^{\ell} +b_0.
 \end{equation}

On the other hand,
\begin{equation}
\varphi( x^n)=  N_n^{\sigma}(\alpha) x^n= \alpha^n ( a_{\ell}x^{\ell}+a_0) = \alpha^n  a_{\ell} x^{\ell}+ \alpha^n a_0.
\end{equation}

Comparing term by term, we obtain  that $ a_0 \alpha^n  = b_0   $ and $  a_{\ell} \alpha^n =  \alpha^{\ell} b_{\ell}. $ As $ \alpha \in (\mathbb{F}_q^{\sigma})^*$, we have 
$$ a_0 \alpha^n= a_0 N_n^{\sigma}(\alpha)   = b_0   \quad  \text{and} \quad   a_{\ell} \alpha^{n-\ell} = a_{\ell}N_{n-\ell}^{\sigma}(\alpha) =  b_{\ell}, $$
which means that $\alpha$ is a common right root of the polynomials $a_0 x^n - b_0$ and $a_{\ell} x^{n-\ell} - b_{\ell}$.

    \item[(3) $\Rightarrow$ (4)\;] and (4) $\Rightarrow$ (5) are immediate.
 \item[(5) $\Rightarrow$ (6):]
    Let $ \alpha \in (\mathbb{F}_q^{\sigma})^*$ be a right root of the polynomial $\gcrd(x^{n} - b_0 a_0^{-1}, x^{n-\ell} - b_{\ell} a_{\ell}^{-1})$. Then $\alpha$ is a common root of the polynomials $a_0 x^n - b_0$ and $a_{\ell} x^{n-\ell} - b_{\ell}$. It follows that  $a_i N^{\sigma}_{n-i}(\alpha)=a_i \alpha^{n-i} = b_i,\ \text{for any} \ i\in \{ 0,\ell\}$, i.e.,
$$  (b_0, b_{\ell} )= (\alpha^n a_0, \alpha^{n-\ell} a_{\ell})=(\alpha^n, \alpha^{n-\ell})\star(a_0, a_{\ell} ).$$ 

\item[(6) $\Rightarrow$ (7):] Suppose that

$$  (b_0, b_{\ell} )= \left( \alpha^n a_0, \alpha^{n-\ell} a_{\ell}\right)=\left(\alpha^n, \alpha^{n-\ell})\right)\star (a_0, a_{\ell}).$$ 
It follows that
$$  (a_0,a_{\ell})^{-1}\star (b_0, a_{\ell})=  \left( \alpha^n, \alpha^{n-\ell} \right)= \left( \zeta^{jn}, \zeta^{j(n-\ell)}\right)= 
\left( \zeta^{n}, \zeta^{(n-\ell)}\right)^j .$$
Then $ (a_0,a_{\ell})^{-1}\star (b_0, a_{\ell}) $ belongs to the cyclic subgroup  $H_{\ell,\Fq^{\sigma}}$ generated by $  
\left( \zeta^{n}, \zeta^{(n-\ell)}\right) $  as a subgroup of $ (\Fq^{\sigma})^{*}\times (\Fq^{\sigma})^{*}.$

\item[(7) $\Rightarrow$ (1):]
     Suppose that $ (a_0,a_{\ell})^{-1}\star (b_0, b_{\ell}) $ is an element of the cyclic subgroup  $H_{\ell,\sigma}$ generated by 
     $  (\zeta^n, \zeta^{n-\ell})  $  as a subgroup of $ (\Fq^{\sigma})^{*}\times (\Fq^{\sigma})^{*}.$ Then there exists an integer $ h$ such that 
     $$ (a_0,a_{\ell})^{-1}\star (b_0, b_{\ell})=   \left( \zeta^{nh}, \zeta^{(n-\ell)h}\right).$$
    For $ \alpha=  \zeta^{h}, $ we obtain  that $  a_i \alpha^{n-i}= b_i, $ for any $ i\in \{ 0,\ell\}.$
As in the proof of theorem \ref{Th.1}, we can prove that $ \varphi_{\alpha}$ given by
  \begin{equation}
        \begin{array}{cccc}
        \tilde{\varphi}_{\alpha} : & \mathbb{F}_q[x;\sigma]/\langle x^n - b_{\ell} x^{\ell} - b_0 \rangle,   & \longrightarrow & \mathbb{F}_q[x;\sigma] /\langle x^n - a_{\ell} x^{\ell} - a_0 \rangle, \\ 
        & f(x) & \longmapsto & f(\alpha x),
        \end{array}
    \end{equation} 
    is an $\Fq$-morphism isometry with respect to the rank metric, with particularity that $ \alpha \in (\mathbb{F}_q^{\sigma})^* $ and $ N_i^{\sigma}(\alpha)= \alpha^i.$ 
\end{itemize}
By the equivalence between (1) and (7), we deduce that the number of $(n,\ell,\Fq^{\sigma})$-equivalence classes on $\Fq^*\times \Fq^*$ corresponds to the order of the group $ \left( \Fq^*\times \Fq^* \right)/H_{\ell,\Fq^{\sigma}},$ which equals 
$$ N_{(n,\ell,\Fq^{\sigma})}=\dfrac{ (q-1)^2}{ \lcm(\frac{q_0-1}{ \gcd(n,q_0-1)}, \frac{q_0-1}{ \gcd(n-\ell,q_0-1)})} . $$
 \end{proof}

\subsection{Proof of Theorem  \ref{RankClasses}}
\begin{proof}\;
\begin{enumerate}
    \item  If $ d = 1 $, then the cyclic group $ H_{\ell,\sigma} $ generated by $ \left( \zeta^n, \zeta^{n-\ell}\right) $ is isomorphic to the group $ \langle  N_n^{\sigma}(\zeta) \rangle \times \langle N_{n-\ell}^{\sigma}(\sigma^{\ell}(\zeta))  \rangle $ and has order $ \frac{(q'-1)^2}{d_0 d_{\ell}} $. By Theorem \ref{Th.1}, the number of $ (n,\ell,\sigma, \mathbb{F}_{q'})$-equivalence classes is $ \left(\frac{q-1}{q'-1}\right)^2 d_0 d_{\ell}= n'^2 d_0 d_{\ell} $. 
    Therefore, we can partition $ \Fq^* \times \Fq^* $ as follows

$$
\Fq^* \times \Fq^* = \bigcup_{l =0}^{n'^2} \bigcup_{i=0}^{d} \bigcup_{j=0}^{d_{\ell}-1} ( \xi^{i+ln'}\xi^i, \xi^j ) H_{\ell, \mathbb{F}_{q'} }.
$$
Then any pair $(a_0,a_{\ell})$ is $( n,\ell,\sigma, \mathbb{F}_{q'})$-equivalent to one of the pairs   $ ( \xi^{i+ln' },  \xi^{j } ) , $ for  $i=0,1,\ldots, d_0-1$; \   $j=0,1,\ldots, d_{\ell}-1;  $ and $ l=0,\ldots, n'^2-1.$
 \item 
 If $d\neq 1, $ the number of  $(n,\ell,\sigma, \mathbb{F}_{q'})$-equivalence classes is $ n'^2 d d_0 d_{\ell} ,$ and so we partition $ \Fq^* \times \Fq^* $ as follows:
$$  \Fq^{*}\times \Fq^{*} = 
\bigcup_{h=0}^{d n'^2 -1}  \bigcup_{i=0}^{d_0-1} \bigcup_{j=0}^{d_{\ell}-1} (\xi^{i +hn'} ,\xi^{j}) H_{\ell,\mathbb{F}_{q'}}, $$
which implies the second statement, similarly to the first case.
\end{enumerate}
\end{proof}\;

\subsection{Proof of Proposition \ref{RankEquiv_Consta}}
\begin{proof}
\begin{enumerate}
\item  As $ (a_0, a_{\ell}) \sim_{(n,\ell,\sigma, \mathbb{F}_{q'})} (b_0, b_{\ell})$, by the second  statement of Theorem \ref{Th.1}, 
there exists $\alpha\in \Fq^*$ such that $$ a_0 N_n^{\sigma}(\alpha)  = b_0   \quad  \text{and} \quad   a_{\ell} N^{\sigma}_{n-\ell}(\sigma^{\ell}(\alpha))=  b_{\ell}. $$
This gives
$$ a_0^{-1}  b_0 = N_n^{\sigma}(\alpha)  \quad  \text{and} \quad   a_{\ell}^{-1}   b_{\ell}=  N^{\sigma}_{n-\ell}(\sigma^{\ell}(\alpha)) . $$
Which means that  
$$ a_0^{-1}  b_0 \in   \langle  N_n^{\sigma}(\xi) \rangle  \quad  \text{and} \quad   a_{\ell}^{-1}   b_{\ell} \in  
\langle   N^{\sigma}_{n-\ell}(\sigma^{\ell}(\xi))  \rangle. $$

\item Let   $i\in \{0,\ell\}$ then the order of $ \langle   N^{\sigma}_{n-i}(\sigma^{i}(\xi))  \rangle =
\langle   \xi^{p^{ri}[n]_r}  \rangle $ is  $$d_i=\dfrac{q-1}{\gcd(p^{ri}[n-i]_r, q-1)}= \dfrac{q-1}{\gcd([n-i]_r, q-1)} , \text{ since $ \gcd(p^{ri},q-1)=1.$ }$$  It follows that  $(b_i a_i^{-1})^{d_i} =1 ,\ i\in \{0,\ell\}.$
\item by the second  statement of Theorem \ref{Th.1}, 
there exists $\alpha\in \Fq^*$ such that $$ a_0 N_n^{\sigma}(\alpha)  = b_0   \quad  \text{and} \quad   a_{\ell} N^{\sigma}_{n-\ell}(\sigma^{\ell}(\alpha))=  b_{\ell}. $$  As in the proof of \cite[Theorem 6]{Ouazzou2025}. For $i\in \{0,\ell\},$ we verify
 that the  map  $\varphi_i$ given by 
\begin{equation}
  	\begin{array}{cccc}
  	\varphi_{i}:& \mathbb{F}_q[x] /\langle x^{n-i}-b_0 \rangle  &\longrightarrow & \mathbb{F}_q[x] /\langle x^{n-i}-a_0 \rangle, \\
  	& f(x) & \longmapsto &  f(\sigma^i(\alpha) x).
  	\end{array}
  	\end{equation} 
    is an $\Fq$-morphism isometry with respect to the Hamming  metric.
    \item Follows from the third statement.
\end{enumerate}
\end{proof}

\subsection{More details in relation to Remark \ref{ConstacyclicFq_sigma}}
We prove here an alternative version of Proposition \ref{RankEquiv_Consta}, when we calculate the rank over the fixed subfield.
 \begin{proposition}\label{PropFq_sigma}
Let $ (a_0, a_{\ell}) $ and $ (b_0, b_{\ell}) $ be elements of $ \Fq^{*} \times \Fq^{*} $ such that $ (a_0, a_{\ell}) \sim_{(n, \ell,\Fq^{\sigma)}} (b_0, b_{\ell}) $. Then:
\begin{enumerate}
    \item There exists $\alpha\in (\Fq^{\sigma})^{*} ,$ such that for each $ i \in \{0, \ell\} $, $ a_i^{-1} b_i \in \langle \alpha^{n-i} \rangle.$
    \item For each $ i \in \{0, \ell\} $, $ (a_i^{-1} b_i)^{d_i} = 1 $, where $ d_i := \dfrac{q_0-1}{\gcd(n-i, q_0-1)} $.
    \item For each $ i \in \{0, \ell\} $, the map 
    \begin{equation*}
  	\begin{array}{cccc}
  	\varphi_{i}:& \mathbb{F}_q[x;\sigma] /\langle x^{n-i}-b_i\rangle  &\longrightarrow & \mathbb{F}_q[x;\sigma] /\langle
    x^{n-i}-b_i\rangle, \\ 
  	& f(x)& \longmapsto &  f(\alpha x)  
  	\end{array}
  	\end{equation*} 
is an $\mathbb{F}_q$-morphism isometry  with respect to the rank  distance.
    
    \item For each $ i \in \{0, \ell\} $ the class of skew  $(a_i,\sigma)$-constacyclic codes of length $n-i$ is { rank equivalent} to the class of skew   $(b_i,\sigma)$-constacyclic codes of length $n-i$  over $\Fq.$ 
\end{enumerate}
\end{proposition}
\begin{proof}
\begin{enumerate}
\item  As $ (a_0, a_{\ell}) \sim_{(n, \ell,\Fq^{\sigma)}} (b_0, b_{\ell})$, by the third statement of Theorem \ref{Th.1FqSigma}, there
exists  $\alpha\in \Fq^{\sigma}$ which is a common right root of $\displaystyle a_i x^{n-i } - b_i \in \Fq[x;\sigma], i \in \{0,\ell\}. $ It follows that $$ a_0 \alpha^n= a_0 N_n^{\sigma}(\alpha)   = b_0   \quad  \text{and} \quad   a_{\ell} \alpha^{n-\ell} = a_{\ell}N_{n-\ell}^{\sigma}(\alpha) =  b_{\ell}. $$ 
This gives
$$ a_0^{-1}  b_0 = \alpha^n  \quad  \text{and} \quad   a_{\ell}^{-1}   b_{\ell}=  \alpha^{n-\ell} . $$
Which means that  
$$ a_0^{-1}  b_0 \in   \langle \alpha^{n} \rangle  \quad  \text{and} \quad   a_{\ell}^{-1}   b_{\ell} \in   \langle \alpha^{n-\ell} \rangle. $$

\item  Let $\zeta $ be a primitive element of $ \Fq^{\sigma}$ then $\alpha= \zeta^i\in \Fq^{\sigma}$ for some $ 0\leq i\leq q_0-1.$  
Then for each  $i\in \{0,\ell\}$ the order of $ \langle \alpha^{n-i} \rangle $ is  $d_i=\dfrac{q_0-1}{\gcd(n-i, q_0-1)}.$ It follows that  $(b_i a_i^{-1})^{d_i} =1 ,\ i\in \{0,\ell\}.$
\item By Theorem \ref{Th.1FqSigma}, there is  a common root $\alpha$  of $a_0  x^n-b_0$ and $a_1 x^{n-1}-b_1$. Then for 
 any  $i\in \{0,1\}$,  we have $ a_0  \alpha^{n-i}=b_0.$ As in the proof of \cite[Theorem 6]{Ouazzou2025}, we verify
 that the  map  $\varphi_i$ given by 
\begin{equation}
  	\begin{array}{cccc}
  	\varphi_{i}:& \mathbb{F}_q[x] /\langle x^{n-i}-b_0 \rangle  &\longrightarrow & \mathbb{F}_q[x] /\langle x^{n-i}-a_0 \rangle, \\
  	& f(x) & \longmapsto &  f(\alpha x).
  	\end{array}
  	\end{equation} 
    is an $\Fq$-morphism isometry with respect to the rank metric.
    \item Follows from the third statement.
\end{enumerate}
\end{proof}

\end{document}